\documentclass[11pt,letterpaper]{article}
\usepackage[margin=1in]{geometry}
\usepackage{latexsym,graphicx,amssymb}
\usepackage{amsmath}
\usepackage{amsthm}
\usepackage{float}
\usepackage{xspace}
\usepackage{paralist}
\usepackage{enumerate,multicol}
\usepackage{cases}
\usepackage{caption}
\usepackage{graphicx}
\usepackage{hyperref}
\usepackage{tikz,pgfmath}
\usepackage{fourier}
\usepackage{xurl}
\usepackage{cleveref}
\usetikzlibrary{matrix,positioning,quotes}
\usetikzlibrary{shapes.multipart}
\usetikzlibrary{calc}
\usetikzlibrary{arrows,decorations.markings}
\usetikzlibrary{matrix}
\usepackage[ruled,linesnumbered]{algorithm2e}

\usepackage[noend]{algpseudocode}
\usepackage{xcolor}
\usepackage{multirow}
\usepackage{multicol}
\usepackage{graphicx}
\usepackage{subcaption}
\usepackage{varwidth}
\usepackage{wrapfig}
\usepackage{bbm}
\usepackage{enumitem}
\usepackage[thinlines]{easytable}

\usepackage{{makecell}}
\setcellgapes{3pt}\makegapedcells

\newcolumntype{M}{>{\hfil$\displaystyle}X<{$\hfil}} \newcolumntype{L}{>{\collectcell\AddLabel}r<{\endcollectcell}}

\usepackage{booktabs}
\usepackage{tabularx}
\usepackage{siunitx}

\usepackage{{makecell}}

\usepackage{array,collcell}
\newcommand\AddLabel[1]{\refstepcounter{equation}(\theequation)\label{#1}}

\newtheorem{remark}{Remark}[section]
\usepackage[numbers]{natbib}

\usepackage{pifont}\usepackage{todonotes}

\DeclareMathOperator*{\argmax}{argmax}

\newcommand{\notshow}[1]{{}}

\newcommand{\Xcomment}[1]{{}}

\newcommand{\xmark}{\ding{55}}

\newcommand{\LF}{\textsc{Left}}
\newcommand{\RT}{\textsc{Right}}

\newcommand{\alg}{\textsf{ALG}}
\newcommand{\opt}{\textsf{OPT}}
\newcommand{\indic}{\mathbbm{1}}
\newcommand{\E}{\mathop{{}\mathbb{E}}}

\newcommand{\I}{\mathcal{I}}
\newcommand{\OPTW}{\text{OPT-}\mathcal{W}}
\newcommand{\ALGW}{\mathcal{W}}
\newcommand{\Mecha}{\mathcal{M}}
\newcommand{\ALG}{\mathrm{ALG}}
\newcommand{\OPT}{\mathrm{OPT}} 
\newcommand{\pdf}{f_N^i}
\newcommand{\emF}{\tilde{F}}

\newcommand{\given}{\;\vert\;}

\newcommand{\UpperRatio}{0.72}
\newcommand{\LowerRatio}{0.7381}
\newcommand{\defeq}{\mathrel{\overset{\makebox[0pt]{\mbox{\normalfont\tiny\sffamily def}}}{=}}}

\newcounter{inlineequation}
\renewcommand{\theinlineequation}{(\Roman{inlineequation})}

\newcommand{\inlineeq}[1]{\refstepcounter{inlineequation}\theinlineequation\ \(#1\)}

\newtheorem{definition}{Definition}[section]

\newtheorem{theorem}{Theorem}[section]
\newtheorem{lemma}{Lemma}[section]

\title{On the Optimal Fixed-Price Mechanism in Bilateral Trade}

\author{ Yang Cai\thanks{Yang Cai is supported by a Sloan Foundation Research Fellowship and the National Science Foundation Award CCF-1942583 (CAREER). Part of this work was done while the author was visiting the Simons Institute for the Theory of Computing.} \\Yale University, USA\\yang.cai@yale.edu\\
  \and Jinzhao Wu\\ Yale University, USA\\ jinzhao.wu@yale.edu
}
\begin{document}

\maketitle

\begin{abstract}
 We study the problem of social welfare maximization in bilateral trade, where two agents, a buyer and a seller, trade an indivisible item. The seminal result of Myerson and Satterthwaite~\cite{myerson_efficient_1983} shows that no incentive compatible and \emph{budget balanced} (i.e., the mechanism does not run a deficit) mechanism can achieve the optimal social welfare in bilateral trade. Motivated by this impossibility result, we focus on approximating the optimal social welfare. We consider arguably the simplest form of mechanisms -- the fixed-price mechanisms, where the designer offers trade at a fixed price to the seller and buyer. Besides the simple form, fixed-price mechanisms are also the only \emph{dominant strategy incentive compatible}  and \emph{budget balanced} mechanisms in bilateral trade~\cite{HAGERTY198794}.

We obtain improved approximation ratios of fixed-price mechanisms in both (i) the setting where the designer has the \emph{full prior information}, that is, the value distributions of both the seller and buyer; and (ii) the setting where the designer only has access to limited information of the prior. In the full prior information setting, we show that the optimal fixed-price mechanism can achieve at least $\UpperRatio$ of the optimal welfare, and no fixed-price mechanism can achieve more than $\LowerRatio$ of the optimal welfare. Prior to our result the state of the art approximation ratio was $1 - {1\over e} + 0.0001\approx 0.632$~\cite{kang2022fixed}. Interestingly, we further show that the optimal approximation ratio achievable with full prior information is \emph{identical} to the optimal approximation ratio obtainable with only \emph{one-sided prior information}, i.e., the buyer's or the seller's value distribution. As a simple corollary, our upper and lower bounds in the full prior information setting also apply to the one-sided prior information setting. 

We further consider two limited information settings. In the first one, the designer is only given the mean of the buyer's value (or the mean of the seller's value). We show that with such minimal information, one can already design a fixed-price mechanism that achieves $2/3$ of the optimal social welfare, which surpasses the previous state of the art ratio even when the designer has access to the full prior information. Furthermore, $2/3$ is the optimal attainable ratio in this setting. In the second limited information setting, we assume that the designer has access to finitely many samples from the value distributions. 
Recent results show that one can already obtain a constant factor approximation to the optimal welfare using a single sample from the seller's distribution~\cite{babaioff_bulow-klemperer-style_2020,dutting_efficient_2021, kang2022fixed}. Our goal is to understand what approximation ratios are possible if the designer has more than one but still finitely many samples. This is usually a technically more challenging regime and requires tools different from the single-sample analysis. We propose a new family of sample-based fixed-price mechanisms that we refer to as the \emph{order statistic mechanisms} and provide a complete characterization of their approximation ratios for any fixed number of samples. Using the characterization, we provide the optimal approximation ratios obtainable by order statistic mechanism for small sample sizes (no more than $10$ samples) and observe that they significantly outperform the single sample mechanism.

\end{abstract}
\thispagestyle{empty}
\addtocounter{page}{-1}
\newpage
\section{Introduction}\label{sec:Intro}
We study a fundamental problem in mechanism design -- maximizing social welfare in bilateral trade, in which two agents, a seller and a buyer, trade an indivisible item. More specifically, we consider the Bayesian setting where the seller's private value $S$ for the item that is drawn from distribution $F_S$, and the buyer's private value $B$ for the item is drawn from distribution $F_B$. The social welfare is therefore defined as ${\E_{B,S}\left[S+(B-S)\cdot x(B,S)\right]}$, where $x(B,S)$ denotes the probability that the trade happens when the seller's value is $S$ and the buyer's value is $B$.

Surprisingly, exactly maximizing the social welfare in bilateral trade is impossible. The seminal result by Myerson and Satterthwaite~\cite{myerson_efficient_1983} shows that no mechanism can \emph{simultaneously} be (i) incentive compatible (to the buyer and the seller), (ii) \emph{budget balanced}, i.e., the mechanism does not run a deficit, and (iii) maximizes the social welfare. For example, the VCG mechanism is incentive compatible and maximizes the social welfare but is not budget balanced in general.  Motivated by this impossibility result, our goal is to design incentive compatible and budget balanced mechanisms to approximate the optimal welfare. We focus on the fixed-price mechanisms, in which the designer offers trade at a fixed price to the seller and buyer. It is also known  that fixed-price mechanisms are the only \emph{dominant strategy} incentive compatible and budge balance mechanisms in bilateral trade~\cite{HAGERTY198794}.

\subsection{Our Contributions}\label{sec:contribution}
We make progress on this problem on multiple fronts. We first consider the \emph{full prior information} setting, where the designer knows both $F_S$ and $F_B$. We show how to use a \emph{factor revealing min-max program} to improve the approximation ratio achievable by a fixed-price mechanism.

\medskip\noindent
\hspace{.2in}\begin{minipage}{0.93\textwidth}\textbf{Contribution 1:} For any $F_B,F_S$, there exists a fixed-price mechanism whose welfare is at least $0.72\cdot \opt$, where $\opt=\E_{B,S} [\max\{B,S\}]$ is the optimal welfare. Moreover, there exists a $F_B$ and $F_S$ such that no fixed-price mechanism can attain welfare more than $0.7381\cdot \opt$. The formal statement of our result can be found in \Cref{thm:fullinfo}.
\end{minipage}
\medskip

We also have a ``constant time'' algorithm for computing the fixed-price mechanism that achieves the welfare guarantee above. More specifically, we construct a collection of numbers $p_1,\cdots, p_n$, so that for any $F_S$, $F_B$, our algorithm chooses the best price in the set ${\left \{p_1\cdot \opt,\cdots,p_n\cdot \opt\right \}}$. Clearly, the approximation ratio will be better when we increase $n$. We show that when $n=16$, our algorithm already computes a fixed-price mechanism that has welfare at least $0.72\cdot \opt$. 

Our result significantly improves on the state-of-the-art approximation $1-{1\over e} + 0.0001\approx 0.6322$~\cite{kang2022fixed}. Our new hardness result also strengthens the previous best bound of $0.7385$~\cite{kang2018strategy}.  Our upper and lower bounds are obtained by considering two discretized variants of an infinite-dimensional min-max optimization problem defined in \Cref{sec:characterization fullinfo}. We show in \Cref{lem:optconvergence} that, in the limit when the discretization accuracy approaches $0$, the upper bound and lower bound obtainable by our method will converge to the optimal approximation ratio. Of course, the factor-revealing program becomes more expensive to solve with finer discretization. Our upper and lower bounds are derived using the finest discretization that we can computationally solve, but one could further close the gap with more computational resources.

We next consider the case where the designer only has access to either the buyer's or the seller's value distribution. Clearly, the performance of the fixed price mechanism cannot be better than the case when the designer knows the full prior. Surprisingly, we show that the optimal fixed-price mechanism's performance in terms of the worst-case approximation ratio remains unaffected despite the absence of information from one side of the market.

\medskip\noindent
\hspace{.2in}\begin{minipage}{0.93\textwidth}\textbf{Contribution 2:} Given only one-sided prior information, i.e., $F_B$ or $F_S$, the optimal approximation ratio obtainable by a fixed-price mechanism is identical to the optimal approximation ratio obtainable given access to both $F_B$ and $F_S$. As a simple corollary, our upper and lower bounds in \Cref{thm:fullinfo} also apply to the case when the designer only has access to one-sided prior information. The formal statement is in \Cref{thm:one-sided}. 
\end{minipage}
\medskip

\paragraph{Fixed-price mechanism based on only $\E[B]$ or $\E[S]$.} Our first two results require the designer to know either both $F_B$ and $F_S$ \footnote{Our first result uses $F_B$ and $F_S$ in two places: (1) to compute $\opt$ and (2) to identify the best price in the set.} or at least one of the two distributions. However, information about the underlying distributions of the agents' values is often scarce in practice, thus it is more desirable to design approximately optimal mechanisms using only limited prior information. Our third contribution concerns the case where the designer does not have the full information of the underlying distributions but only knows the mean of $F_S$ or $F_B$.

\medskip\noindent
\hspace{.2in}\begin{minipage}{0.93\textwidth}\textbf{Contribution 3:} We provide the \emph{max-min optimal} fixed-price mechanism, when the designer only has access to $\E[B]$ or $\E[S]$. More specifically, given only $\E[B]$ (or $\E[S]$), we provide a \emph{closed-form} randomized fixed-price mechanism $\Mecha_B$ (or $\Mecha_S$) whose welfare is at least $\frac23\cdot \opt$ for any buyer value distribution with mean $\E[B]$ (or any seller value distribution with mean $\E[S]$) and seller value distribution $F_S$ (or any buyer value distribution $F_B$). Additionally, $\frac 23$ is the optimal approximation ratio obtainable by any fixed-price mechanism that only uses $\E[B]$ (or $\E[S]$). See \Cref{def:optimal mechanism using only mean} for details of $\Mecha_B$ and $\Mecha_S$, and \Cref{thm:PartialMecha} for the formal statement of our result.
\end{minipage}
\medskip

We would like to highlight that the ratio of $\frac23$  exceeds the previous state-of-the-art approximation ratio achievable in the full prior information setting. \cite{blumrosen2021almost,kang2022fixed} consider the setting where only $F_S$ is known to the designer and show that a \emph{quantile mechanism} (Mechanism~\ref{alg:quantile mech}), i.e., a fixed-price mechanism that chooses the trading price according to a distribution of quantiles of the seller's distribution, can obtain at least $1-{1\over e}\approx 0.6321$ fraction of the optimal welfare. 
\cite{kang2022fixed} further shows that no quantile mechanism can obtain more than $1-{1\over e}$ fraction of the optimal welfare in the worst case. This result is sometimes interpreted as saying no mechanism can obtain an approximation ratio better than $1-{1\over e}$ with only information about the seller's value distribution. Our second contribution (\Cref{thm:one-sided}) shows that there is a strictly better way to use the information about seller's value distribution, as the ratio should be at least $\UpperRatio$. Furthermore, our third contribution (\Cref{thm:PartialMecha}) shows that, with minimal information about $F_S$, i.e., its mean $\E[S]$, one can design a fixed-price mechanism that strictly outperforms the optimal quantile mechanism that requires the full knowledge of $F_S$. Moreover, the quantile mechanism is asymmetric and only defined when we know the seller's value distribution. We show in Theorem~\ref{thm:NonBuyerInfo} that this is unavoidable, as no quantile mechanism over buyer's value distribution can guarantee a constant fraction of the optimal welfare.\footnote{This asymmetry is due to the asymmetry of the initial allocation -- the item is owned by the seller.} In contrast, our third result holds when the designer only knows the mean of the buyer's value distribution $F_B$.

\begin{algorithm}[H]
\SetAlgorithmName{Mechanism}
~~Input: A distribution $Q$ over the interval $[0,1]$\;
Randomly choose a quantile $x\in [0,1]$ according to  $Q$\;
Output the $x$-quantile of the seller's distribution as the price. Let $F_S$ be the seller's distribution and $F_S^{-1}(\cdot)$ be seller's quantile function mapping any quantile to its corresponding value. The quantile mechanism  outputs $F_S^{-1}(x)$ as the price.
\caption{{\sf \quad Quantile mechanism.}}\label{alg:quantile mech}
\end{algorithm}

\vspace{-.1in}
\paragraph{Fixed-price mechanism using finitely many samples.} Finally, we consider a different limited information model and initiate the study of approximating the optimal social welfare in using \emph{finitely many samples}. Namely, we are given a finite and limited number of samples, e.g., $3$ or $5$ samples, and the goal is to design the best mechanism possible using these samples. It is important to distinguish this setting from the more standard \emph{large sample} setting, where the goal is to determine the number of samples needed to design a $(1-\varepsilon)$-optimal mechanism (or optimal within a certain mechanism class) grows as a function of $\frac{1}{\varepsilon}$ and other parameters of the mechanism design environment. The sample complexity in large sample settings is usually stated using the big-O notation and ignores the accompanying constant. As a result, these bounds are often vacuous when apply to the \emph{small sample} regime, where there are only a small finite number of samples available.

\medskip\noindent
\hspace{.2in}\begin{minipage}{0.93\textwidth}\textbf{Contribution 4:} We introduce a new family of mechanisms -- \emph{order statistic mechanisms} (Mechanism~\ref{alg:order statistic mech}) and provide an exact characterization of the optimal order statistic mechanisms for any fixed number of samples (\Cref{thm:symoptorderstatistic} and \Cref{thm:asymoptorderstatistic}). Using our characterization, we can compute the optimal approximation ratio obtainable for any sample size. 
\end{minipage}
\medskip

 Recent results show that one can already obtain a constant factor approximation to the optimal welfare using a single sample from the seller's distribution~\cite{babaioff_bulow-klemperer-style_2020,dutting_efficient_2021, kang2022fixed}. However, techniques from these papers are tailored to the single sample setting and are difficult to generalize to even the case when two samples are available. We provide a rich family of mechanisms that is well-defined for any number of samples and characterize their performance. Using the characterization, we manage to optimize within this family of mechanisms for any fixed number of samples.

\begin{algorithm}[H]
\SetAlgorithmName{Mechanism}
~~Input: A distribution $P$ over $[N]$\;
 Randomly choose a number $i \in [N]$ according to the distribution $P$\;
Given $N$ samples from the seller, select the $i$-th smallest sample as the price, which is the $i$-th order statistic of these samples.
\caption{{\sf \quad Order statistic mechanism with $N$ samples.}}\label{alg:order statistic mech}
\end{algorithm}

By numerically computing the optimal approximation ratios of order statistic mechanisms, we observe that the optimal order statistic mechanism with a small number of samples is usually sufficient to significantly boost the approximation ratio. For example, in the symmetric setting, i.e., $F_S=F_B$, \emph{five samples} is sufficient to obtain an approximation ratio that is within $1\%$ of the optimal ratio achievable by any fixed-price mechanism; in the asymmetric setting, i.e., $F_S\neq F_B$, the approximation ratio improves from $1/2$ to $0.578$ when the sample size increases from one to three. Another natural mechanism is the empirical risk minimization (ERM) mechanism, where one selects a price to maximize the social welfare w.r.t. the empirical distribution. We compare the performance of the optimal order statistic mechanism with ERM for sample size $N=2,3,5,10$ in the symmetric setting. In all cases, the order statistic mechanism substantially outperforms the ERM. See Table~\ref{table:symmetric} and~\ref{table:asymmetric} for our computed ratios in the symmetric and asymmetric cases respectively.

Our analysis of the order statistic mechanisms builds on an interesting connection between the order statistic mechanisms and the quantile mechanisms, that is, any order statistic mechanism is also a quantile mechanism. Note that the $i$-th order statistic over $N$ samples drawn uniformly and independently from $[0, 1]$ has density $\pdf(x) = N \binom{N-1}{i-1} \cdot x ^{i-1}\cdot(1-x)^{N-i}$. Suppose we use the $i$-th order statistic as the price, then it is equivalent to the quantile mechanism who selects a quantile corresponding to the density function $\pdf(\cdot)$. More generally, if we choose the $i$-th order statistic with probability $P_i$, then the order statistic mechanism is equivalent to the quantile mechanism that chooses the quantile according to the density function $q(\cdot) = \sum_{i=1}^N P_i \pdf(\cdot)$. With this connection, we can focus on quantile mechanisms, and we characterize the approximation ratio of any quantile mechanism as the solution of a minimization problem (Lemma~\ref{lem:asymratio}).  By applying this characterization for quantile mechanisms to order statistic mechanisms, we show that for any fixed sample size $N$, the ratio of the optimal order statistic mechanism is exactly the solution of a max-min optimization problem. 
Although the optimization problem seems intractable in general, we manage to solve it with sufficient numerical accuracy for $N\leq 10$. We only study approximating social welfare in bilateral trade in this paper, but we believe this perspective of viewing sample-based mechanisms through the lens of quantile mechanisms is novel and has broader applications, especially in the small sample regime where the designer only has access to finitely many samples.

\subsection{Related Work}
\paragraph{Gains from Trade Maximization in Two-Sided Markets.} Another important objective in two-sided markets is the gains from trade (GFT), which measures the increment of the welfare after the trade. Note that~\cite{myerson_efficient_1983} also implies that optimal GFT is not achievable in bilateral trade. There has been increasing interest from the algorithmic mechanism design community to study the approximability of the optimal GFT~\citep{blumrosen_approximating_2016,brustle_approximating_2017,colini-baldeschi_fixed_2017,babaioff_best_2018,babaioff_bulow-klemperer-style_2020,cai_multi-dimensional_2021,deng2021approximately}. It will be interesting to study the optimal approximation ratio obtainable for GFT maximization in both the full information and the limited information settings.

\vspace{-.15in}
\paragraph{Sample-Based Mechanism Design.} Sample-based mechanism design has become a central topic in algorithmic mechanism design as it provides an alternative model that weakens the classical but sometimes unrealistic Bayesian assumption. The results in this direction can be roughly partition into two groups: (1) Large sample results, where the goal is to determine the number of samples needed to design a $1-\varepsilon$-optimal mechanism (or optimal in a certain mechanism class) as a function of $\frac{1}{\varepsilon}$ and other parameters of the mechanism design environment, e.g., ~\citep{elkind_designing_2007,cole_sample_2014,mohri_learning_2014, guo_settling_2019, syrgkanis_sample_2017,morgenstern_pseudo-dimension_2015,morgenstern_learning_2016,cai_learning_2017,brustle_multi-item_2020} or (2) Single sample results, where the goal is to determine the optimal approximation ratio obtainable using a single sample, e.g.,\citep{fu_randomization_2015,dhangwatnotai_revenue_2010,goldner_prior-independent_2016,kang2022fixed,gonczarowski_sample_2018,dutting_efficient_2021}. Our result does not fit in to either of the groups. In particular, we study the regime where the designer has a small fixed number of samples, as a result, the machinery developed for large number of samples or a single sample does not apply to our setting. A recent line of works focus on the same regime as ours but for the monopolist pricing problem~\citep{babaioff_are_2018,daskalakis2020more,allouah_revenue_2021}. Due to the different nature of the studied problems, their techniques also do not apply here. \section{Preliminaries}
\label{sec:Prelim}

\paragraph{Bilateral Trade.} We study the bilateral trade problem. In this setting, there are two agents, a buyer and a seller, trade a single indivisible item. The seller owns the item and values it at $S$ while the buyer values the item at $B$. Both $S$ and $B$ are non-negative and unknown to us but they are respectively drawn from distributions $F_S$ and $F_B$ independently. We assume that $F_S$ and $F_B$ are continous distributions. Actually, such assumption is w.l.o.g. and we discuss the reduction from distributions with point masses to continous ones in Appendix~\ref{appendix:tiebreaking}.
\vspace{-.15in}
\paragraph{Fixed-price Mechanism.} We consider fixed-price mechanisms, which offer a price $p$ to trade the item. The trade happens if and only if both the seller and the buyer accept the price, i.e., $B \ge p \ge S$. As shown by \cite{HAGERTY198794}, fixed-price mechanism is the only dominant-strategy incentive-compatibility mechanism. In this paper, we consider (possibly randomized) fixed-price mechanisms. We abuse notation and use $\Mecha(F_S, F_B)$ or $\Mecha(\I)$ where $\I = (F_S, F_B)$ to denote the distribution of prices $p$ selected by mechanism $\Mecha$ on instance $\I = (F_S, F_B)$.

\vspace{-.15in}
\paragraph{Welfare and Approximation Ratio.} We consider the objective of social welfare in this paper. For an instance $\I = (F_S, F_B)$, the optimal welfare is defined as:
\begin{align*}
\OPTW(\I) = \E_{S\sim F_S, B\sim F_B}\left[\max(S, B)\right]
\end{align*}

Similarly, for a fixed-price mechanism $\Mecha$, the expected welfare on instance $\I$ can be written as:
\begin{align*}
\ALGW(\Mecha, \I) = \E_{S\sim F_{S}, B\sim F_{B} \atop p\sim \Mecha(F_S, F_B)} \left[S + \indic[S \le p \le B]\cdot (B - S) \right]	
\end{align*}
Specifically, we use $\ALG\left(q, \I\right)$ to denote the expected welfare when using a fixed price $q$.

Our goal is to maximize the approximation ratio. That is, find some mechanism $\Mecha$ maximize the following ratio.
\begin{align*}
	\min_{\I = (F_S, F_B)} \frac{\ALGW(\Mecha, \I) }{\OPTW(\I)}
\end{align*}

\vspace{-.15in}
\paragraph{Quantile Function.} Suppose $F(\cdot)$ is the c.d.f. of a distribution, and we define $F^{-1}(\cdot)$ as the quantile function mapping the quantile to its corresponding value in this distribution. That is, ${F^{-1}(x) = \inf\{y\given F(y) = x\}}$. 

\section{A Near-Optimal Mechanism in the Full Prior Information Setting}
\label{sec:fullinfo}	

In this section, we show a near-optimal fixed-price mechanism when given the full prior information of the buyer and the seller.

\begin{theorem}
\label{thm:fullinfo}
There exists a DSIC, individually rational, budget balanced mechanism that achieves at least \UpperRatio~fraction of the optimal welfare for any instance $\I = (F_S, F_B)$. Moreover, no such mechanism has an approximation ratio better than \LowerRatio.
\end{theorem}

To prove this, we first identify the best fixed-price mechanism when given the instance $\I = (F_S, F_B)$. Then, the approximation ratio is determined by the mechanism's performance on the worst-case instance. Such a worst-case instance could be characterized by an infinite dimensional quadratically constrained quadratic program (QCQP).
However, the infinite dimensional program is hard to solve directly. Instead, we use two finite programs that can be solved numerically to upper bound and lower bound the infinite dimensional program. Additionally, we show that the optimal solutions of these two programs converge to the optimal solution of the infinite dimensional program as the number of variables tends to infinity. 

\subsection{Characterizing the Optimal Mechanism}\label{sec:characterization fullinfo}

We first characterize the optimal fixed-price mechanism via an infinite dimensional QCQP. Given any instance $\I = (F_S, F_B)$, we could assume that $\OPTW(\I) = 1$ without loss of generality since we can always scale the instance so that this is true. The optimal fixed-price mechanism corresponds to choosing a price $p \in \arg\max_p \ALGW(\I, p)$. The following program captures the worst-case instance for fixed-price mechanisms.
\begin{center}
\begin{tabular}{|c|}
\hline 
The Optimization Problem $\mathsf{FullOp}$ \\
\hline
\parbox{15cm}{
\begin{align}
\min_{\mu,\nu, r}\quad r \nonumber \\
\textsf{s.t.} \quad  & \mu,\nu \text{ are probability measures defined on $\mathbb{R}_{\geq 0}$}\\
& \OPTW(\I) \defeq \int_{\mathbb{R}_{\geq 0}} \int_{\mathbb{R}_{\geq 0}} \max(x,y)\, \nu(d y) \, \mu (d x)  \geq 1& \nonumber    \\
& \ALGW(\I, t) \defeq \int_{\mathbb{R}_{\geq 0}} x\, \mu(dx)+  \int_{\mathbb{R}_{\geq 0}} \int_{\mathbb{R}_{\geq 0}} (y-x)\cdot \indic[x\leq t\leq y]\, \nu(dy)\, \mu(dx) \leq r & \forall t \geq 0 \nonumber
\end{align}
}\\
\hline 
\end{tabular}
\end{center}
\begin{lemma}
\label{lem:fullinfo}
The value of the optimal solution of $\mathsf{FullOp}$ is the \emph{tight} worst-case approximation ratio achievable by a fixed-price mechanism.
\end{lemma}

The proof of Lemma~\ref{lem:fullinfo} is postponed to Appendix~\ref{appendix:prooffullinfo}. Since it is difficult to directly solve an infinite dimensional program like $\mathsf{FullOp}$, we approximate $\mathsf{FullOp}$ from both above and below by constructing two families of finite programs which provide an upper bound and a lower bound respectively. 

\subsection{Factor Revealing Program for the Approximation Ratio under Full Prior Information}

We show that the approximation ratio of the optimal fixed price mechanism is at least $\UpperRatio$, which significantly improves the previous state of the art bound of $1-1/e+\varepsilon$ with $\varepsilon\approx 10^{-4}$.  Our approach is to find a fixed-price mechanism whose performance under the worst distribution is maximized. This is exactly captured by the optimization problem $\mathsf{FullOp}$. However, it is an infinite-dimensional program. In this section, we consider a discretized version of $\mathsf{FullOp}$. {More specifically, we assume that $\OPTW(\I) = 1$,  and we restrict the mechanism to only choose price from a finite set $P=\{p_1,p_2,\ldots, p_n\}$. What we manage to show is that the optimal value of the optimization problem $\mathsf{LowerOp}$ is indeed a lower bound on the maximum approximation ratio one can obtain using prices from $P$ for instance $\I$. We establish the following two crucial properties: (i) For any $\I=(F_S,F_B)$ satisfying $\OPTW(\I) = 1$, we can carefully round $F_S$ and $F_B$ to two discrete distributions supported on $P$, where $\{s_1,\ldots, s_n\}$ and $\{b_1,\ldots, b_n\}$ can be viewed as the corresponding ``probability mass function'' for the discretized distributions of the seller and the buyer.\footnote{For technical reasons, $\{s_1,\ldots, s_n\}$ and $\{b_1,\ldots, b_n\}$ do not exactly correspond to probability mass functions, but viewing them as the probability mass functions gives the right intuition.} Importantly, $\{s_1,\ldots, s_n\}$ and $\{b_1,\ldots, b_n\}$ satisfy inequalities~\eqref{eq:LowerOp1} -~\eqref{eq:LowerOp4}. (ii) For any price $p_t$, the welfare from the corresponding fixed-price mechanism under $\I$ is at least the welfare under the rounded distributions $\sum_{i=1}^n s_i p_i+\sum_{i=1}^{t - 1} \sum_{j=t+1}^{n} s_i b_j (p_j - p_i)$. Therefore, if we choose $r$ to be $\max_{t\in [n]} \ALGW(\I, p_t)$, $\{s_1,\ldots, s_n\}$,  $\{b_1,\ldots, b_n\}$, and $r$ form a feasible solution of $\mathsf{LowerOp}$, which implies that the optimal value of $\mathsf{LowerOp}$ is no greater than the constructed $r$. As the rounded distribution needs to satisfy a sequence of constraints (especially constraint~\eqref{eq:LowerOp4}), the procedure we use to round $F_S$ and $F_B$ is subtle and does not simply round things up or down. See Appendix~\ref{appendix:prooffullinfolower} for details.}

\begin{center}
\begin{tabular}{|c|}
\hline 
The Optimization Problem $\mathsf{LowerOp}$ \\
\hline
\parbox{15cm}{
\begin{align}
\min_{s_1,s_2\cdots, s_n\atop b_1,b_2,\cdots,b_n, r} \quad r \nonumber \\
\textsf{s.t.} \quad  & s_i, b_i \geq 0 & \forall i \in [n]\label{eq:LowerOp1}\\
& \sum_{i=1}^n s_i \geq 1 \quad \text{ {and} } \quad  \sum_{i=1}^n b_i \geq 1 \label{eq:LowerOp2}\\
& \sum_{i=1}^n s_i \leq 1 + \frac{1}{p_n} \quad \text{ {and} } \quad  \sum_{i=1}^n b_i \leq 1 + \frac{1}{p_n} & \label{eq:LowerOp3}    \\
& \sum_{i=1}^n \sum_{j=1}^n s_i b_j \max(p_i,p_j) \geq 1 & \label{eq:LowerOp4}    \\
& \sum_{i=1}^n s_i p_i+\sum_{i=1}^{t - 1} \sum_{j=t+1}^{n} s_i b_j (p_j - p_i) \leq r & \forall t \in [n] \label{eq:LowerOp5}
\end{align}
}\\
\hline 
\end{tabular}
\end{center}
\begin{lemma}
\label{lem:fullinfolower}
For any $0 = p_1 < p_2 < \cdots < p_n$, let $r^{*}$ be the optimal value of~$\mathsf{LowerOp}$. Suppose $M$ is the mechanism that chooses the best price from the set $\big\{p_1\cdot \E[\max(S, B)],p_2\cdot \E[\max(S,B)],\cdots, p_n\cdot \E[\max(S, B)]\big\}$ to maximize the welfare. The welfare obtained by $M$ is at least $r^{*}\cdot \opt$.
\end{lemma}

We defer the proof of the lemma to Appendix~\ref{appendix:prooffullinfolower}.

\subsection{Hardness Result under Full Prior Information}

In this section, our goal is to find a threshold and an instance such that no fixed-price mechanism has an approximation ratio better than the threshold on this instance. We focus on discrete distributions and consider an instance $\I = (F_S, F_B)$ where $F_S$ is a discrete distribution supported on $\{p_1+\varepsilon,p_2+\varepsilon,\cdots,p_n+\varepsilon\}$, and  $F_B$ is a discrete distributions supported on $P=\{p_1,p_2,\cdots, p_n\}$ where $\varepsilon > 0$ is a small enough constant. For such instance, the optimal price must also lie in the set $\{p_i + \varepsilon\}_{i\in [n]}$, as choosing a price $x$ where $p_{i} + \varepsilon \le x < p_{i+1} + \varepsilon$ is equivalent to choosing a price of $p_{i} + \varepsilon$ . Therefore, any valid solution for the optimization problem below corresponds to a hard instance.

\begin{lemma}\label{lem:fullinfoupper}
	For any valid solution $(s_1,s_2,\cdots, s_n, b_1,b_2,\cdots, b_n, r)$ of~$\mathsf{UpperOp}$ (defined in \Cref{appendix:prooffullinfoupper}) satisfying $r = \max_{t\in [n]} \sum_{i=1}^n s_i p_i + \sum_{i=1}^t\sum_{j=t+1}^n s_ib_j(p_j - p_i)$ and $\varepsilon > 0$, there exists an instance $\I = (F_S, F_B)$ such that no fixed-price mechanism can achieve more than $\left(r+\varepsilon\right)$-fraction of the optimal welfare on this instance.
\end{lemma}

The proof of Lemma~\ref{lem:fullinfoupper} is deferred to Appendix~\ref{appendix:prooffullinfoupper}.

\begin{proof}[Proof of Theorem~\ref{thm:fullinfo}]
	
	With Lemma~\ref{lem:fullinfolower} and Lemma~\ref{lem:fullinfoupper}, we are now ready to prove Theorem~\ref{thm:fullinfo}. For the numerical results, our anonymous \href{https://github.com/BilateralTradeAnonymous/On-the-Optimal-Fixed-Price-Mechanism-in-Bilateral-Trade}{GitHub repository}(\texttt{https://github.com/BilateralTradeAnonymou\\s/On-the-Optimal-Fixed-Price-Mechanism-in-Bilateral-Trade}) provides all the certificates and codes and also carefully explains all the details.
		
	For the lower bound, we choose $n = 16$. Using Gurobi \cite{gurobi}, we obtain a lower bound of \UpperRatio~for the optimization problem $\mathsf{LowerOp}$ for a carefully chosen set of price $\{p_1,p_2,\cdots, p_n\}$.\footnote{We choose $\{p_1,p_2,\cdots, p_{16}\}$ to be $\{0.0, 0.1, 0.19, 0.27, 0.315, 0.355, 0.395, 0.44, 0.485, 0.535, 0.595, 0.665, 0.74, 0.875, 1.195, 1000.0\}$ to derive the $\UpperRatio$. These numbers are chosen heuristically to provide good coverage between $0.3$ to $0.5$, which is the region with concentration of probability mass in some bad instances we encounter.}  Therefore, by Lemma~\ref{lem:fullinfolower}, there exists a \UpperRatio-approximate fixed-price mechanism.
	
	Things become much easier for the upper bound since we only need to find a feasible solution instead of proving a lower bound of the optimal value. We choose $n = 100$ and numerically solve $\mathsf{UpperOp}$ with a specific support $\{p_1,p_2,\cdots, p_n\}$ and find a feasible solution that satisfies the constraints in Lemma~\ref{lem:fullinfoupper} where $r \leq \LowerRatio$. Together with Lemma~\ref{lem:fullinfoupper}, we then find a hard instance such that no fixed-price mechanism attains a \LowerRatio-approximation of the optimal welfare.  Please check our \href{https://github.com/BilateralTradeAnonymous/On-the-Optimal-Fixed-Price-Mechanism-in-Bilateral-Trade}{GitHub repository} for the detailed specification of the distributions.

\end{proof}

Finally, we would like to point out that the optimal value obtained by $\mathsf{LowerOp}$ and $\mathsf{UpperOp}$ will converge to the optimal value as the discretization accuracy tends to $0$. 

\begin{lemma}
\label{lem:optconvergence}
Let $r^*$ be the optimal value of $\mathsf{FullOp}$, i.e. the optimal approximation ratio. For any $\varepsilon>0$, there exists two sets numbers $0=p_1 < p_2 < \cdots < p_n$ and $\{p_1',p_2',\cdots, p_{n'}'\}$ such that the optimal value of $\mathsf{LowerOp}$ with respect to $\{p_1,p_2,\cdots, p_n\}$ is at least $r^{*} - \varepsilon$ and the optimal value of $\mathsf{UpperOp}$ w.r.t. $\{p_1',p_2',\cdots,p_{n'}'\}$ is at most $r^{*} + \varepsilon$.
\end{lemma}

The proof of Lemma~\ref{lem:optconvergence} is deferred to Appendix~\ref{appendix:proofconvergence}

\section{One-Sided Prior Information}
\label{subsec:one-sided}
We discuss the setting where we only have access to either the seller's or buyer's distribution in this section. We show that the optimal approximation ratio achievable with full prior information is identical to the optimal approximation ratio obtainable when only one-sided prior information is known.

\begin{theorem}\label{thm:one-sided}
Let $r^{*}$ be the approximation ratio of the optimal fixed-price mechanism in the full prior information setting. The optimal mechanism with only access to the value distribution of the buyer (or the seller) has exactly the same approximation ratio $r^{*}$.
\end{theorem}

\begin{remark} Note that \Cref{thm:one-sided} does not imply that the optimal mechanism in the full prior information setting is also an optimal mechanism in the setting where only one-sided prior information is known.  Indeed, 
the mechanism proposed in \Cref{thm:fullinfo} that achieves the ratio $\UpperRatio$ relies on $\OPT$, which can only be computed with full prior information. \Cref{thm:one-sided} simply states that the optimal approximation ratio is identical in the two settings, but we do not provide explicit approximately-optimal mechanisms that use only one-sided prior information. 
\end{remark}
To provide a high-level sketch of the proof, we consider the case in which we only have information of the buyer's distribution as an example. Suppose our goal is only to achieve an approximation ratio of at least $r$. This requires us to find a distribution of price ${\omega}$ which only depends on $F_B$ so that for any seller's distribution $F_S$,  the term $\E_{q\sim \omega}[\ALG(q, \I) - r\cdot \OPT(\I)]$ is always non-negative for $\I = (F_B, F_S)$. Therefore, the approximation ratio is at least $r$ if and only if the optimum of the following optimization problem is non-negative:
\[\min_{F_B}\ \max_{\omega}\ \min_{F_S}\E_{q\sim \omega}\left[\ALG(q, \I) - r\cdot \OPT(\I)\right].\]

Observe that, for any fixed $F_B$, $\E_{q\sim \omega}\left[\ALG(q, \I) - r\cdot \OPT(\I)\right]$ is bilinear with respect to the probability density function of $\omega$ and $F_S$. This allows us to apply minimax theorem to swap the order of $\max_\omega$ and $\min_{F_S}$:
\begin{align*}
    \min_{F_B}\ \max_{\omega}\ \min_{F_S}\E_{q\sim \omega}\left[\ALG(q, \I) - r\cdot \OPT(\I)\right] &=\min_{F_B}\min_{F_S} \max_{\omega} \E_{q\sim \omega}\left[\ALG(q, \I) - r\cdot \OPT(\I)\right]\\
    & = \min_{\I = (F_S,F_B)} \max_{\omega} \E_{q\sim \omega}\left[\ALG(q, \I) - r\cdot \OPT(\I)\right].
\end{align*} Note that the term $\min_{\I = (F_S,F_B)} \max_{\omega} \E_{q\sim \omega}\left[\ALG(q, \I) - r\cdot \OPT(\I)\right]$ is non-negative if and only if the optimal mechanism given the full prior information has an approximation ratio of at least $r$. This equality indicates that for any $r > 0$, there exists a fixed-price mechanism using only buyer's distribution that can attain an $r$ fraction of the optimal welfare if and only if the approximation ratio of the optimal fixed-price mechanism with full prior information is at least $r$. In other words, the full prior information setting and the one-sided prior information setting have the same approximation ratio. The formal proof is more complicated as the ``variables'' of our optimization problem are infinite-dimensional.
The full proof of Theorem~\ref{thm:one-sided} is postponed to Appendix~\ref{subsec:one-sided-appendix}. 
\section{Breaking $1 - 1/e$ with Only $\E[B]$ or $\E[S]$}
\label{sec:partial}

We consider a limited information setting in which only the mean of either the seller's value or the buyer's value is known. \cite{kang2022fixed} shows that any mechanism that only uses quantile information from the seller can not achieve a ratio better than $1 - 1/e$. However, we observe that with minimal information of $F_S$ (or $F_B$), i.e.,  $\E[S]$ (or similarly $\E[B]$), we can break the $1 - 1/e$ barrier.

\begin{definition}\label{def:optimal mechanism using only mean}
We define the following two fixed-price mechanisms using only $\E[S]$ or $\E[B]$.
\begin{itemize}
    \item[] $\Mecha_S:$ Given $\E[S]$, the mechanism randomly picks a number $x\sim U[0,3]$, and sets the price as $x\cdot \E[S]$.
    \item[] $\Mecha_B:$ Given $\E[B]$, the mechanism randomly picks a number $x\sim P_B$, and sets the price as $x\cdot \E[B]$, where $P_B$ is a distribution over the interval $[0, 2]$ with the following cdf $F_B(x)$:
\begin{align*}
\begin{split}
F_B(x)=\left\{
\begin{aligned}
&\frac{x}{3-3x}  \quad & 0 \leq x\leq \frac12&\\
&(4x - 1) / 3  \quad & \frac12 <  x\leq \frac23\\
&(x + 1) / 3  \quad & \frac23 < x\ \leq 2\\
\end{aligned}
\right.
\end{split}
\end{align*}
\end{itemize}
\end{definition}

Both mechanisms attain at least $2/3$ of the optimal social welfare for any possible distribution of the other side. Additionally, $2/3$ is the best possible approximation ratio when only $\E[B]$ (or $\E[S]$) is known, indicating the optimality of $\Mecha_S$ and $\Mecha_B$.

\begin{theorem}
\label{thm:PartialMecha}
Given only $\E[S]$ (or $\E[B])$, $\Mecha_S$ (or $\Mecha_B$) obtains $\frac23$ of the optimal welfare. Furthermore, $\frac23$ is the optimal approximation ratio achievable using only $\E[S]$ (or $\E[B])$. 
\end{theorem}

Let us examine the setting where only $\E[S]$ is known to offer some insights. Since the mean of the seller is given, we can assume, without loss of generality, that $\E[S] = 1$ by applying appropriate scaling. To check  whether the approximation ratio of $\Mecha_S$ is $2/3$, it suffices to verify that \begin{equation}\label{eq:UB using seller mean}
    \min_{\I = (F_S,F_B)\atop \E[S] = 1}\E_{q\sim \Mecha_{S}}\left[\ALG(q, \I) - {2\over 3}\cdot \OPT(\I)\right]
\end{equation} is non-negative. Similar to the one-sided prior information setting, this term is bilinear w.r.t. the probability density function of $F_S$ and $F_B$. As a result,  we can argue that \emph{one of} the ``worst-case'' instances must have the following simple form: (i) the buyer's distribution is concentrated at a single point; (ii) the seller's distribution is supported on only two distinct points. Equipped with this observation, we can simplify the minimization problem in \Cref{eq:UB using seller mean} and certify the non-negativity of its minimum.

To demonstrate that $2/3$ is indeed the optimal ratio, we construct two instances with the same $\E[S]$ (or $\E[B]$). As only $\E[S]$ (or $\E[B]$) is known, the price for these two instances must be chosen from an identical distribution. We then show that, for any distribution of prices, the worse of the two instances must have an approximation ratio no greater than $2/3$. The complete proof of Theorem~\ref{thm:PartialMecha} is in Appendix~\ref{appendix:partial}.

\section{Fixed-Price Mechanism with Different Numbers of Samples}
\label{sec:Sample}

In this section, we consider the limited information setting where we only have sample access to the distributions. We focus on order statistic mechanisms which is defined in \Cref{sec:contribution} and our results cover different number of samples for both symmetric and general instances. In the small sample regime, we are able to characterize the optimal order statistic mechanism with any fixed number of samples. When the number of samples goes to infinity, we show that the optimal quantile mechanism can be approximated by order statistics mechanism as closely as desired and also obtain an upper bound on the sample complexity. Finally, recall that we assume the distributions for the seller and the buyer are continuous. See~\Cref{appendix:tiebreaking} for details.

\subsection{Order Statistic Mechanisms}
\label{subsec:Mecha}

To start with, we briefly discuss these two families of mechanisms that is used in the sample setting and give high level ideas on how to design the order statistic mechanisms. Order statistic mechanisms will be used when we only have samples from the distribution and quantile mechanisms will help us analyze the performance of order statistic mechanisms. Actually, we will point out that quantile mechanisms and order statistic mechanisms are equivalent in some sense.

\subsubsection{Connection Between Two Mechanisms}
\label{subsec:Connection}

Next we aim to show the connection between these two mechanisms. Such observations give us insights on designing mechanisms with small or large number of samples.

\paragraph{The order statistic mechanism is a special kind of quantile mechanisms} First, we can see that the following two operations are equivalent:
\begin{itemize}
	\item Draw a sample from distribution $F$.
	\item Uniformly sample a quantile $x$ from $[0, 1]$, and use $F^{-1}(x)$ as the sample.
\end{itemize}

Now suppose $\pdf(x) = N \binom{N-1}{i-1} \cdot x ^{i-1}\cdot(1-x)^{N-i}$ be the p.d.f. of the $i$-th order statistic over $N$ samples drawn uniformly and independently from $[0, 1]$ and let $P_i$ to be $\Pr_{x\sim P}[x = i]$ for any distribution $P$ over $[N]$. Using similar ideas above, it can be proved that any order statistic mechanism $P$ is equivalent to a quantile mechanism $Q$ with probability density function \[q(x) = \sum_{i=1}^N P_i \pdf(x)\]

Therefore, we can analyze the approximation ratio of quantile mechanism $Q$ instead of  order statistic mechanism $P$. If we are able to compute the approximation ratio of any quantile mechanism $Q$, it follows that we can also characterize the optimal order statistic mechanism exactly. When the number of samples are small, we can have a fine-grained analysis of the order statistic mechanisms and use these limited samples carefully. Section~\ref{subsec:smallsample} actually follow such intuitions to characterize the best possible order statistic mechanism.

\paragraph{Quantile mechanisms can be approximated by order statistic mechanisms within any small error} Our goal is that for any quantile mechanism $Q$ with p.d.f. $q(x)$, we need to find some integer $N$ and a distribution $P$ over $[N]$, such that \[q(x) \approx \sum_{i=1}^N P_i f_N^i(x)\]

Since $\sum_{i=1}^N P_i f_N^i(x)$ is a polynomial of degree $N-1$, this could be done for any continuous $q(x)$ on $[0, 1]$ since the Weierstrass approximation theorem states that every continuous function defined on a closed interval can be uniformly approximated as closely as desired by a polynomial function. What's more interesting is that $\{\pdf(x)\}_{i = 1}^N$ are Bernstein basis polynomials and there are a series of work showing that (stochastic) Bernstein polynomials can efficiently and uniformly approximate to any continous function. Therefore, we can have an asymptotic analysis of the order statistic mechanism. What's more, such observation also shows that we have a block-box transformation from any quantile mechanism to mechanisms only using samples.  Section~\ref{subsec:largesample} uses such techniques and ideas.

\subsection{Small Sample Regime}
\label{subsec:smallsample}

In this section, we characterize the optimal order statistic mechanisms with any fixed number of samples for both symmetric and general instances. We first show that, in any setting, if we are able to give a tight analysis of the quantile mechanism, we could directly characterize the optimal order statistic mechanism with any fixed number of samples via an optimization problem. In the next, we show a tight analysis of the quantile mechanism on both symmetric and general instances, and thus we obtain the characterization of the optimal order statistic mechanism.

Recall that an order statistics mechanism with $N$ samples randomly choose a number $i\in  [N]$ according to a previously defined distribution $P$ and select the $i$-th smallest sample as the price, and a quantile mechanism randomly choose a quantile $x\in [0,1]$ from a determined distribution $Q$ and choose the $x$-quantile, i.e. $F^{-1}_S(x)$, as the price. Since every quantile mechanism and order statistic mechanism is determined by the previously defined distribution, we abuse the notation and use distribution $P$ over $[N]$ denote its corresponding order statistic mechanism and distribution $Q$ over $[0, 1]$ denote its corresponding quantile mechanism. 

\begin{lemma}\label{lem:smallsamplekey}
Suppose $\mathcal{C}:\Delta([0,1])\mapsto \mathbb{R}$ maps every quantile mechanism $P$ to its exact approximation ratio. {Let $\mathcal{P}(Q)$ be the corresponding quantile mechanism of the order statistic mechanism $Q$.} Fixing the number of samples $N$, the optimal order statistic mechanism with $N$ samples $Q_N^*$ is characterized by the following optimization problem:
\[Q^*_N = \arg\max_{Q\in \Delta_N} \mathcal{C}(\mathcal{P}(Q))\]
where $\Delta([0, 1])$ is the set of all distributions over $[0, 1]$, i.e. the set of all quantile mechanisms, and $\Delta_N$ is the set of all distributions over $[N]$, i.e. the set of all order statistic mechanisms with $N$ samples.
\end{lemma}

The proof of Lemma~\ref{lem:smallsamplekey} is quite straightforward and thus is postponed to Appendix~\ref{appendix:proofofsmallsamplekey}.

\subsubsection{Symmetric Instances}
\label{subsec:symsmallsample}

Now we study the case when the distributions are symmetric, i.e., $F_S = F_B$, which means that the seller's value $S$ and the buyer's value $B$ are drawn from the same distribution. For simplification, we will use $F$ to refer to their distributions in this setting.

In order to find out the optimal order statistic mechanism, we need to first give a tight analysis of the quantile mechanism.

\begin{lemma} \label{lem:symratio}
For any quantile mechanism for symmetric instance with distribution $Q$ over $[0,1]$, the approximation ratio is exactly
\[ \inf_{x\in [0,1)} \frac{\int_{[0,x]} t(1-x) \, d Q(t) + \int_{(x, 1]} (1-t)x\, dQ(t) + (1-x)}{1-x^2}\]
where $Q(t)$ is the cumulative distribution function of distribution $Q$.
\end{lemma}

Therefore, combining Lemma~\ref{lem:smallsamplekey} and Lemma~\ref{lem:symratio}, we could characterize the optimal order statistic mechanism via an optimization problem.

\begin{theorem}
\label{thm:symoptorderstatistic}

{The optimal order statistic mechanism with $N$ samples for symmetric instances is the solution to the following optimization problem:}
\[P_N^* = \max_{P_1,P_2,\cdots, P_N \geq 0\atop \sum_{i=1}^N P_i = 1} \inf_{x\in [0,1)}  \frac{\int_0^x p(t) t(1-x) \, d t + \int_x^1 p(t)(1-t)x\, dt + (1-x)}{1-x^2}\]
where $p(t) = \sum_{i=1}^N P_i\pdf(x)$ and  $\pdf(x) = N \binom{N-1}{i-1} \cdot x ^{i-1}\cdot(1-x)^{N-i}$ is the p.d.f. of the $i$-th order statistic over $N$ samples drawn uniformly and independently from $[0, 1]$.
\end{theorem}

The proof of Lemma~\ref{lem:symratio} and Theorem~\ref{thm:symoptorderstatistic} is in Appendix~\ref{appendix:proofofsymratio} and \ref{appendix:proofsymoptorderstatistic}. 

It turns out the optimization above is computationally tractable when $N$ is not too large. We solve the optimization problem and find out the optimal order statistic mechanism numerally with different numbers of samples $N$.

To compare with the order statistic mechanisms, we also consider the most natural sample-based mechanism -- the \emph{Empirical Risk Minimization mechanism (ERM)}. We first provide the formal definition below.

\begin{definition}[Empirical Risk Minimization Mechanism]		
	Given $N$ samples $X_1, X_2,\dots, X_N$ drawn from $F$, define $\emF$ be the empirical distribution of these $N$ samples. That is to say, $\emF$ is the distribution with c.d.f. $\emF(x)$ satisfying:
	\[\emF(x) = \frac{1}{N}\sum_{i=1}^N \indic[x\ge X_i]\]
	
	The Empirical Risk Minimization mechanism (ERM) is the mechanism that computes the optimal price according the empirical distribution $\emF$. In particular, for $N$ samples $X_1,X_2,\dots, X_N$,
	
	\[ \mathsf{ERM}(X_1,X_2,\dots, X_N) = \arg\max_{p} \E_{S\sim \emF, B\sim \emF}[S + (B - S)\cdot\indic[B\ge p\ge S]] \]
	
	If there are multiple prices $p$ that maximize the expected welfare, the ERM mechanism may select any of them.
		
\end{definition}

For $N = 1,2,3,5, 10$, we compute the approximation ratio of order statistic mechanisms and also show the upper bound of ERM. The results are listed below. To prove the upper bound, we use a counter example in \cite{kang2022fixed} and show that ERM has a bad performance on this instance. We defer the complete proof of the upper bound of ERM to Appendix~\ref{appendix:ERM} and the details of numerical results to Appendix~\ref{appendix:symmetricexperiment}.

\begin{table}[H]
\centering
\begin{tabular}{ c | c | c  }\hline
 \#Samples & Order Statistics Mechanism & ERM   \\ 
 $1$ & $0.75$  & $0.5$\\  
 $2$ & $0.821$  & $\le 0.67$\\
 $3$ & $0.822$  &$\le 0.75$\\
 $5$ & $0.847$  &$\le 0.76$\\
 $10$ & $0.852$   &$\le 0.80$\\
 $ \infty$ & $\frac{2 + \sqrt{2}}{4}\approx 0.8536$ & /
\end{tabular}
\caption{Approximation Ratios with Different Number of Samples\\	 in the Symmetric Setting.}\label{table:symmetric}
\end{table}

\subsubsection{General Instances}
\label{subsec:generalsmallsample}

Now we consider the general setting, where the buyer's distribution may be different from the seller's. Recall that we only consider mechanisms over seller's information since there is no constant quantile or order statistic mechanism over seller's information. Using similar ideas, we first show a tight analysis regarding quantile mechanisms, which would guide us to discover the optimal order statistic mechanism.

\begin{lemma}[Theorem~4.1 of \cite{blumrosen2021almost}]
\label{lem:asymratio}
For any quantile mechanism $Q$ (over seller's distribution) with cumulative distribution function $Q$, its approximation ratio is exactly \[\min_{x\in [0,1]}{\int_{[0,x]} t\, dQ(t)  +  1-x}\]
\end{lemma}

Similarly, combining Lemma~\ref{lem:smallsamplekey} and Lemma~\ref{lem:asymratio}, we are able to charaterize the optimal order statistic mechanism over with $N$ samples from seller's distribution by an optimization problem:

\begin{theorem}
\label{thm:asymoptorderstatistic}
The optimal order statistic mechanism with $N$ samples for symmetric instances is the solution to the following optimization problem:
\[P_N^* = \max_{P_1,P_2,\cdots, P_N \geq 0\atop \sum_{i=1}^N P_i = 1} \min_{x\in [0,1]}{\int_{[0,x]} t\cdot p(t) d t   +  1-x},     \]
where $p(x) = \sum_{i=1}^N P_i \cdot \pdf(x)$ and $\pdf(x) = N \binom{N-1}{i-1} \cdot x ^{i-1}\cdot(1-x)^{N-i}$ is the p.d.f. of the $i$-th order statistic over $N$ samples drawn uniformly and independently from $[0, 1]$. 
\end{theorem}

The proof of Lemma~\ref{lem:asymratio} and Theorem~\ref{thm:asymoptorderstatistic} is deferred to Appendix~\ref{appendix:proofasymratio} and \ref{appendix:proofasymoptorderstatistic}.

Similarly, such optimization problem is easy to solve when the number of samples $N$ is not to large. We solve the optimization problem numerically for $N = 1,2,3,5, 10$. Note that we do not compare our mechanism to the Empirical Risk Minimization mechanism in the general setting. This is because we only have sample access to the seller's distribution, and the ERM can not be implemented without the buyer's samples. The details of numerical results is defered to Appendix~\ref{appendix:generalexperiment}.

\begin{table}[H]
\centering
\begin{tabular}{ c | c   }\hline
 \#Samples & Order Statistic Mechanism  \\ 
 $1$ & $0.5$  \\  
 $2$ & $0.531$  \\
 $3$ & $0.578$ \\
 $5$ & $0.601$ \\
 $10$ & $0.615$  \\
 $ \infty$ & $1 - {1\over e} \approx 0.632$ 
\end{tabular}
\caption{Approximation Ratios with Different Number of Samples\\ In the General Setting.}\label{table:asymmetric}
\end{table}

\subsection{Asymptotic Analysis: From Quantile to Order Statistics}
\label{subsec:largesample}

In this section, we turn to the case when the number of samples tends to infinity. As we show in section~\ref{subsec:Connection}, we could approximate any quantile mechanisms by order statistic mechanisms within any small error. Using such ideas, we provide a "black-box" reduction that allows us to convert any quantile mechanism with continuous probability density function $q(x)$ to order statistic mechanism with $N$ samples. Here $N$ is usually a polynomial of $\frac{1}{\varepsilon}$, as long as the probability density function is not too crazy. We now formally write it down.

\begin{lemma}\label{lem:convert}
Let $\mathcal{C}:\Delta([0,1])\mapsto \mathbb{R}$ be a function that maps every quantile mechanism $Q$ with continuous probability density function to its approximation ratio such that for any quantile mechanism $Q_1$ with p.d.f. $q_1(x)$ and quantile mechanism $Q_2$ with p.d.f. $q_2(x)$, it holds that 
\[\mathcal{C}(Q_1) - \mathcal{C}(Q_2) \geq -c\cdot \left|q_1 - q_2\right|_{\infty}\]
where $c$ is a constant.

Now let $Q$ be any quantile mechanism with continuous probability density function $q(x)$. Define $M$ as $\max_{x\in [0,1]} q(x)$. For any $\varepsilon > 0$, suppose $n$ is a positive integer satisfying that 
\begin{align}
        \omega\left(\frac{1}{\sqrt{n - 1}}\right) &\leq \varepsilon / 100\label{eq:con1}\\
        2n\exp\left(-\frac{\varepsilon^2}{8\omega^2\left(\frac{1}{\sqrt{n - 1}}\right)}\right) &\leq \varepsilon\label{eq:con2}\\
        \exp\left(-\frac{\varepsilon^2 n}{48M^2}\right) &\leq \varepsilon/2 \label{eq:con3}
\end{align}
where $\omega(h) = \sup_{x,y\in [0,1]\atop |x - y| \leq h} \left|q(x) - q(y)\right|$. Then, there exists an order statistic mechanism with $n$ samples that achieves an approximation ratio of $\mathcal{C}(Q) - c\cdot \varepsilon$.
\end{lemma}

The high level idea of the proof is as follows. Since we know that probability density functions of order statistics form Bernstein basis polynomials, we could approximate the p.d.f. of the quantile mechanism $q(x)$ within any small error. Inequality~\eqref{eq:con1}, \eqref{eq:con2} and \eqref{eq:con3} actually help us to get an order statistic mechanism whose corresponding distribution of quantile is close to the desired quantile mechanism $Q$. Finally, by the property of $\mathcal{C}$, we know that their approximation ratio is also close. The proof is postponed to Appendix~\ref{appendix:proofconvert}. 

In the following, we show that we could apply lemma~\ref{lem:convert} to both the symmetric and general settings and convert the optimal quantile mechanism to order statistic mechanism within a error of at most $\varepsilon$ using $\text{poly}\left(\frac{1}{\varepsilon}\right)$ samples.

\subsubsection{Symmetric Instance}

We first study the case when the distributions are symmetric. \cite{kang2022fixed} provide a mechanism that chooses the mean of the distribution $F$ as the price. They show that in the symmetric setting, this is the optimal fixed price mechanism and achieves an approximation ratio of $\frac{2+\sqrt{2}}{4}$. However, what we want here is a quantile mechanism and we could not convert such mean-based mechanism directly into an order statistic mechanism. We show that quantile mechanisms can also reach the optimal $\frac{2 + \sqrt{2}}{4}$-approximation ratio. After that, we use the technique in Lemma~\ref{lem:convert} to produce an order statistic mechanism that achieves an approximation ratio of $\frac{2 + \sqrt{2}}{4} - \varepsilon$ with $\mathrm{poly}\left(\frac{1}{\varepsilon}\right)$ samples.

To start with, we first show our optimal order statistic mechanism.

\begin{theorem}
\label{thm:symquantile}
	There is a $\frac{2 + \sqrt{2}}{4}$-approximation quantile mechanism in the symmetric setting. 
\end{theorem}

\begin{proof}
    Our quantile mechanism $Q$ runs as following:
    \begin{itemize}
    \item Let $F$ be the c.d.f. of the distribution.
    \item Output $F^{-1}\left(\frac{\sqrt{2}}{2}\right)$, i.e. $\frac{\sqrt{2}}{2}$-quantile of the distribution, as the price.
\end{itemize}
    The approximation ratio could be directly calculated using Lemma~\ref{lem:symratio}.
    One could see that
    \begin{align*}
    \begin{split}
         &\inf_{x\in [0,1)} \frac{\int_{[0,x]} t(1-x) \, d Q(t) + \int_{(x, 1]} (1-t)x\, dQ(t) + (1-x)}{1-x^2} \\
         &= \min\left(\min_{x\in[0,\frac{\sqrt{2}}{2}]} \frac{x\cdot(1-\frac{\sqrt{2}}{2}) + 1-x}{1-x^2}, \inf_{x\in[\frac{\sqrt2}2,1)}\frac{\frac{\sqrt{2}}{2}\cdot(1 -x) + (1-x)}{1-x^2}\right)\\
         & = \frac{2+\sqrt{2}}{4}.
    \end{split}
    \end{align*}which completes the proof.    
\end{proof}

Now we aim to convert it to an order statistic mechanism. 

\begin{theorem}
\label{thm:symapprox}
There exists an order statistic mechanism $P$ with $N = O\left(\frac{1}{\varepsilon^7}\right)$ samples that achieves $\frac{2 + \sqrt{2}}{4} - \varepsilon$ approximation.
\end{theorem}
To start with, we may notice that it is impossible to directly apply Lemma~\ref{lem:convert} to the optimal quantile mechanism a since it is does not have a continiuous probability density function. So our first step is to provide a quantile mechanism with continuous distribution.
\begin{lemma}
\label{lem:symturning}
For any $\varepsilon > 0$, there exists a quantile mechanism $Q'$ with a probability density function $\tilde{q}(x)$ such that $\omega\left(c \cdot \varepsilon^{3.5}\right) \leq 32^2 c\cdot \varepsilon^{1.5}$ and  $\max_{x\in[0,1]} \tilde{q}(x) \leq 32/\varepsilon$. Furthermore, the mechanism $Q'$ achieves an approximation ratio of $\frac{\sqrt2 + 2}{4} - \varepsilon/2$.
\end{lemma}

Our last step is to make sure that the approximation ratio would not differ to much for two probability density functions that are close to each other.
\begin{lemma}
\label{lem:symclose}
Suppose $\mathcal{C}(p)$ is the function that maps every quantile mechanism $Q$ for symmetric instances with a continuous probabilitiy density function $q(x)$ to its approximation ratio where
\[ \mathcal{C}(p)=\inf_{x\in [0,1)} \frac{\int_{[0,x]} q(t)t(1-x) \, dt + \int_{(x, 1]} q(t)(1-t)x\, dt + (1-x)}{1-x^2}.\]

For any quantile mechanism $Q_1$ with continuous p.d.f $q_1(x)$ and $Q_2$ with continuous p.d.f. $q_2(x)$, it holds that
\[\mathcal{C}(p_1) - \mathcal{C}(p_2) \geq - \left|q_1 - q_2\right|_{\infty}.\]
\end{lemma}
We first use these lemmas to give a proof of Theorem~\ref{thm:symapprox}, and leave the proof of Lemma~\ref{lem:symturning} and Lemma~\ref{lem:symclose} to Appendix~\ref{appendix:proofofsymturning} and \ref{appendix:proofsymclose}.

\begin{proof}[Proof Of Theorem~\ref{thm:symapprox}]
	
	Let $n = c\cdot \frac{1}{\varepsilon^7} + 1$ where $c$ is a large enough constant. Notice that $\omega\left(\frac{1}{\sqrt{n - 1}}\right)$ equals $\omega\left(c^{-0.5}\cdot \varepsilon^{3.5}\right)\leq 32^2 c^{-0.5}\cdot \varepsilon^{1.5}$. This implies that $2n\exp\left(-\frac{\varepsilon^2}{8\omega^2\left(\frac{1}{\sqrt{n - 1}}\right)}\right) \leq \varepsilon / 2$. Besides, define $M$ be $\max_{x\in [0,1]} \tilde{q}(x)$ which is at most $ 32/\varepsilon$. It is also easy to verify that $\exp\left(-\frac{\varepsilon^2 n}{48M^2}\right) \leq \varepsilon / 4$. Combining the properties above with Lemma~\ref{lem:symclose}, we could apply Lemma~\ref{lem:convert}, and see that there exists an order statistic mechanism with $n$ samples with an approximation of at least $\mathcal{C}(Q') - \varepsilon / 2$. Together with Lemma~\ref{lem:symturning}, it follows that this order statistic mechanism is $\frac{2+\sqrt2}{4} - \varepsilon$-approximate. This concludes our proof.

    We would like to comment that if we always choose the $\left\lfloor \frac{\sqrt2}{2}\cdot n\right\rfloor$-th order statistic as the price, there is an argument to prove that we could achieve an approximation ratio of $\frac{2+\sqrt2}{4} - \varepsilon$ with $\tilde{O}\left(\frac{1}{\varepsilon^2}\right)$ samples. 
\end{proof}

\subsubsection{General Instance}

We now consider the general instance. \cite{blumrosen2021almost} provides a $1 - 1/e$ approximation quantile mechanism that is also shown to be optimal by \cite{kang2022fixed}. Using the block-box reduction shown in Lemma~\ref{lem:convert}, we show that the optimal quantile mechanism can be approximated by order statistics mechanism as closely as desired and also obtain an upper bound on the sample complexity.

\begin{theorem}
\label{thm:asymapprox}
There exists an order statistic mechanism $P$ with $N = O\left(\frac{1}{\varepsilon^5}\right)$ samples that achieves $1 - \frac{1}{e} - \varepsilon$ approximation. 
\end{theorem}

In the following proof, we will use order statistic mechanism to approximate the optimal quantile mechanism $Q$ with p.d.f. $q(x) = \frac{1}{x}$ on $[1/e, 1]$. Similarly, $q(x)$ is not a continuous function on $[0, 1]$. Thus, we need to first convert it to a continuous function $\tilde{q}(x)$ on $[0, 1]$ and then apply Lemma~\ref{lem:convert}.

Similarly, we introduce the following two lemmas first.

\begin{lemma}
\label{lem:asymturning}
For any $\varepsilon > 0$, there exists a quantile mechanism $Q'$ with a probability density function $\tilde{q}(x)$ such that $\omega\left(c\cdot \varepsilon^{-2.5}\right) \leq 100c^{0.5} \cdot \varepsilon^{1.5}$ and $\max_{x\in[0,1]} \tilde{q}(x)\leq 2e$. Besides, the quantile mechanism $Q'$ has an approximation ratio of  $1 - \frac{1}{e} - \varepsilon/2$.

\end{lemma}

\begin{lemma}
\label{lem:asymclose}
Suppose $\mathcal{C}(p)$ is the function that maps every quantile mechanism $Q$ for general instances with a continuous probabilitiy density function $q(x)$ to its approximation ratio where
\[ \mathcal{C}(p)=\min_{x\in [0,1]} {\int_0^x q(t) t\, dt  +  1-x}.\]

For any quantile mechanism $Q_1$ with continuous p.d.f $q_1(x)$ and $Q_2$ with continuous p.d.f. $q_2(x)$, it holds that
\[\mathcal{C}(p_1) - \mathcal{C}(p_2) \geq - \left|q_1 - q_2\right|_{\infty}.\]
	
\end{lemma}

	The proofs of Lemma~\ref{lem:asymturning} and Lemma~\ref{lem:asymclose} are respectively in Appendix~\ref{appendix:proofasymturning} and \ref{appendix:proofasymclose}. 

\begin{proof}[Proof Of Theorem~\ref{thm:asymapprox}]
	We follow the same argument to prove Theorem~\ref{thm:asymapprox}.

Let $n = c\cdot \frac{1}{\varepsilon^5} + 1$ where $c$ is a large enough constant. Again it is easy to see that $\omega\left(\frac{1}{\sqrt{n - 1}}\right) \leq 100c^{-0.5} \cdot \varepsilon^{1.5}$. Thus it holds that $2n\exp\left(-\frac{\varepsilon^2}{8\omega^2\left(\frac{1}{\sqrt{n - 1}}\right)}\right) \leq \varepsilon / 2$. 

Besides, define $M = \max_{x\in [0,1]} \tilde{q}(x) \leq 2e$, we could also see that $\exp\left(-\frac{\varepsilon^2 n}{48M^2}\right) \leq \varepsilon / 4$. Combining the properties above with Lemma~\ref{lem:asymclose}, we could see the existence of an order statistic mechanism with $n$ samples with an approximation of at least $\mathcal{C}(Q') - \varepsilon / 2$ by applying Lemma~\ref{lem:convert}. Together with Lemma~\ref{lem:asymturning}, we know that the approximation ratio of this order statistic mechanism is $1 - 1/e - \varepsilon$. This finishes our proof.
\end{proof}

\bibliographystyle{plain}
\bibliography{bibliography,references}

\begin{thebibliography}{10}

\bibitem{allouah_revenue_2021}
Amine Allouah, Achraf Bahamou, and Omar Besbes.
\newblock Revenue {Maximization} from {Finite} {Samples}.
\newblock In {\em Proceedings of the 22nd {ACM} {Conference} on {Economics} and
  {Computation}}, {EC} '21, page~51, New York, NY, USA, July 2021. Association
  for Computing Machinery.

\bibitem{babaioff_best_2018}
Moshe Babaioff, Yang Cai, Yannai~A. Gonczarowski, and Mingfei Zhao.
\newblock The {Best} of {Both} {Worlds}: {Asymptotically} {Efficient}
  {Mechanisms} with a {Guarantee} on the {Expected} {Gains}-{From}-{Trade}.
\newblock In {\em Proceedings of the 2018 {ACM} {Conference} on {Economics} and
  {Computation}}, {EC} '18, page 373, New York, NY, USA, June 2018. Association
  for Computing Machinery.

\bibitem{babaioff_bulow-klemperer-style_2020}
Moshe Babaioff, Kira Goldner, and Yannai~A. Gonczarowski.
\newblock Bulow-{Klemperer}-{Style} {Results} for {Welfare} {Maximization} in
  {Two}-{Sided} {Markets}.
\newblock In Shuchi Chawla, editor, {\em Proceedings of the 2020 {ACM}-{SIAM}
  {Symposium} on {Discrete} {Algorithms}, {SODA} 2020, {Salt} {Lake} {City},
  {UT}, {USA}, {January} 5-8, 2020}, pages 2452--2471. SIAM, 2020.

\bibitem{babaioff_are_2018}
Moshe Babaioff, Yannai~A. Gonczarowski, Yishay Mansour, and Shay Moran.
\newblock Are {Two} ({Samples}) {Really} {Better} {Than} {One}?
\newblock In {\em Proceedings of the 2018 {ACM} {Conference} on {Economics} and
  {Computation}, {Ithaca}, {NY}, {USA}, {June} 18-22, 2018}, page 175, 2018.

\bibitem{blumrosen2021almost}
Liad Blumrosen and Shahar Dobzinski.
\newblock (almost) efficient mechanisms for bilateral trading.
\newblock {\em Games and Economic Behavior}, 130:369--383, 2021.

\bibitem{blumrosen_approximating_2016}
Liad Blumrosen and Yehonatan Mizrahi.
\newblock Approximating {Gains}-from-{Trade} in {Bilateral} {Trading}.
\newblock In {\em Web and {Internet} {Economics} - 12th {International}
  {Conference}, {WINE} 2016, {Montreal}, {Canada}, {December} 11-14, 2016,
  {Proceedings}}, pages 400--413, 2016.

\bibitem{brustle_multi-item_2020}
Johannes Brustle, Yang Cai, and Constantinos Daskalakis.
\newblock Multi-{Item} {Mechanisms} without {Item}-{Independence}:
  {Learnability} via {Robustness}.
\newblock In {\em Proceedings of the 21st {ACM} {Conference} on {Economics} and
  {Computation}}, {EC} '20, pages 715--761, New York, NY, USA, July 2020.
  Association for Computing Machinery.

\bibitem{brustle_approximating_2017}
Johannes Brustle, Yang Cai, Fa~Wu, and Mingfei Zhao.
\newblock Approximating {Gains} from {Trade} in {Two}-sided {Markets} via
  {Simple} {Mechanisms}.
\newblock In {\em Proceedings of the 2017 {ACM} {Conference} on {Economics} and
  {Computation}}, {EC} '17, pages 589--590, New York, NY, USA, June 2017.
  Association for Computing Machinery.

\bibitem{cai_learning_2017}
Yang Cai and Constantinos Daskalakis.
\newblock Learning {Multi}-{Item} {Auctions} with (or without) {Samples}.
\newblock In {\em 2017 {IEEE} 58th {Annual} {Symposium} on {Foundations} of
  {Computer} {Science} ({FOCS})}, pages 516--527, October 2017.
\newblock ISSN: 0272-5428.

\bibitem{cai_multi-dimensional_2021}
Yang Cai, Kira Goldner, Steven Ma, and Mingfei Zhao.
\newblock On {Multi}-{Dimensional} {Gains} from {Trade} {Maximization}.
\newblock In {\em Proceedings of the 2021 {ACM}-{SIAM} {Symposium} on
  {Discrete} {Algorithms} ({SODA})}, Proceedings, pages 1079--1098. Society for
  Industrial and Applied Mathematics, January 2021.

\bibitem{cole_sample_2014}
Richard Cole and Tim Roughgarden.
\newblock The sample complexity of revenue maximization.
\newblock In {\em Proceedings of the 46th {Annual} {ACM} {Symposium} on
  {Theory} of {Computing}}, 2014.

\bibitem{colini-baldeschi_fixed_2017}
Riccardo Colini-Baldeschi, Paul Goldberg, Bart de~Keijzer, Stefano Leonardi,
  and Stefano Turchetta.
\newblock Fixed price approximability of the optimal gain from trade.
\newblock In {\em International {Conference} on {Web} and {Internet}
  {Economics}}, pages 146--160. Springer, 2017.

\bibitem{daskalakis2020more}
Constantinos Daskalakis and Manolis Zampetakis.
\newblock More revenue from two samples via factor revealing sdps.
\newblock In {\em Proceedings of the 21st ACM Conference on Economics and
  Computation}, pages 257--272, 2020.

\bibitem{deng2021approximately}
Yuan Deng, Jieming Mao, Balasubramanian Sivan, and Kangning Wang.
\newblock Approximately efficient bilateral trade.
\newblock {\em arXiv preprint arXiv:2111.03611}, 2021.

\bibitem{dhangwatnotai_revenue_2010}
Peerapong Dhangwatnotai, Tim Roughgarden, and Qiqi Yan.
\newblock Revenue maximization with a single sample.
\newblock In {\em Proceedings of the 11th {ACM} {Conference} on {Electronic}
  {Commerce} ({EC})}, 2010.

\bibitem{dutting_efficient_2021}
Paul D{\"u}tting, Federico Fusco, Philip Lazos, Stefano Leonardi, and Rebecca
  Reiffenh{\"a}user.
\newblock Efficient {Two}-{Sided} {Markets} with {Limited} {Information}.
\newblock {\em arXiv:2003.07503 [cs]}, April 2021.
\newblock arXiv: 2003.07503.

\bibitem{elkind_designing_2007}
Edith Elkind.
\newblock Designing and learning optimal finite support auctions.
\newblock In {\em Proceedings of the eighteenth annual {ACM}-{SIAM} symposium
  on {Discrete} algorithms}, pages 736--745. Society for Industrial and Applied
  Mathematics, 2007.

\bibitem{fu_randomization_2015}
Hu~Fu, Nicole Immorlica, Brendan Lucier, and Philipp Strack.
\newblock Randomization {Beats} {Second} {Price} as a {Prior}-{Independent}
  {Auction}.
\newblock In {\em Proceedings of the {Sixteenth} {ACM} {Conference} on
  {Economics} and {Computation}, {EC} '15, {Portland}, {OR}, {USA}, {June}
  15-19, 2015}, page 323, 2015.

\bibitem{goldner_prior-independent_2016}
Kira Goldner and Anna~R Karlin.
\newblock A prior-independent revenue-maximizing auction for multiple additive
  bidders.
\newblock In {\em International {Conference} on {Web} and {Internet}
  {Economics}}, pages 160--173. Springer, 2016.

\bibitem{gonczarowski_sample_2018}
Yannai~A Gonczarowski and S~Matthew Weinberg.
\newblock The {Sample} {Complexity} of {Up}-to-\${\textbackslash}varepsilon\$
  {Multi}-{Dimensional} {Revenue} {Maximization}.
\newblock In {\em 2018 {IEEE} 59th {Annual} {Symposium} on {Foundations} of
  {Computer} {Science} ({FOCS})}, pages 416--426. IEEE, 2018.

\bibitem{guo_settling_2019}
Chenghao Guo, Zhiyi Huang, and Xinzhi Zhang.
\newblock Settling the {Sample} {Complexity} of {Single}-parameter {Revenue}
  {Maximization}.
\newblock In {\em the 51st {Annual} {ACM} {Symposium} on {Theory} of
  {Computing} ({STOC})}, 2019.

\bibitem{gurobi}
{Gurobi Optimization, LLC}.
\newblock {Gurobi Optimizer Reference Manual}, 2022.

\bibitem{HAGERTY198794}
Kathleen~M Hagerty and William~P Rogerson.
\newblock Robust trading mechanisms.
\newblock {\em Journal of Economic Theory}, 42(1):94--107, 1987.

\bibitem{kang2022fixed}
Zi~Yang Kang, Francisco Pernice, and Jan Vondrák.
\newblock Fixed-price approximations in bilateral trade.
\newblock In {\em Proceedings of the 2022 Annual ACM-SIAM Symposium on Discrete
  Algorithms (SODA)}, pages 2964--2985. SIAM, 2022.

\bibitem{kang2018strategy}
Zi~Yang Kang and Jan Vondr\'ak.
\newblock Strategy-proof approximations of optimal efficiency in bilateral
  trade.
\newblock 2018.

\bibitem{mohri_learning_2014}
Mehryar Mohri and Andres~Munoz Medina.
\newblock Learning {Theory} and {Algorithms} for revenue optimization in second
  price auctions with reserve.
\newblock In {\em {ICML}}, pages 262--270, 2014.

\bibitem{morgenstern_learning_2016}
Jamie Morgenstern and Tim Roughgarden.
\newblock Learning simple auctions.
\newblock In {\em Proceedings of the 30th {Annual} {Conference} on {Learning}
  {Theory} ({COLT})}, 2016.

\bibitem{morgenstern_pseudo-dimension_2015}
Jamie~H Morgenstern and Tim Roughgarden.
\newblock On the pseudo-dimension of nearly optimal auctions.
\newblock In {\em Proceedings of the the 29th {Annual} {Conference} on {Neural}
  {Information} {Processing} {Systems} ({NIPS})}, 2015.

\bibitem{myerson_efficient_1983}
Roger~B Myerson and Mark~A Satterthwaite.
\newblock Efficient mechanisms for bilateral trading.
\newblock {\em Journal of economic theory}, 29(2):265--281, 1983.
\newblock Publisher: Elsevier.

\bibitem{DBLP:journals/moc/SunWZ21}
Xingping Sun, Zongmin Wu, and Xuan Zhou.
\newblock On probabilistic convergence rates of stochastic bernstein
  polynomials.
\newblock {\em Math. Comput.}, 90(328):813--830, 2021.

\bibitem{syrgkanis_sample_2017}
Vasilis Syrgkanis.
\newblock A sample complexity measure with applications to learning optimal
  auctions.
\newblock In {\em Advances in {Neural} {Information} {Processing} {Systems}},
  pages 5352--5359, 2017.

\end{thebibliography}
\appendix
\section{Tie Breaking}\label{appendix:tiebreaking}

For distribution $D$ with point masses, the following reduction will convert it to continuous one. We will overload the notation of $D$ and think of it as a bivariate distribution with the first coordinate drawn from the previous single-variate distribution $D$ and the second tie-breaker coordinate drawn independently and uniformly from $[0,1]$. And $(X_1, t_1) > (X_2, t_2)$ if and only if either $X_1 > X_2$, or $X_1 = X_2$ and $t_1 > t_2$. Since the tie-breaker coordinate is continuous, the probability of having $(X_1, t_1) = (X_2, t_2)$ for any two values during a run of any mechanism is zero. Therefore we could define the c.d.f. of $D$ as \[F_{D}(X,t) = \Pr_{(Y,u) \sim (D,U[0,1])}[(Y, u) < (X, t)]\]

Remind the second coordinate is only used to break ties, and it does not affect the calculation of welfare. After including the additional random variable, we can see that $D$ has been converted into a continuous distribution since its second coordinate is continuous. 

\section{Missing Proofs in Section~\ref{sec:fullinfo}}
\label{appendix:fullinfo}

\subsection{Proof of Lemma~\ref{lem:fullinfo}}
\label{appendix:prooffullinfo}
We first show the proof of Lemma~\ref{lem:fullinfo}.

The approximation ratio of the optimal fixed-price mechanism could be written as
\begin{align*}
    \begin{split}
        \min_{\I = (F_S, F_B)} \max_{p\in \mathbb{R}} \frac{\ALGW(\I, p)}{\OPTW(\I)}.
    \end{split}
\end{align*}

We first show that for any instance $\I = (F_S, F_B)$, there is a valid solution $(u, v, r)$ such that $r = \max_{p\in \mathbb{R}}\\ \frac{\ALGW(\I, p)}{\OPTW(\I)}$. We could first simply scale the instance by $\frac{1}{\OPTW(\I)}$ to $\I' = (F_S', F_B')$ where $\OPTW(\I') = 1$. Such scaling means that $\ALGW(\I, \OPTW(\I)\cdot p) = \OPTW(\I)\cdot \ALGW(\I', p)$ for all $p\in \mathbb{R}_{\ge 0}$. This implies that
\begin{align*}
	\begin{split}
		\max_{p\in \mathbb{R}}  \frac{\ALGW(\I, p)}{\OPTW(\I)} = \max_{p\in \mathbb{R}}  {\ALGW(\I', p)}.
	\end{split}
\end{align*}

Therefore, let $u$ and $v$ be the probability measures of $F_S'$ and $F_B'$ and $r$ be $\max_{p\in \mathbb{R}}  {\ALGW(\I', p)}$. It is easy to verify that $(u, v, r)$ is a valid solution. Let $r^*$ be the optimal value of $\mathsf{FullOp}$, this implies that $r^*\leq \max_{p\in \mathbb{R}}  \frac{\ALGW(\I, p)}{\OPTW(\I)}$ for any instance $\I = (F_S, F_B)$. Taking the minimum over all possible $\I$, we then get that 
\begin{equation}
\label{eq:fullinfopr1}
r^* \leq \min_{\I = (F_S, F_B)} \max_{p\in \mathbb{R}} \frac{\ALGW(\I, p)}{\OPTW(\I)}
\end{equation}
 
 Next, let $(u^*, v^*, r^*)$ be the optimal solution of $\mathsf{FullOp}$. Since $u^*, v^*$ are both probability measures,  let $F_S^*, F_B^*$ be the corresponding distributions of $u$ and $v$ and $\I^* = (F_S^*, F_B^*)$ be the instance. Now by the constraint of $\mathsf{FullOp}$, we know that $\OPTW(\I^*)\geq 1$. Besides, $(u^{*}, v^*, r^*)$ is an optimal solution implies that $r^{*} = \max_{p\in \mathbb{R}} \ALGW(\I^*, p)$. Therefore,
\begin{equation}
\label{eq:fullinfopr2} 	
 	r^* = \max_{p\in \mathbb{R}} \ALGW(\I^*, p) \geq \max_{p\in \mathbb{R}} \frac{\ALGW(\I^*, p)}{\OPTW(\I^*)} \geq \min_{\I = (F_S, F_B)} \max_{p\in \mathbb{R}} \frac{\ALGW(\I, p)}{\OPTW(\I)}
\end{equation}
 
 Now combining inequality~\eqref{eq:fullinfopr1} and \eqref{eq:fullinfopr2}, it follows that
 
 \[r^* = \min_{\I = (F_S, F_B)} \max_{p\in \mathbb{R}} \frac{\ALGW(\I, p)}{\OPTW(\I)}\]
 
 which completes the proof.

\subsection{Proof of Lemma~\ref{lem:fullinfolower}}
\label{appendix:prooffullinfolower}

Before we give the proof of Lemma~\ref{lem:fullinfolower}, we first show prove a lemma that helps us discretize a continuous distribution.
\begin{lemma}
\label{lem:discretization}
    For any instance $\I = (F_S, F_B)$, and $0 = p_1 < p_2 < \cdots < p_n$, there exists a set of numbers $\{s_i\}_{i\in [n]}, \{b_i\}_{i\in [n]}$ satisfying the following equations.
    \begin{align}
        s_i , b_i & \geq 0\quad \forall i\in [n]\label{eq:discretization1}\\
        1\le\sum_{i=1}^n s_i &\le 1 + \frac{\E[S]}{p_n} \quad 1 \le \sum_{i=1}^n b_i \le 1+ \frac{\E[B]}{p_n}\label{eq:discretization2}\\ 
        \sum_{i=1}^t s_i &\geq \Pr[S < p_t]\quad  \sum_{i=1}^t b_i \geq \Pr[B<p_t]\quad \forall t\in [n]\label{eq:discretizationNew}\\
        \sum_{i=1}^{n-1} s_i &\leq 1 \quad \sum_{i=1}^{n-1} b_i \leq 1\label{eq:discretization6}\\
        \E_{S\sim F_S}\left[S\right] &= \sum_{i=1}^n s_i p_i\label{eq:discretization3}\\
        \E_{B\sim F_B}\left[B\right] &= \sum_{j=1}^n b_j p_j\label{eq:discretization4}\\
        \OPTW(\I) &\defeq \E_{S\sim F_S\atop B\sim F_B}[\max(S, B)]\leq \sum_{i=1}^n \sum_{j=1}^n s_i b_j \cdot \max(p_i, p_j)\label{eq:discretization5}\\
        \ALGW(\I, p_t) &\defeq \E_{S\sim F_S}[S] + \E_{S\sim F_S\atop B\sim F_B}\Big[(B-S)\cdot \indic[S\leq p_t\leq B]\Big]\nonumber\\
            & \quad \geq \sum_{i=1}^n s_i p_i+\sum_{i=1}^{t - 1} \sum_{j=t+1}^{n} s_i b_j (p_j - p_i) \quad \forall t\in [n] \label{eq:discretization7}
    \end{align}
    Additionally, $\{b_i\}_{i\in [n]}$ and $\{s_i\}_{i\in [n]}$ only depends on $F_B$ and $F_S$ respectively.
\end{lemma}
\begin{proof}
   We construct $(s_1,\cdots,s_n,b_1,\cdots,b_n)$ as follows. For the seller, define \[q_{s,i} = \Pr_{S\sim F_S}[p_i\le S < p_{i+1}] \text{ and } E_{s,i} = \E_{S\sim F_S}[S\cdot\indic[p_i\le S < p_{i+1}]], ~ \forall i\in [n],\]

    where we assume that $p_{n+1} = +\infty$. It is clear from the definition that $q_{s,i} \cdot p_i \le E_{s,i} \le q_{s,i} \cdot p_{i+1}$. Therefore, for any $i\in [n - 1]$, there exists non-negative numbers $s_{i,\LF}$ and $s_{i + 1, \RT}$ such that  
    \begin{equation}\label{eq:discredef}
        s_{i,\LF} + s_{i + 1, \RT} = q_{s,i}~~ \text{and}~~s_{i,\LF} \cdot p_i + s_{i + 1, \RT}\cdot p_{i+1} = E_{s,i}.
    \end{equation}

   {We further define $s_{n,\LF}$ as $E_{s,n} / p_n$ and $s_{1,\RT} = 0$. Now set $s_i = s_{i,\LF} + s_{i,\RT}$ for all $i\in [n]$. For the buyer, we define $\{b_i\}_{i\in [n]}$ similarly.} Therefore, from the construction, it is clear that $\{b_i\}_{i\in [n]}$ and $\{s_i\}_{i\in [n]}$ only depend on $F_B$ and $F_S$ respectively.

    We now verify that $\left(\{s_i\}_{i\in [n]}, \{b_i\}_{i\in [n]}\right)$ satisfies the properties above. The non-negativity of $s_i$ and $b_i$ is immediately derived from $s_{i,\LF},s_{i,\RT} \geq 0$. { From our definition, it is clear that  $\sum_{i=1}^n q_{s,i} = 1$, therefore} 
    
    \[\sum_{i=1}^n s_i  = \sum_{i=1}^n s_{i,\LF} + s_{i,\RT} \geq \sum_{i=1}^n q_{s,i} = 1,\]

    and

    \[\sum_{i=1}^{n - 1} s_i = \sum_{i=1}^{n-1} s_{i,\LF} + s_{i,\RT} \leq \sum_{i=1}^n q_{s,i} = 1.\]
    
    We could also see that
       { \[\sum_{i=1}^n s_i = \sum_{i=1}^n s_{i,\LF} + s_{i,\RT} \leq \sum_{i=1}^n q_{s,i} + s_{n,\LF} = 1 + \frac{E_{s,n}}{p_n} \leq 1 + \frac{\mathbb{E}[S]}{p_n}.\]}

	For any $t\in [n]$, it is clear that 
	\[\sum_{i=1}^t s_i = \sum_{i=1}^t s_{i,\LF} + s_{i,\RT} \geq \sum_{i=1}^{t-1} q_{s,i} = \Pr[S < p_t].\]

    For the expectations, it holds that
    \[\sum_{i=1}^n s_i p_i = \sum_{i=1}^n \left(s_{i,\LF} + s_{i,\RT}\right)\cdot p_i = \sum_{i=1}^n E_{s,i} = \E[S]\]

    By symmetry, similar inequalities also holds for for $\{b_i\}_{i\in [n]}$. So far, we have verified that properties \eqref{eq:discretization1}, \eqref{eq:discretization2}, \eqref{eq:discretizationNew}, \eqref{eq:discretization6}, \eqref{eq:discretization3} and \eqref{eq:discretization4} are satisfied. It only remains to show that \eqref{eq:discretization5} and\eqref{eq:discretization7} holds.

    For any $i\neq j\in [n]$, w.l.o.g. we can assume that $i < j$.  We could see that 
\begin{align}\label{eq:FullVeri1}
    \begin{split}        
        &\E_{S\sim F_S\atop B\sim F_B}\Big[\max(S,B)\cdot \indic[S\in [p_i,p_{i+1})] \indic[B \in [p_j,p_{j+1})]\Big] \\
        & =  \E_{S\sim F_S\atop B\sim F_B}\Big[B\cdot  \indic[S\in [p_i,p_{i+1})] \indic[B \in [ p_j,p_{j+1})]\Big]\\
       & =  \E_{B\sim F_B}\Big[B\cdot \indic[B \in [p_j,p_{j+1})]\Big] \cdot \Pr_{S\sim F_S}\Big[ S\in [p_i, p_{i+1})\Big]\\
       & =  E_{b, j} \cdot q_{s, i}\\
       & = (b_{j,\LF} \cdot p_j + b_{j + 1, \RT} \cdot p_{j+1}) \cdot (s_{i,\LF} + s_{i+1, \RT})\\
       & = s_{i,\LF} \cdot b_{j,\LF} \max(p_i, p_j) + s_{i+1,\RT} \cdot b_{j,\LF} \max(p_{i+1}, p_j)   \\
       &\quad + s_{i,\LF} \cdot b_{j+1,\RT} \max(p_i, p_{j+1}) + s_{i+1,\RT} \cdot b_{j+1,\RT} \max(p_{i+1}, p_{j+1})
    \end{split}
    \end{align}
  The second equality is due to the independence between $S$ and $B$.

    Now consider the case when $i = j \leq n - 1$. For any $x, y\in [p_i, p_{i+1}]$, we have 
    \begin{align*}
        \begin{split}
            \max(x,y) \le \max(p_i, y) \cdot \frac{p_{i+1} - x}{p_{i + 1} - p_i}+ \max(p_{i+1}, y) \cdot\frac{x - p_i}{p_{i+ 1} - p_{i}},
        \end{split}
    \end{align*} as $\frac{p_{i+1} - x}{p_{i + 1} - p_i} + \frac{x-p_{i}}{p_{i+ 1} - p_{i}} = 1$ and $\frac{p_{i+1} - x}{p_{i + 1} - p_i} \cdot p_i + \frac{x - p_{i}}{p_{i+ 1} - p_{i}} \cdot p_{i+1} = x$.

  Based on the inequality above, for any fixed $y\in [p_i, p_{i+1})$, we have
\begin{align*}
        \begin{split}
            &\E_{S\sim F_S}\Big[\max(S, y)\cdot \indic[S\in [p_i, p_{i+1})]\Big] \\
            \leq & \E_{S\sim F_S}\left[\left( \max(p_i, y) \cdot \frac{p_{i+1} - S}{p_{i + 1} - p_i}+ \max(p_{i+1}, y) \cdot\frac{S - p_i}{p_{i+ 1} - p_{i}} \right) \indic[S\in [p_i,p_{i+1})]\right]\\
             = &\max(p_i, y) \E_{S\sim F_S}\left[\frac{p_{i+1} - S}{p_{i+1} - p_i} \cdot \indic[S\in [p_i, p_{i+1})]\right] + \max(p_{i+1}, y)\E_{S\sim F_S}\left[\frac{S - p_i}{p_{i+1} - p_i}\cdot \indic[S\in [p_i, p_{i+1})]\right]\\
             = &y\cdot s_{i,\LF} + p_{i+1}\cdot s_{i+1,\RT}
        \end{split}
    \end{align*} The last equality is because of the following identities:
    \begin{align*}
        \begin{split}
    \E_{S\sim F_S}\left[\frac{p_{i+1} - S}{p_{i+1} - p_i} \cdot \indic[S\in [p_i, p_{i+1})]\right] + \E_{S\sim F_S}\left[\frac{S - p_i}{p_{i+1} - p_i}\cdot \indic[S\in [p_i, p_{i+1})]\right] &= \Pr[S\in [p_i, p_{i+1})]]        \\
    p_i\cdot \E_{S\sim F_S}\left[\frac{p_{i+1} - S}{p_{i+1} - p_i} \cdot \indic[S\in [p_i, p_{i+1})]\right] + p_{i+1} \cdot \E_{S\sim F_S}\left[\frac{S - p_i}{p_{i+1} - p_i}\cdot \indic[S\in [p_i, p_{i+1})]\right] &= \E[S\cdot \indic[S\in[p_i, p_{i+1})]].
        \end{split}
    \end{align*}

    Hence, we can conclude that $\left(\E_{S\sim F_S}\left[\frac{p_{i+1} - S}{p_{i+1} - p_i} \cdot \indic[S\in [p_i, p_{i+1})]\right], \E_{S\sim F_S}\left[\frac{S - p_i}{p_{i+1} - p_i}\cdot \indic[S\in [p_i, p_{i+1})]\right]\right)$ is the unique solution to \eqref{eq:discredef}. Thus, these two numbers respectively equal to $s_{i,\LF}$ and $s_{i+1,\RT}$.

    Due to the inequality above, we have
\begin{align}\label{eq:FullVeri2}
        \begin{split}
            &\E_{B\sim F_B}\left[ \E_{S\sim F_S}\Big[\max(S, B)\cdot \indic[S\in [p_i, p_{i+1})]\Big]\cdot \indic[B\in [p_i, p_{i+1})]\right]\\
             {\leq} &\E_{B\sim F_B}\left[\left(B\cdot s_{i,\LF} + p_{i+1}\cdot s_{i+1,\RT}\right)\cdot \indic[B\in [p_i, p_{i+1})]\right]\\
             = &b_{i,\LF} s_{i,\LF}\cdot p_i + b_{i + 1, \RT} s_{i,\LF}\cdot p_{i+1} + b_{i,\LF} s_{i+1,\RT} \cdot p_{i+1} + b_{i + 1,\RT} s_{i+1,\RT}\cdot p_{i+1}
        \end{split}
    \end{align}

    The last special case is when $i = j = n$. 
\begin{align}\label{eq:FullVeri3}
    \begin{split}        
        \E_{S\sim F_S\atop B\sim F_B}\left[\max(S,B)\cdot \indic[S\geq p_n] \indic[B\geq p_n]\right] &{\leq} \E_{S\sim F_S\atop B\sim F_B}\left[B S/ p_n\cdot \indic[S\geq p_n] \indic[B\geq p_n]\right]\\
         &= p_n \cdot \left(\E_{S\sim F_S}\left[S \cdot \indic[S\geq p_n]\right]/p_n\right) \cdot \left(\E_{S\sim F_B}\left[{B}\cdot \indic[B\geq p_n]\right]/p_n\right) \\
         &= s_{n,\LF} \cdot b_{n,\LF} \cdot p_n.
    \end{split}
    \end{align} Combining inequality~\eqref{eq:FullVeri1}, \eqref{eq:FullVeri2} and \eqref{eq:FullVeri3}, we have
    \begin{align*}
        \begin{split}
            \sum_{i=1}^n \sum_{j=1}^n s_i b_j \cdot \max(p_i,p_j) &\geq \sum_{i=1}^n \sum_{j=1}^n \E_{S\sim F_S\atop B\sim F_B}\Big[\max(S, B)\indic[S\in [p_{i}, p_{i+1})]\indic[B\in [p_j, p_{j+1})]\Big]  \\
                &= \E_{S\sim F_S\atop B\sim F_B}\left[\max(S, B)\right],
        \end{split}
    \end{align*}
   so inequality~\eqref{eq:discretization5} is satisfied.
    
Finally, we are only left to show that property~\eqref{eq:discretization7} holds. For any $t\in [n]$, it follows that
\begin{small}
	    \begin{align}
        \begin{split}
            \ALGW(\I, p_t) & = \E_{S\sim F_S}[S] + \E_{S\sim F_S\atop B\sim F_B}\left[(B - S)\cdot \indic[S\leq p_t\leq B]\right]\\
                            & \geq \E_{S\sim F_S}[S] + \E_{S\sim F_S\atop B\sim F_B}\left[(B - S)\cdot \indic[S< p_t\leq B]\right]\\
                            &= \sum_{i=1}^n E_{s,i} + \E_{B\sim F_B}[B\cdot \indic[B \geq p_t]]\cdot \Pr_{S\sim F_S}[S< p_t] - \E_{S\sim F_S}[S\cdot \indic[S < p_t]]\cdot \Pr_{B\sim F_B}[B\geq p_t]\\
& \geq \sum_{i=1}^n s_i \cdot p_i + \left(\sum_{j=t+1}^n b_j\cdot p_j + b_{t,\LF}\cdot p_t\right)\cdot \left(\sum_{i=1}^{t-1} s_i + s_{t,\RT}\right) \\
       & \quad\quad\quad\quad\quad~ - \left(\sum_{i=1}^{t-1} s_i\cdot p_i + s_{t,\RT}\cdot p_t\right)\cdot \left(\sum_{j=t + 1}^n b_j + b_{t,\LF}\right)\\
       & = \sum_{i=1}^n s_i \cdot p_i + \sum_{i=1}^{t-1} \sum_{j=t + 1}^n s_i b_j (p_j - p_i) + \sum_{j=t+1}^n b_j\cdot s_{t,\RT}\cdot (p_j - p_t) + \sum_{i=1}^{t-1} s_i \cdot b_{t,\LF}\cdot (p_t - p_i)\\
       & \geq \sum_{i=1}^n s_i \cdot p_i + \sum_{i=1}^{t-1} \sum_{j=t+1}^n s_i b_j (p_j - p_i).
        \end{split}
    \end{align}
    \end{small}where the second inequality follows from the fact that 
    \[\Pr[B\geq p_t] = \sum_{j=t}^{n-1} (b_{j,\LF} + b_{j+1,\RT}) + \Pr[B\geq p_n]{\leq \sum_{j=t}^{n-1} (b_{j,\LF} + b_{j+1,\RT}) +b_{n,\LF}}\leq \sum_{j=t+1}^n b_j + b_{t,\LF} \]

    Therefore, we could see that inequality~\eqref{eq:discretization7} holds. This finishes our proof.
\end{proof}

With the lemma above, we are ready to give the proof of Lemma~\ref{lem:fullinfolower}.

    Consider the following fixed-price mechanism: Given any instance $\I = (F_S, F_B)$, we first compute the optimal welfare of the instance. Suppose $\OPTW(\I) = c$, we choose the fixed price $p^*$ from $\{c p_1,\cdots, c p_n\}$ to maximizes the welfare, i.e., $p^* \in \arg\max_{p\in \{c p_1,\cdots,cp_n\}} \ALGW(\I, p)$. In the following, we show that this mechanism is an $r^{*}$-approximation to the optimal welfare.

   Note that the approximation ratio of our mechanism is independent of $c$.\footnote{The price $p^*$ depends on $c$, but the approximation ratio to the optimal welfare does not.} To keep our analysis clean, we first assume that the instance $\I = (F_S,F_B)$ has optimal welfare $1$. The approximation ratio of our mechanism could be written as
    \begin{align*}
        \begin{split}
            \min_{\I = (F_S,F_B)\atop \OPTW(\I) = 1} \max_{p\in \{p_1,p_2,\cdots, p_n\}} \ALGW(\I, p).
        \end{split}
    \end{align*}

    Next, we argue that for any instance $\I = (F_S, F_B)$ satisfying $\OPTW(\I) = 1$, there exists a valid solution $(s_1,\cdots,s_n,b_1,\cdots,b_n,r)$ of $\mathsf{LowerOp}$ such that $r \le \max_{p\in \{p_1,p_2,\cdots,p_n\}} \ALGW(\I, p)$. This immediately implies that $r^{*}$ is a lower bound of the approximation ratio.
    
	Given an instance $\I = (F_S, F_B)$ s.t. $\OPTW(\I) = 1$, the solution $(s_1,\cdots,s_n,b_1,\cdots,b_n,r)$ is constructed as follows. Let $(s_1,s_2,\cdots, s_n,b_1,b_2,\cdots, b_n)$ be the set of numbers that satisfies all the properties stated in Lemma~\ref{lem:discretization}. Let $r$ be $$\max_{t\in [n]} \sum_{i=1}^n s_i p_i + \sum_{i=1}^{t-1} \sum_{j=t+1}^{n} s_i b_j (p_j-p_i).$$

    We first verify that $(\{s_i\}_{i\in [n]}, \{b_i\}_{i\in [n]}, r)$ is a valid solution of $\mathsf{LowerOp}$. Notice that $\E[S]\leq \E[\max(S,B)] = 1$ and $\E[B] \leq \E[\max(B,S)] = 1$. Therefore, constraints~\eqref{eq:LowerOp1}, \eqref{eq:LowerOp2} and \eqref{eq:LowerOp3} directly follows from inequality~\eqref{eq:discretization1} and \eqref{eq:discretization2}. What's more, we could see \eqref{eq:LowerOp5} holds by the definition of $r$.
    
    Now by property~\eqref{eq:discretization5}, we have
    \begin{align*}
        \begin{split}
            \sum_{i=1}^n \sum_{j=1}^n s_i b_j \cdot \max(p_i,p_j) \geq \sum_{i=1}^n \sum_{j=1}^n \E_{S\sim F_S\atop B\sim F_B}\Big[\max(S, B)\indic[S\in [p_{i}, p_{i+1})]\indic[B\in [p_j, p_{j+1})]\Big]  = 1,
        \end{split}
    \end{align*}
   so constraint~\eqref{eq:LowerOp4} is satisfied.
    
Finally, we are only left to show that the best price in $\{p_1,p_2,\cdots, p_k\}$ must obtain an approximation ratio that is at least $r$ on instance $\I$, i.e., $r \leq \max_{p\in \{p_1,p_2,\cdots, p_k\}} \ALGW(\I, p)$.  Inequality~\eqref{eq:discretization7} states that
\[\ALGW(\I, p_t)\geq \sum_{i=1}^n s_i p_i+\sum_{i=1}^{t - 1} \sum_{j=t+1}^{n} s_i b_j (p_j - p_i).\]

Taking maximum over $t\in [n]$, we then get that 
	\[r = \max_{t\in [n]}\sum_{i=1}^n s_i \cdot p_i + \sum_{i=1}^{t-1} \sum_{j=t+1}^n s_i b_j (p_j - p_i) \leq \max_{t\in [n]} \ALGW(\I, p_t)\]
	
	which finishes our proof.

\subsection{Proof of Lemma~\ref{lem:fullinfoupper}}
\label{appendix:prooffullinfoupper}
In the following, we complete the proof of Lemma~\ref{lem:fullinfoupper}.

\begin{center}
\begin{tabular}{|c|}
\hline 
The Optimization Problem $\mathsf{UpperOp}$ \\
\hline
\parbox{15cm}{
\begin{align}
\min_{s_1,s_2\cdots, s_n\atop b_1,b_2,\cdots,b_n,r} \quad r \nonumber \\
\textsf{s.t.} \quad  & s_i, b_i \geq 0 & \forall i \in [n]\nonumber\\
& \sum_{i=1}^n s_i = 1 \quad \text{ and } \quad \sum_{i=1}^n b_i = 1 \nonumber\\
& \sum_{i=1}^n \sum_{j=1}^n s_i b_j \max(p_i,p_j) \geq 1 & \nonumber    \\
& \sum_{i=1}^n s_i p_i+\sum_{i=1}^{t} \sum_{j=t+1}^{n} s_i b_j (p_j - p_i) \leq r & \forall t \in [n] \nonumber
\end{align}
}\\
\hline 
\end{tabular}
\end{center}

	For any fixed support $0 = p_1 < p_2 < \cdots < p_n$ and a valid solution $(s_1, s_2, \cdots, s_n, b_1, b_2,\cdots, b_n, r)$, define an instance $\I = (F_S, F_B)$ satisfying
 
\begin{align*}
\begin{split}
S \sim F_S, S=\left\{
\begin{aligned}
&p_1 + \varepsilon  \quad & w.p. \quad s_1\\
&p_2 + \varepsilon  \quad & w.p. \quad s_2\\
& \cdots & \\
&p_n + \varepsilon  \quad & w.p. \quad s_n
\end{aligned}
\right.
\quad\quad\quad
B \sim F_B, B=\left\{
\begin{aligned}
&p_1   \quad & w.p. \quad b_1\\
&p_2   \quad & w.p. \quad b_2\\
& \cdots & \\
&p_n   \quad & w.p. \quad b_n
\end{aligned}
\right.
\end{split}
\end{align*} where $\varepsilon > 0$ is a constant that small enough.

It is easy to see that both $F_S$ and $F_B$ are valid distributions since the $\mathsf{UpperOp}$ requires the non-negativity of $s_i, b_i$ and $\sum_{i=1}^n s_i = \sum_{i=1}^n b_i = 1$. Next, we aim to show that no fixed-price mechanism have an approximation ratio of $r + \varepsilon$ on this instance $\I = (F_S, F_B)$. For any $x\in \mathbb{R}_{\ge 0}$, we could first see that $x < {\varepsilon}$ would never be a optimal price. Thus let $p_i$ be the largest $p\in \{p_1,p_2,\cdots, p_n\}$ that is not greater than $x - \varepsilon $. Notice that both $F_S$ is a distribution on support $\{p_i + \varepsilon\}_{i\in [n]}$ and $F_B$ is a discrete distribution on support $\{p_i\}_{i\in [n]}$. This means choosing $p_i + \varepsilon$ instead of $x$ would never become worse. Therefore, we could see that the optimal fixed-price mechanism on this instance is simply choosing one $p_t\in \{p_1,p_2,\cdots, p_k\}$ that maximizes $\ALGW(\I, p_t + \varepsilon )$. Again, by the fact that $F_S$ and $F_B$ are discrete distributions, $\ALGW(\I, p_t + \varepsilon )$ could be written as:
\begin{align*}
    \begin{split}
        \ALGW(\I, p_t + \varepsilon ) &= \sum_{i=1}^n (p_i + \varepsilon) s_i + \sum_{i=1}^t \sum_{j=t+1}^n s_i b_j (p_j - p_i - \varepsilon) \\& \leq \sum_{i=1}^n p_i s_i + \sum_{i=1}^t \sum_{j=t+1}^n s_i b_j (p_j - p_i) + \varepsilon.  
    \end{split}
\end{align*}

Also notice that the constraints of $\mathsf{UpperOp}$ guarantee that 
\begin{align*}
    \begin{split}
        \OPTW(\I) = \sum_{i=1}^n\sum_{j=1}^n s_i b_j \max(p_i + \varepsilon,p_j)  \geq  \sum_{i=1}^n\sum_{j=1}^n s_i b_j \max(p_i,p_j) \geq 1.
    \end{split}
\end{align*}

Therefore, the approximation ratio of the optimal fixed-price mechanism on instance $\I = (F_S, F_B)$ is upper bounded by
\begin{align*}
    \begin{split}
        \frac{\max_{t\in [n]} \ALGW(\I, p_t + \varepsilon)}{\OPTW(\I)} \le \max_{t\in [n]} \sum_{i=1}^n p_i s_i + \sum_{i=1}^t \sum_{j=t}^n s_i b_j (p_j - p_i) + \varepsilon = r + \varepsilon.
    \end{split}
\end{align*}

And this finishes our proof.

\subsection{Proof of Lemma~\ref{lem:optconvergence}}
\label{appendix:proofconvergence}

In this section, we assume that $\varepsilon > 0$ is a small enough constant such that $\varepsilon^2 \ll \varepsilon$.

We first show that, for any $\varepsilon > 0$, there exists a set of support $\{p_1,p_2,\cdots, p_n\}$ such that $\mathsf{UpperOp}$ has an optimal value of at most  $r^{*} + \varepsilon$.

As we show before, we could assume that the instance $\I = (F_S, F_B)$ has optimal welfare $1$. Thus, the approximation ratio of the optimal fixed-price mechanism is
\begin{align*}
    \begin{split}
      r^{*} = \min_{\I = (F_S, F_B)\atop \OPTW(\I) = 1} \max_{p\in \mathbb{R}} {\ALGW(\I, p)}.
    \end{split}
\end{align*}

Suppose $r^{*}$ is attained at $\I^{*} = (F_S^*, F_B^*)$. Now define $n = 1/\varepsilon^4$, and $p_i = i\cdot (1/\varepsilon^2) + \varepsilon / 2$  for $i\in [n+1]$. Our idea is to construct a valid solution $\{s_i, b_i\}_{i\in [n+1]}$ by rounding up $\I^{*}$ to $p_i$ and show that this solution has an objective value that is close to $r^*$.

Suppose $p_0 = 0$. Now we define \[s_i = \Pr_{S\sim F_S^*}\Big[S\in [(i-1) \varepsilon^2, i \varepsilon^2)\Big] \text{ and } b_i = \Pr_{B\sim F_B^*}\Big[B\in [(i-1) \varepsilon^2, i \varepsilon^2)\Big]\]
for $i\in [n]$. Especially, let
\[s_{n+1} = \E_{S\sim F_S^*}\Big[S\cdot \indic[S \geq n\varepsilon^2]\Big]/\left(n\varepsilon^2\right) \text{ and } b_{n+1} = \E_{B\sim F_B^*}\Big[B\cdot \indic[B \geq n\varepsilon^2]\Big]/\left(n\varepsilon^2\right).\]

Since $\E[S]$ and $\E[B]$ are upper bounded by $1$, we could see that $s_{n+1},b_{n+1}\leq \varepsilon^2$. In the last, let $s = \sum_{i=1}^{n+1} s_i$ and $b = \sum_{i=1}^{n+1}$ be the normalization factors. It's also straightforward to see that $s\leq \sum_{i=1}^n s_i + s_{n+1} \leq 1 + \varepsilon^2$. Following the same argument, it also holds that $b\leq 1 + \varepsilon^2$. Now define \[r = \max_{t\in [n]} \sum_{i=1}^n (s_i/s)p_i + \sum_{i=1}^t\sum_{j=t+1}^n (s_i/s) (b_j/b)(p_j-p_i).\] We aim to verify that $(s_1/s,s_2/s,\cdots,s_{n+1}/s,b_1/b,b_2/b,\cdots,b_{n+1}/b,r)$ is a valid solution of $\mathsf{UpperOp}$.

It is easy to see the non-negativity of $s_i,b_i$ and $\sum_{i=1}^{n+1} s_i/s = \sum_{i=1}^{n+1} b_i/b = 1$. What's more, from the definition of $r$, we could see the last constraint holds. Now we only need to check the third constraint. For any $i,j\in [n]$, it holds that
\begin{align*}
    &\E_{S\sim F_S^*\atop B\sim F_B^*}\Big[\max(S, B)\indic[S\in [(i-1) \varepsilon^2, i \varepsilon^2)]\indic[B\in [(j-1) \varepsilon^2, j \varepsilon^2)]\Big] \\
    & \leq (\max(p_i, p_j) - \varepsilon / 4) s_i b_j
\end{align*}

When one of $i,j$ equals to $n+1$(we can assume $i = n+1$ w.l.o.g.), it is true that
\begin{align*}
    &\E_{S\sim F_S^*\atop B\sim F_B^*}\Big[\max(S, B)\indic[S\geq n\varepsilon^2]\indic[B\in [(j-1) \varepsilon^2, j \varepsilon^2)]\Big] \\
    & = \E_{S\sim F_S^*}[S\cdot \indic[S\geq n\varepsilon^2]]\Pr_{B\sim F_B^*}[B\in [(j-1) \varepsilon^2, j \varepsilon^2)]\\
    & = \left(n\varepsilon^2\right) s_{n+1} b_j\\
    & \leq (p_{n+1} - \varepsilon / 4) s_i b_j
\end{align*}

Finally, for the special case that $i = j= n+1$, we could see that
\begin{align*}
    &\E_{S\sim F_S^*\atop B\sim F_B^*}\Big[\max(S, B)\indic[S\geq n\varepsilon^2]\indic[B\geq n\varepsilon^2]\Big] \\
    & \leq \E_{S\sim F_S^*\atop B\sim F_B^*}\Big[B S/\left(n\varepsilon^2\right)\cdot \indic[S\geq n\varepsilon^2]\indic[B\geq n\varepsilon^2]\Big]\\
    & = \left(n\varepsilon^2\right) s_{n+1} b_{n+1}\\
    & \leq (p_{n+1} - \varepsilon / 4) s_{n+1} b_{n+1}
\end{align*}

Summing up all the inequalities above, we then get that
\begin{align*}
    &\E_{S\sim F_S^*\atop B\sim F_B^*}[\max(S, B)]
     \leq \sum_{i=1}^{n+1} \sum_{j=1}^{n+1} \left(\max(p_i,p_j) - \varepsilon / 4\right) s_i b_j
\end{align*}

This implies that
\begin{align*}
    &\sum_{i=1}^{n+1}\sum_{j=1}^{n+1} (s_i/s) (b_j/b) \max(p_i,p_j)\\
    &\geq \sum_{i=1}^{n+1}\sum_{j=1}^{n+1} s_i b_j \max(p_i, p_j) \cdot \left(1 + \varepsilon^2\right)^{-2}\\
    & \geq \left(\E_{S\sim F_S^*\atop B\sim F_B^*}[\max(S, B)] +\sum_{i=1}^{n+1}\sum_{j=1}^{n+1} \varepsilon/4 s_ib_j\right)\cdot (1 - \varepsilon^2)^2\\
    & \geq 1 + \varepsilon/8.
\end{align*}
which means that $(s_1/s,s_2/s,\cdots,s_{n+1}/s,b_1/b,b_2/b,\cdots,b_{n+1}/b,r)$ is truly a valid solution. 

Next, we give an upper bound of $r$. To start with, notice that
\begin{align}
\label{eq:upperboundES}
\begin{split}
    \sum_{i=1}^{n+1} s_i p_i &= \sum_{i=1}^n \Pr\Big[S\in [(i-1) \varepsilon^2, i \varepsilon^2)\Big]\cdot \left((i-1)\varepsilon + \varepsilon^2 + \varepsilon /4 \right) + \E[S\cdot \indic[S\geq n\varepsilon^2]] + s_{n+1}\cdot(\varepsilon^2 + \varepsilon/4)\\
    & \leq \sum_{i=1}^n \E\Big[S\cdot\indic\left[S\in [(i-1) \varepsilon^2, i \varepsilon^2)\right] \Big] + \E[S\cdot \indic[S\geq n\varepsilon^2]] + \sum_{i=1}^{n+1} s_i (\varepsilon^2 + \varepsilon/4)\\
    & \leq \E[S] + \varepsilon/2.
\end{split}
\end{align}

For the term of gain from trade, it holds that 
\begin{align}
\label{eq:upperboundGFT}
\begin{split}
    \sum_{i=1}^t \sum_{j=t+1 }^{n+1} s_i b_j (p_j - p_i) &= \sum_{i=1}^t \sum_{j=t+1}^n s_i b_j(j\varepsilon^2-i\varepsilon^2) + \sum_{i=1}^t s_i b_{n+1} \left(n\varepsilon^2 - (i-1)\varepsilon^2\right)\\
    & \leq \sum_{i=1}^t \sum_{j=t+1}^n s_i b_j((j-1)\varepsilon^2-i\varepsilon^2) + \sum_{i=1}^t s_i b_{n+1} \left(n\varepsilon^2 - i\varepsilon^2\right) + \varepsilon^2 (1 + \varepsilon^2)^2\\
    & \leq \sum_{i=1}^t \sum_{j=t+1}^{n} \E[(B-S)\cdot\indic[S\in[(i-1)\varepsilon^2,i\varepsilon^2)]]\indic[B\in[(j-1)\varepsilon^2,j\varepsilon^2]]\\
    & \quad + \E[(B-S)\cdot \indic[S\in[0,t\cdot\varepsilon^2)]\indic[B\geq t\cdot \varepsilon^2]] + \varepsilon / 2\\
    & \leq  \E[(B-S)\cdot \indic[S\in[0,t\cdot\varepsilon^2)]\indic[B\geq t\cdot \varepsilon^2]] + \varepsilon / 2.
\end{split}
\end{align}

Combining~\eqref{eq:upperboundES} and \eqref{eq:upperboundGFT}, we know that for any $t\in [n + 1]$, 
\begin{align*}
    \begin{split}
        &\sum_{i=1}^{n+1} (s_i/s)p_i + \sum_{i=1}^t\sum_{j=t+1}^{n+1} (s_i/s) (b_j/b)(p_j-p_i)\\
        & \leq \sum_{i=1}^{n+1} s_i p_i+ \sum_{i=1}^t \sum_{j=t+1}^{n+1} s_i b_j (p_j - p_i)\\
& = \E[S] + \E[(B-S)\cdot \indic[S\in[0,t\cdot\varepsilon^2)]\indic[B\geq t\cdot \varepsilon^2]] + \varepsilon\\
        & \leq \ALGW(\I, t\cdot \varepsilon^2) + \varepsilon.
    \end{split}
\end{align*}

Taking maximum over $[n + 1]$, we then get that
\begin{align*}
    r = \max_{t\in [n+1]}\sum_{i=1}^{n+1} (s_i/s)p_i + \sum_{i=1}^t\sum_{j=t}^{n+1} (s_i/s) (b_j/b)(p_j-p_i) \leq \max_{t\in [n+1]}\ALGW(\I, t\cdot \varepsilon^2) + \varepsilon \leq r^* +\varepsilon.
\end{align*}

This means that the optimal value of $\mathsf{UpperOp}$ with respect to $\{p_i\}_{i\in[n+1]}$ is at most $r^*+\varepsilon$, and this finishes our proof.

Next, we aim to show that for any $\varepsilon > 0$, there exists $0 = p_1 < p_2 <\cdots p_n$ such that $\mathsf{LowerOp}$ has an optimal value of at least $r^{*} - \varepsilon$. We first present the following lemma that supports our proof.

\begin{lemma}\label{lem:generateinstance}
For any small enough constant $\varepsilon > 0$, suppose there exists a set of supports $\{p_i\}_{i\in [n]}$ with $\{s_i\}_{i\in [n]}$ and $\{b_{i}\}_{i\in [n]}$ that satisfies the subsequent conditions:
\begin{itemize}
	\item For all $i\in [n - 1]$, $p_{i+1}-p_i \leq \varepsilon^3$.
	\item $1\leq \sum_{i=1}^n s_i, \sum_{i=1}^n b_i \leq 1 + \varepsilon^3$.
	\item $\sum_{i=1}^n \sum_{j=1}^n \max(p_i,p_j) s_i b_j \geq 1$.
\end{itemize}

Let $r$ be defined as $\frac{\max_{t\in [n]} {\sum_{i=1}^n s_i p_i+\sum_{i=1}^{t - 1} \sum_{j=t+1}^{n} s_i b_j (p_j - p_i)}}{\sum_{i=1}^n \sum_{j=1}^n \max(p_i,p_j) s_i b_j }$.
Then, an instance $\I$ exists for which no fixed-price mechanism attains a welfare exceeding $r + \varepsilon$ times the optimal welfare; that is,
\[
    \frac{\max_{p\in \mathbb{R}}\ALGW(\I, p)}{\OPTW(\I)}  \leq r+\varepsilon.
\]
\end{lemma}

Before proving Lemma~\ref{lem:generateinstance}, we first illustrate how this completes our proof. Define $n = \left\lceil\varepsilon^{-6}\right\rceil + 1$, and $p_i = (i - 1)\cdot \varepsilon^3$ for $i\in [n]$. Let $(s_1,\cdots,s_n,b_1,\cdots,b_n,r')$ be the optimal solution of the optimization problem $\mathsf{LowerOp}$ with respect to $\{p_i\}_{i\in [n]}$. It is equivalent to show that there exists an instance $\I=(F_S, F_B)$ such that the optimal approximation ratio of $\I$, i.e. $\max_{x\in \mathbb{R}}\frac{\ALGW(\I, x)}{\OPTW(\I)}$, is at most $r' + \varepsilon$. Notice that the definition of $\{p_i\}_{i\in [n]}$ directly implies fulfillment of the first condition in Lemma~\ref{lem:generateinstance}. Furthermore, since it is a valid solution of $\mathsf{LowerOp}$, we can deduce that the second and the third conditions are satisfied. The optimality of $(s_1,\cdots,s_n,b_1,\cdots,b_n,r')$ implies that $r = \frac{\max_{t\in [n]} {\sum_{i=1}^n s_i p_i+\sum_{i=1}^{t - 1} \sum_{j=t+1}^{n} s_i b_j (p_j - p_i)}}{\sum_{i=1}^n \sum_{j=1}^n \max(p_i,p_j) s_i b_j }$ is at most $r'$. By applying Lemma~\ref{lem:generateinstance}, we confirm the existence of an instance exhibiting an approximation ratio no greater than $r' + \varepsilon$, thereby completing our proof of Lemma~\ref{lem:optconvergence}.

\begin{proof}[Proof of Lemma~\ref{lem:generateinstance}]

In this proof, we define $v_1 = \max_{t\in [n]} {\sum_{i=1}^n s_i p_i+\sum_{i=1}^{t - 1} \sum_{j=t+1}^{n} s_i b_j (p_j - p_i)}$, $v_2 = \sum_{i=1}^n \sum_{j=1}^n \max(p_i,p_j) s_i b_j$, and thus $r = \frac{v_1}{v_2}$. We construct the instance as follows. Let $n' = n + \left\lceil\frac{4}{\varepsilon}\right\rceil$, and $s_i = b_i = 0$ for $n< i\le n'$. Now define $\{s'_i\}$ where 
\[s_{j}' = \sum_{i=\max\left(1, j - \left\lceil\frac{4}{\varepsilon}\right\rceil+1\right)}^{j} s_i / \left\lceil\frac{4}{\varepsilon}\right\rceil.\]

It follows that 
\[\sum_{j=1}^{n'} s'_{j} = \sum_{j = 1}^n\sum_{i=\max\left(1, j - \left\lceil\frac{4}{\varepsilon}\right\rceil+1\right)}^{j} s_i / \left\lceil\frac{4}{\varepsilon}\right\rceil = \sum_{i=1}^{n'} s_i.\]

Let $s = \sum_{i=1}^{n'} s'_i$ and $b = \sum_{i=1}^{n'} b_i$ be the normalization factors. The second condition guarantees that $1\leq s,b\leq 1 + \varepsilon^3$. What's more, we could also see that
\begin{align*}
    s_j' = \sum_{i=\max\left(1, j - \left\lceil\frac{4}{\varepsilon}\right\rceil+1\right)}^{j} s_i / \left\lceil\frac{4}{\varepsilon}\right\rceil \leq (1 + \varepsilon^3) \left\lceil\frac{4}{\varepsilon}\right\rceil \leq \varepsilon /3 .
\end{align*} holds for all $j\in [n']$.

Consider the following instance $\I = (F_S, F_B)$:

\begin{align*}
\begin{split}
    		 S \sim F_{S}, S=\left\{
\begin{aligned}
&p_1 + \varepsilon^4   \quad & w.p. \quad s'_1 / s\\
& \cdots & \\
&p_{n'} + \varepsilon^4 & w.p.\quad s'_{n'} / s
\end{aligned}
\right.
\quad\quad\quad
B \sim F_{B}, B=\left\{
\begin{aligned}
&p_1    \quad & w.p. \quad b_1 / b\\
& \cdots & \\
&p_{n'}  & w.p.\quad b_{n'} / b
\end{aligned}
\right.
\end{split}
\end{align*}

First, it is straight forward to verify this is a valid distribution. We first calculate $\OPTW(\I)$:
\begin{align*}
    \begin{split}
        \OPTW(\I) &= \sum_{i=1}^{n'} \sum_{j=1}^{n'} \max(p_i + \varepsilon^4, p_j)\cdot (s_i'/s) (b_j/b)\\
        & \geq \sum_{i=1}^{n'} \sum_{j=1}^{n'} \max(p_i, p_j)\cdot (s_i'/s) (b_j/b) - \varepsilon^4\\
        & \geq \sum_{i=1}^{n'} \sum_{j=1}^{n'} \max(p_i, p_j)\left(\sum_{k=\max\left(1, i - \left\lceil\frac{4}{\varepsilon}\right\rceil+1\right)}^{i} \left(s_k / \left\lceil\frac{4}{\varepsilon}\right\rceil\right)/s\right) \cdot (b_j / b)- \varepsilon^4\\
        & \geq \sum_{i=1}^{n'} \sum_{j=1}^{n'} \max(p_i, p_j)\cdot s_i b_j / (bs) - \varepsilon^4\\
& \geq \left(1 - \varepsilon^2\right)\cdot  v_2
           \end{split}
\end{align*}where the last inequality follows from $v_2 \geq 1$.

Now consider the optimal fixed-price mechanism for the instance. As we have shown in the proof of Lemma~\ref{lem:fullinfoupper}, the optimal mechanism only need to choose price from the support of the discrete distribution. This implies that
\begin{align*}
    \begin{split}
        \max_{p\in \mathbb{R}} \ALGW(\I, p) = \max_{t\in [n]}\sum_{i=1}^{n'} (p_i + \varepsilon^4) (s_i' / s) + \sum_{i=1}^t \sum_{j=t+1}^{n'} (s_i'/s) (b_j/b) (p_j - p_i - \varepsilon^4)
    \end{split}
\end{align*}

For any $t\in [n]$, one could see that
\begin{align}
\label{eq:LowerBoundES}
    \begin{split}
        \sum_{i=1}^{n'} p_i s'_i &= \sum_{i=1}^{n'}\left(\sum_{k=\max\left(1, i - \left\lceil\frac{4}{\varepsilon}\right\rceil+1\right)}^{i} \left(s_k / \left\lceil\frac{4}{\varepsilon}\right\rceil\right)\right) p_i\\
        &\leq \sum_{i=1}^{n'}\left(\sum_{k=\max\left(1, i - \left\lceil\frac{4}{\varepsilon}\right\rceil+1\right)}^{i} \left(s_k / \left\lceil\frac{4}{\varepsilon}\right\rceil\right)\cdot \left(p_k + \left\lceil\frac{4}{\varepsilon}\right\rceil \cdot  \varepsilon^3 \right)\right)\\
        & \leq \sum_{i=1}^{n'}  (p_i + 5\varepsilon^2) s_i
    \end{split}
\end{align}where the first inequality is because that the gap between any $p_i$ and $p_{i+1}$ is at most $\varepsilon^3$.

For the term of gain from trade, it follows that 
\begin{align}
\label{eq:LowerBoundGFT}
    \begin{split}
    \sum_{i=1}^t \sum_{j=t+1}^{n'} s_i' b_j (p_j - p_i) &= \sum_{i=1}^{t-1} \sum_{j=t+1}^{n'} s_i' b_j (p_j - p_i) + \sum_{j=t+1}^{n'} s_{t}' b_j (p_j - p_i)\\
   & \leq  \sum_{i=1}^{t-1} \sum_{j=t+1}^{n'} \left(\sum_{k=\max\left(1, i - \left\lceil\frac{4}{\varepsilon}\right\rceil+1\right)}^{i} \left(s_k / \left\lceil\frac{4}{\varepsilon}\right\rceil\right)\right) b_j (p_j - p_i) + s_t' \sum_{j=1}^{n'} b_j p_j\\
   & \leq \sum_{i=1}^{t-1}\sum_{j=t+1}^{n'} s_i b_j (p_j - p_i) + \varepsilon / 3 \cdot v_2.
    \end{split}
\end{align}where we use the fact that $s_k b_j (p_j - p_i) \leq s_k b_j (p_j - p_k)$ for $j > i \geq k$, $s_t' \leq \varepsilon /3$, and $\sum_{j=1}^{n'} b_j p_j \leq \sum_{i=1}^{n'} \sum_{j=1}^{n'}\\ \max(p_i,p_j) s_i b_j = v_2$.

Again by combining the two inequalities above, we know that
\begin{align*}
    \begin{split}
        &\sum_{i=1}^{n'} (p_i + \varepsilon^4) (s_i' / s) + \sum_{i=1}^t \sum_{j=t+1}^{n'} (s_i'/s) (b_j/b) (p_j - p_i - \varepsilon^4)\\
        &\leq \sum_{i=1}^{n'} p_i s_i' + \sum_{i=1}^t \sum_{j=t+1}^{n'} s_i' b_j (p_j - p_i) + \varepsilon^4\\
& \leq \sum_{i=1}^{n'}  (p_i + 5\varepsilon^2) s_i + \sum_{i=1}^{t-1}\sum_{j=t+1}^{n'} s_i b_j (p_j - p_i) + \varepsilon / 3\cdot v_2 + \varepsilon^4\\
& \leq \sum_{i=1}^{n'} p_i s_i + \sum_{i=1}^{t-1} \sum_{j=t+1}^{n'} s_i b_j (p_j - p_i) + \varepsilon / 2 \cdot v_2
    \end{split}
\end{align*}where we apply~\eqref{eq:LowerBoundES} and \eqref{eq:LowerBoundGFT} in the second inequality. The last inequality holds since $v_2 \geq 1$, and $\sum_{i=1}^{n'} s_i \leq 1 + \varepsilon^3$.

Now taking the maximum over $t\in [n']$, we then get that 
\begin{align*}
    \max_{p\in \mathbb{R}} \ALGW(\I, p) &= \max_{t\in [n']}\sum_{i=1}^{n'} (p_i + \varepsilon^4) (s_i' / s) + \sum_{i=1}^t \sum_{j=t+1}^{n'} (s_i'/s) (b_j/b) (p_j - p_i - \varepsilon^4)\\
    & \leq \max_{t\in [n']}\sum_{i=1}^{n'} p_i s_i + \sum_{i=1}^{t-1} \sum_{j=t+1}^{n'} s_i b_j (p_j - p_i) + \varepsilon / 2\cdot v_2\\
    & = \max_{t\in [n]}\sum_{i=1}^{n} p_i s_i + \sum_{i=1}^{t-1} \sum_{j=t+1}^{n} s_i b_j (p_j - p_i) + \varepsilon / 2\cdot v_2\\
    & = v_1 + \frac{\varepsilon}{2}\cdot v_2.
\end{align*}where the second equation follows from $s_i = b_i  = 0$ when $i > n$.

Therefore, on instance $\I$, it holds that
\begin{align*}
    \frac{\max_{p\in \mathbb{R}}\ALGW(\I, p)}{\OPTW(\I)} \leq\frac{v_1 + \frac{\varepsilon}{2}\cdot v_2}{(1-\varepsilon^2)\cdot v_2} \leq r+\varepsilon.
\end{align*}
\end{proof}

\section{Proof of Theorem~\ref{thm:one-sided}}
\label{subsec:one-sided-appendix}

As discussed in Section~\ref{subsec:one-sided}, our intention is to prove Theorem \ref{thm:one-sided} utilizing the minimax theorem. However, we are unable to directly apply the theorem due to the infinite-dimensional nature of the problem. Fortunately, with the assistance of the discretization lemma (Lemma \ref{lem:discretization}), we are able to transform the problem into a finite-dimensional one, and subsequently, prove Theorem \ref{thm:one-sided} by employing the minimax theorem. The detailed proof is presented below.

We let $r^{*}$ be the optimal approximation ratio in the full prior information setting. Our goal is to demonstrate that, given solely some one-sided information, i.e., only the distribution of the seller or the buyer, there still exists some fixed-price mechanism that obtains $r^{*}$ fraction of the optimal welfare.

We initiate the proof by considering the scenario in which only the buyer's distribution $F_B$ is known. Given the distribution $F_B$, we can once again assume that $\E_{B\sim F_B}[B] = 1$. For any sufficiently small constant $\varepsilon > 0$, let $n$ be defined as $\left\lceil \frac{1}{\varepsilon^8}\right\rceil + 1$, and consider the set of support $\{p_i\}_{i\in [n]}$, where $p_i = (i - 1) \cdot \varepsilon^4$. As stated in Lemma~\ref{lem:discretization}, the construction of $\{b_i\}_{i\in [n]}$ (or $\{s_i\}_{i\in [n]}$) relies exclusively on $F_B$ (or $F_S$), so we can first construct $\{b_i\}_{i\in [n]}$ according to \Cref{lem:discretization} with the $\{p_i\}_{i\in [n]}$ defined above.
Let us now consider the following min-max optimization problem $\mathsf{BuyerFull}$:
\begin{align}
 \min_{s_1,s_2\cdots, s_n} \max_{\omega_1,\omega_2,\cdots,\omega_n} \quad & \sum_{t=1}^n \omega_t \left(\sum_{i=1}^n s_i p_i+\sum_{i=1}^{t-1}\sum_{j=t+1}^n s_ib_j(p_j-p_i)\right) - (r^* - 2\varepsilon) \cdot \sum_{i=1}^n \sum_{j=1}^n \max(p_i,p_j)s_i b_j\nonumber \\
\textsf{s.t.} \quad  & s_i \geq 0 & \forall i \in [n]\nonumber\\
& \sum_{i=1}^{n - 1} s_i \leq  1 \quad \text{ and } \quad \sum_{i=1}^n s_i \geq 1 \quad \text{ and }\quad s_n \leq c\nonumber\\
& \omega_i \geq 0 & \forall i \in [n]\nonumber \nonumber    \\
& \sum_{i=1}^n \omega_i = 1 \nonumber
\end{align} Here we allow $c$ to be an arbitrary non-negative number. We enforce an upper bound of $s_n$ to make sure that $\{s_i\}_{i\in [n]}$ lie in a compact set. We demonstrate that this min-max optimization problem $\mathsf{BuyerFull}$ has a non-negative optimal value for any constant $c \geq 0$.

\begin{lemma}\label{lem:buyerfull}
The optimal objective value of $\mathsf{BuyerFull}$ is non-negative for any constant $c \geq 0$.
\end{lemma}

\begin{proof}
To prove the non-negativity for any constant $c \geq 0$, we simply eliminate the  constraint that $s_n \leq c$ and show that the program has a non-negative objective value even after dropping the upper bound constraint of $s_n$, as removing a constraint on $s_n$ can only decrease the objective value.

	We prove the lemma by way of contradiction. Given $\{s_i\}_{i\in [n]}$, it is clear that 
	\[\max_{\omega \in \Delta(n)} \sum_{t=1}^n \omega_t \left(\sum_{i=1}^n s_i p_i+\sum_{i=1}^{t-1}\sum_{j=t+1}^n s_ib_j(p_j-p_i)\right) = \max_{t\in [n]} \sum_{i=1}^n s_i p_i+\sum_{i=1}^{t-1}\sum_{j=t+1}^n s_ib_j(p_j-p_i),\]where $\Delta(n)$ is the probability simplex over $[n]$.
	
Suppose $\mathsf{BuyerFull}$ has a negative optimal objective value, then there exists $\{s^*_i\}_{i\in [n]}$ such that \begin{equation}\label{eq:buyerfullop1}
\max_{t\in [n]} \sum_{i=1}^n s^*_i p_i+\sum_{i=1}^{t-1}\sum_{j=t+1}^n s^*_ib_j(p_j-p_i) < (r^* - 2\varepsilon) \cdot \sum_{i=1}^n \sum_{j=1}^n \max(p_i,p_j) s^*_i b_j.
\end{equation}
	
Notice that $\{b_i\}_{i\in [n]}$ is defined according to Lemma~\ref{lem:discretization}, which ensures that \[1 \leq \sum_{i = 1}^n b_i \leq 1 + \frac{\E[B]}{p_n} \leq 1 + \varepsilon^3. \]
	
Let us examine $s^*_n$ first. Note that 
\[\max_{t\in [n]} \sum_{i=1}^n s^*_i p_i+\sum_{i=1}^{t-1}\sum_{j=t+1}^n s^*_ib_j(p_j-p_i) \geq \sum_{i=1}^{n}s^*_i p_i,\]
and 
\[\sum_{i=1}^n \sum_{j=1}^n \max(p_i, p_j) s^*_i b_j \leq \sum_{i=1}^n s^*_ip_i + \sum_{i=1}^n b_ip_i = \sum_{i=1}^n s^*_ip_i + 1\]
	
If {$s^*_n \geq {\varepsilon^3}$}, notice that $\sum_{i=1}^n s^*_i p_i \geq s^*_n p_n \geq \frac{1}{\varepsilon}$ is sufficiently large, while $r^{*}$ does not exceed $0.75$. This means that if {$s^*_n \geq \varepsilon^3$}, \eqref{eq:buyerfullop1} could never happen. Thus, we get that $\sum_{i=1}^n s^*_i \leq \sum_{i=1}^{n-1}s^*_i + s^*_n \leq 1 + {\varepsilon^3}$ and {$\sum_{i=1}^n s^*_i \geq 1$}. Finally, it is easy to see that $\sum_{i=1}^n \sum_{j=1}^n \max(p_i,p_j) s^*_i b_j \geq \sum_{j=1}^n p_j b_j = 1$, and the gap between $p_i$ and $p_{i+1}$ is at most $\varepsilon^4$. Hence, we may invoke Lemma~\ref{lem:generateinstance}, which provides the existence of an instance $\mathcal{I}$ such that
\[ \frac{\max_{p\in \mathbb{R}}\ALGW(\I, p)}{\OPTW(\I)} \leq \frac{\max_{t\in [n]} \sum_{i=1}^n s^*_i p_i+\sum_{i=1}^{t-1}\sum_{j=t+1}^n s^*_ib_j(p_j-p_i)}{\sum_{i=1}^n \sum_{j=1}^n \max(p_i,p_j) s^*_i b_j} + \varepsilon < r^*- \varepsilon.\]

However, this contradicts the fact that $r^{*}$ is the optimal approximation ratio in the full prior information setting. Thus, we prove that the optimal value of $\mathsf{BuyerFull}$ is non-negative.
\end{proof}

Notice that the min-max optimization problem $\mathsf{BuyerFull}$ is bilinear with respect to $\{s_i\}_{i\in [n]}$ and $\{\omega_i\}_{i\in [n]}$. Furthermore, it is clear that the feasible solution spaces for both ${\omega}$ and ${s}$ are convex and compact. Consequently, by employing the minimax theorem, the following max-min optimization problem $\mathsf{BuyerOnly}$ possesses an identical optimal objective value to that of $\mathsf{BuyerFull}$, which is non-negative as we proved above.
\begin{align}
\max_{\omega_1,\omega_2,\cdots,\omega_n} \min_{s_1,s_2\cdots, s_n}  \quad & \sum_{t=1}^n \omega_t \left(\sum_{i=1}^n s_i p_i+\sum_{i=1}^{t-1}\sum_{j=t+1}^n s_ib_j(p_j-p_i)\right) - (r^* - 2\varepsilon) \cdot \sum_{i=1}^n \sum_{j=1}^n \max(p_i,p_j)s_i b_j\nonumber \\
\textsf{s.t.} \quad  & s_i \geq 0 & \forall i \in [n]\nonumber\\
& \sum_{i=1}^{n - 1} s_i \leq  1 \quad \text{ and } \quad \sum_{i=1}^n s_i \geq 1 \quad \text{ and } \quad s_n \leq c\nonumber\\
& \omega_i \geq 0 & \forall i \in [n]\nonumber \nonumber    \\
& \sum_{i=1}^n \omega_i = 1 \nonumber
\end{align}

Consider the following mechanism, given solely the buyer's distribution $F_B$. Without loss of generality we can assume that $\E_{B\sim F_B}[B] = 1$. We generate the discretized support sets $\{p_i\}_{i\in [n]}$ and $\{b_i\}_{i\in [n]}$ in the same manner as previously described. Subsequently, we solve the max-min optimization problem $\mathsf{BuyerOnly}$ with respect to $\{p_i\}_{i\in [n]}$ and $\{b_i\}_{i\in [n]}$ while setting $c$ to be $10$. Denote the optimal solution of {the max player of} this max-min optimization problem as $\left\{\omega_i^{*}\right\}_{i\in [n]}$. We then select $p_i$ as the price with probability $\omega_i^*$.

We now demonstrate that this mechanism attains an approximation ratio of no less than $(r^* - 2\varepsilon)$. We establish this by employing a proof by contradiction. Suppose there exists a seller's distribution $F_S$ so that our mechanism has an approximation ratio lower than $r^* - 2\varepsilon$ on the instance $\I = (F_S, F_B)$. We generate the set $\{s_i\}_{i\in [n]}$ as described in Lemma~\ref{lem:discretization} with respect to $\{p_i\}_{i\in [n]}$ and the distribution $F_S$ {, and we also choose $\{b_i\}_{i\in [n]}$ that correspond to the sets of numbers described in Lemma~\ref{lem:discretization}.} Therefore, the following holds
\begin{align}
\sum_{i=1}^n s_i \geq 1 \text{ and } \sum_{i=1}^{n - 1} s_i &\leq 1.\label{eq:one-sidedfaq} \\
    \sum_{i=1}^n s_i p_i+\sum_{i=1}^{t-1}\sum_{j=t+1}^n s_ib_j(p_j-p_i) &\leq \ALGW(\I, p_t)\label{eq:one-sided1} \\
    \sum_{i = 1}^n \sum_{j=1}^n \max(p_i, p_j) s_i b_j &\geq \OPTW(\I).\label{eq:one-sided2}
\end{align}

Now let us consider $s_n$. Suppose $s_n \geq 10$. Notice that Lemma~\ref{lem:discretization} demonstrates that  
\[\sum_{i = 1}^n s_i \leq 1 + \frac{\E[S]}{p_n}.\]
Combining this with the fact that $s_n \geq 10$, we obtain that 
\[\E[S] \geq (s_n - 1)\cdot p_n \geq 9\varepsilon^{-4}.\]

Therefore, on this instance, any fixed-price mechanism has an approximation ratio of at least
\[\frac{\E[S]}{\E[\max(S, B)]} \geq \frac{\E[S]}{\E[S] + \E[B]} \geq \frac{9\varepsilon^{-4}}{9\varepsilon^{-4} + 1} \geq r^{*}\]
where the last inequality holds since we know that $r^{*}$ is at most $0.75$, while $\varepsilon^{-4}$ is a large enough number. However, this contradicts the claim that the mechanism has an approximation lower than $r^{*} - 2\varepsilon$ on this instance, and thus we prove that $s_n \leq 10$.  Combining this inequality with~\eqref{eq:one-sidedfaq}, we demonstrate that $\{s_i\}_{i\in [n]}$ is a feasible solution to $\mathsf{BuyerOnly}$ when $c = 10$.

Given that our mechanism achieves a fraction less than $r^* - 2\varepsilon$ of the optimal welfare, it follows that
\begin{equation}
\sum_{t=1}^n \omega_t^* \ALGW(\I, p_t) < (r^* - 2\varepsilon) \OPTW(\I)   \label{eq:one-sided3}
\end{equation}

Combining~\eqref{eq:one-sided1},\eqref{eq:one-sided2} and \eqref{eq:one-sided3}, we have that 
\[\sum_{t=1}^n \omega_t^* \left(\sum_{i=1}^n s_i p_i+\sum_{i=1}^{t-1}\sum_{j=t+1}^n s_ib_j(p_j-p_i)\right) -(r^* - 2\varepsilon) \sum_{i = 1}^n \sum_{j=1}^n \max(p_i, p_j) s_i b_j < 0.\]

Observe that $\{\omega_i^*\}$ represents the optimal solution of the max-min optimization problem $\mathsf{BuyerOnly}$, and $\{s_i\}_{i\in [n]}$ is a feasible solution for the inner minimization problem with $c = 10$. This contradicts the assertion that $\mathsf{BuyerOnly}$ possesses a non-negative optimal objective value. Consequently, we deduce that the mechanism utilizing solely the buyer's distribution $F_B$ achieves an approximation ratio of $r^* - 2\varepsilon$. By allowing $\varepsilon \rightarrow 0$, we demonstrate that when considering only the buyer's distribution, the optimal approximation ratio is identical to the optimal ratio obtainable with full prior information.

Next,  we establish the case, where only the seller's distribution $F_S$ is known, by essentially the same argument. Without loss of generality, we assume that the seller's distribution $F_S$ satisfying $\E_{S\sim F_S} [S] = 1$. We define $n$ as $\left\lceil \frac{1}{\varepsilon^{20}} \right\rceil+ 1$ and $p_i = (i - 1)\cdot\varepsilon^{10}$ for some small enough constant $\varepsilon > 0$. We construct a set of numbers $\{s_i\}_{i\in [n]}$ as described in Lemma~\ref{lem:discretization}. Given $\{p_i\}_{i\in [n]}$ and $\{s_i\}_{i\in [n]}$, we define the following min-max optimization problem $\mathsf{SellerFull}$:
\begin{align}
 \min_{b_1,b_2\cdots, b_n} \max_{\omega_1,\omega_2,\cdots,\omega_n} \quad & \sum_{t=1}^n \omega_t \left(\sum_{i=1}^n s_i p_i+\sum_{i=1}^{t-1}\sum_{j=t+1}^n s_ib_j(p_j-p_i)\right) - (r^* - 2\varepsilon) \cdot \sum_{i=1}^n \sum_{j=1}^n \max(p_i,p_j)s_i b_j\nonumber \\
\textsf{s.t.} \quad  & b_i \geq 0 & \forall i \in [n]\nonumber\\
& \sum_{i=1}^{n - 1} b_i \leq  1 \quad \text{ and } \quad \sum_{i=1}^n b_i \geq 1  \quad \text{ and } \quad b_n \leq c\nonumber\\
& \omega_i \geq 0 & \forall i \in [n]\nonumber \nonumber    \\
& \sum_{i=1}^n \omega_i = 1 \nonumber
\end{align}
$c > 0$ is an arbitrary positive constant.

We first argue that $\mathsf{SellerFull}$ has a non-negative objective value.
\begin{lemma}\label{lem:SellerFull}
	The optimal value of $\mathsf{SellerFull}$ is non-negative for any constant $c > 0$.
\end{lemma}

\begin{proof}
	The proof is nearly identical to the proof of Lemma~\ref{lem:buyerfull}, while the only distinction is the argument used to upper bound the value of $b_n$. Again, we first drop the upper bound constraint of $b_n$ and prove the claim by way of contradiction. Suppose $\mathsf{SellerFull}$ has a negative optimal objective value, then there exists a set of $\{b^*_i\}_{i\in [n]}$ such that

\begin{equation}\label{eq:sellerfullop1}
\max_{t\in [n]} \sum_{i=1}^n s_i p_i+\sum_{i=1}^{t-1}\sum_{j=t+1}^n s_ib^*_j(p_j-p_i) < (r^* - 2\varepsilon) \cdot \sum_{i=1}^n \sum_{j=1}^n \max(p_i,p_j) s_i b^*_j.
\end{equation}
	
	Observe that $\{s_i\}_{i\in [n]}$ are generated according to Lemma~\ref{lem:discretization}, so we get that
	\[1\leq \sum_{i = 1}^n s_i \leq 1 + \frac{\E[S]}{p_n} \leq 1 + \varepsilon^3.\]
	
	Now let us consider $b^*_n$. Suppose $b^*_n \geq \varepsilon^3$. Let $t'$ be $\left\lceil\frac{1}{\varepsilon^{11}}\right\rceil$ + 2. From  property~\eqref{eq:discretizationNew} in Lemma~\ref{lem:discretization}, 
it is clear that 
	\[\sum_{i = 1}^{t' - 1} s_i \geq \Pr_{S\sim F_S}[S< p_{t' - 1}] = 1 - \Pr_{S\sim F_S}[S \geq p_{t' - 1}]\geq 1 - \varepsilon,\] where the last inequality follow from Markov's inequality as $p_{t'-1}\geq \frac{1}{\varepsilon}$.

	Therefore, it holds that
	\begin{align*}
		\sum_{i=1}^n s_i p_i+\sum_{i=1}^{t'-1}\sum_{j=t'+1}^n s_ib^*_j(p_j-p_i)
		& =1 + \sum_{j = t' + 1}^n b_j^*p_j\cdot \left(\sum_{i=1}^{t' - 1} s_i\right) - \sum_{i=1}^{t' - 1} s_i p_i\left(\sum_{j=t'+1}^n b^*_j\right)\\
		&\geq 1 + \sum_{j = t' + 1}^n b^*_jp_j\cdot \left(1 - \varepsilon\right) - \left(1 + b_n^*\right)\\ 
		& \geq \left(1 - \varepsilon\right)\cdot \sum_{j = t' + 1}^n b^*_jp_j - b_n^*\\
		& \geq \left(1 - 2\varepsilon\right) \cdot \sum_{j=t'+1}^n b_j^* p_j
	\end{align*} The first inequality follows from the fact that $\sum_{j=t'+1}^n b_j^* = \sum_{j=t'+1}^{n - 1} b_j^* + b_n^* \leq 1 + b_n^* $, and the last inequality is because that $p_n$ is no less than $\varepsilon^{-10}$, and thus is significantly larger than $1$.

  For the optimal welfare, we can upper bound it as follows.
	\begin{align*}
		\begin{split}
			\sum_{i=1}^n \sum_{j=1}^n \max(p_i,p_j) s_ib_j^* & \leq \sum_{i=1}^n p_i s_i + \sum_{j=1}^n p_j b_j^*\\
				& \leq 1 + \sum_{j=1}^{t'}p_{t'}b_j^* + \sum_{j=t'+1}^{n} p_j b_j^*\\
				& \leq 1 + 2\varepsilon^{-1} + \sum_{j = t' + 1}^n b_j^*p_j.
		\end{split}
	\end{align*}
	
	However, notice that $\sum_{j=t'+1}^n b^*_j p_j \geq b^*_n p_n \geq \varepsilon^{-7}$, which is significantly larger than $1 + 2\varepsilon^{-1}$, while $r^{*}$ is at most $0.75$. Thus, \eqref{eq:sellerfullop1} could never be true when $b^*_n \geq \frac{1}{\varepsilon^3}$, hence $1\leq \sum_{i = 1}^n b^*_i \leq 1 + \varepsilon^3$. Furthermore, it is also easy to see that $\sum_{i=1}^n \sum_{j = 1}^n \max(p_i,p_j) s_i b^*_j \geq \sum_{i=1}^n p_i s_i  = 1$, and  the gap between $p_{i + 1} - p_{i}\leq \varepsilon^3$ for any $i\in [n - 1]$. Consequently, we can apply Lemma~\ref{lem:generateinstance}, which establishes the existence of an instance $\mathcal{I}$ such that:
\[ \frac{\max_{p\in \mathbb{R}}\ALGW(\I, p)}{\OPTW(\I)} \leq \frac{\max_{t\in [n]} \sum_{i=1}^n s_i p_i+\sum_{i=1}^{t-1}\sum_{j=t+1}^n s_ib^*_j(p_j-p_i)}{\sum_{i=1}^n \sum_{j=1}^n \max(p_i,p_j) s_i b^*_j} + \varepsilon < r^*- \varepsilon.\]

However, this result contradicts the fact that $r^{*}$ represents the optimal approximation ratio within the full prior information setting, and thus demonstrates that the optimal value of $\mathsf{BuyerFull}$ must be non-negative.
	
\end{proof}

Given the fact that $\mathsf{BuyerFull}$ is non-negative, we are almost done with our proof. Similarly, by applying the Minimax theorem, we can deduce that the following max-min optimization problem $\mathsf{BuyerOnly}$ has the \emph{non-negative} identical objective value as $\mathsf{BuyerFull}$ for any constant $c > 0$.
\begin{align}
\max_{\omega_1,\omega_2,\cdots,\omega_n} \min_{b_1,b_2\cdots, b_n}  \quad & \sum_{t=1}^n \omega_t \left(\sum_{i=1}^n s_i p_i+\sum_{i=1}^{t-1}\sum_{j=t+1}^n s_ib_j(p_j-p_i)\right) - (r^* - 2\varepsilon) \cdot \sum_{i=1}^n \sum_{j=1}^n \max(p_i,p_j)s_i b_j\nonumber \\
\textsf{s.t.} \quad  & b_i \geq 0 & \forall i \in [n]\nonumber\\
& \sum_{i=1}^{n - 1} b_i \leq  1 \quad \text{ and } \quad \sum_{i=1}^n b_i \geq 1 \quad \text{ and }\quad  b_n \leq c \nonumber\\
& \omega_i \geq 0 & \forall i \in [n]\nonumber \nonumber    \\
& \sum_{i=1}^n \omega_i = 1 \nonumber
\end{align}

We next apply almost the same argument as the case  where the buyer's distribution is known. The only difference is how we deal with the case where $b_n > c$.

Similarly, let us consider the following mechanism where only the seller's distribution $F_S$ is known. Without loss of generality we can assume that $\E_{S\sim F_S}[S] = 1$,
and we use $F_S$ to generate the set $\{s_i\}_{i\in [n]}$ according to Lemma~\ref{lem:discretization}, and solve $\mathsf{SellerOnly}$ with respect to $\{p_i\}_{i\in [n]}$ and $\{s_i\}_{i\in [n]}$ with $c$ to be $\varepsilon^{-1}$ to get the optimal solution of the max player, denoted as $\left\{\omega_i^*\right\}_{i\in [n]}$. The mechanism simply chooses the price $p_i$ with probability $\omega_i^*$.

We now aim to illustrate that this mechanism achieves an approximation of at least $r^* -  5\varepsilon$. Let us again prove this by contradiction. Suppose there exists a buyer's distribution $F_B$ so that this mechanism has an approximation ratio lower than $r^* - 5\varepsilon$ on this particular instance $\I = (F_B, F_S)$. According to Lemma~\ref{lem:discretization}, we could correspondingly get a set $\{b_i\}_{i\in [n]}$ such that the following properties hold:
\begin{align}
\sum_{i = 1}^n b_i \geq 1 \text{ and } \sum_{i=1}^{n-1} b_i &\leq 1. \label{eq:seleronesided-1}\\
  \sum_{i=1}^n s_i p_i+\sum_{i=1}^{t-1}\sum_{j=t+1}^n s_ib_j(p_j-p_i) &\leq \ALGW(\I, p_t)\label{eq:seleronesided-2} \\
    \sum_{i = 1}^n \sum_{j=1}^n \max(p_i, p_j) s_i b_j &\geq \OPTW(\I).\label{eq:seleronesided-3}
\end{align}

It is clear that if $b_n \leq \varepsilon^{-1}$ holds, we can straightforwardly use the similar argument to the case for the buyer and show that $\{b_i\}_{i\in [n]}$ is indeed a feasible solution for the inner minimization problem with $c = \varepsilon^{-1}$, and thus reach a contradiction to the fact that the min-max optimization problem has a non-negative objective value. Therefore, the final thing left is to show that $b_n \leq \varepsilon^{-1}$. 

To prove this, consider the following specific set $\left\{b_i'\right\}_{i\in [n]}$ where $b_{i}' = 0$ for $i< n$ and $b_n' = 1$. It is clear that $\left\{b_i'\right\}_{i\in [n]}$ is a feasible solution for the inner minimization problem. As $\{\omega_i^*\}_{i\in [n]}$ is the optimal solution for the max player and has a non-negative objective value, we obtain that
\begin{align}\label{eq:sellerside-4}
    1 + \sum_{t=1}^{n-1}\omega_t^* \sum_{i=1}^{t-1}s_i (p_n - p_i) - (r^* - 2\varepsilon)\sum_{i=1}^n s_i \cdot  p_n \geq 0.
\end{align}

	Now suppose $b_n > \varepsilon^{-1}$, we could see that 
\begin{align}\begin{split}
    \label{eq:sellerside-5}
    \sum_{i=1}^n s_i p_i + \sum_{t=1}^n \omega_t^* \sum_{i=1}^{t-1} \sum_{j=t+1}^n (p_j - p_i) s_i b_j &\geq b_n  \sum_{t=1}^{n-1} \omega_t^* \sum_{i=1}^{t-1} (p_n - p_i) s_i\\
    & \geq (1 - \varepsilon) b_n \left(1 + \sum_{t=1}^{n - 1} \omega_t^* \sum_{i=1}^{t-1}(p_n - p_i) s_i\right)\\
    & \geq (1 - \varepsilon) (r^* - 2\varepsilon)\cdot b_n \sum_{i=1}^n s_i p_n 
\end{split}
\end{align} The second inequality is because that $\sum_{t=1}^{n-1}\omega_t^* \sum_{i=1}^{t-1}s_i (p_n - p_i) \geq (r^*-2\varepsilon) \sum_{i=1}^n s_i p_n -1 \geq 0.5 p_n \geq 0.5 \varepsilon^{-10} \gg \varepsilon^{-1}$ and the last inequality follows from~\eqref{eq:sellerside-4}.

Furthermore, notice that $\sum_{i=1}^{n-1} b_i \leq 1$ and $\sum_{i=1}^{n} s_i \leq 1 + \frac{\E[S]}{p_n} \leq 1 + \varepsilon^{10}$. It follows that 
\begin{align}\label{eq:sellerside-6}
    \sum_{i=1}^n \sum_{j=1}^{n - 1} s_i b_j \max(p_i, p_j) \leq \left(\sum_{i=1}^n s_i\right) \cdot \left(\sum_{j=1}^{n-1} b_j\right) \cdot p_n \leq 2 \varepsilon^{-10}.
\end{align}

However, $b_n > \varepsilon^{-1}$ and $\sum_{i=1}^n s_i \geq 1$ imply that 
\begin{align}\label{eq:sellerside-7}
    b_n \sum_{i=1}^n s_i  p_n \geq \varepsilon^{-11}.
\end{align}

Combining~\eqref{eq:sellerside-6} and \eqref{eq:sellerside-7}, we get that 
\begin{align}\label{eq:sellerside-8}
    b_n \sum_{i=1}^n s_i  p_n \geq (1 - 2\varepsilon) \sum_{i=1}^n \sum_{j=1}^n s_i b_j \max(p_i, p_j).
\end{align}

We further relax the RHS of~\eqref{eq:sellerside-5} using~\eqref{eq:sellerside-8}.
\begin{align}\label{eq:sellerside-9}
       \sum_{i=1}^n s_i p_i + \sum_{t=1}^n \omega_t^* \sum_{i=1}^{t-1} \sum_{j=t+1}^n (p_j - p_i) s_i b_j &\geq (r^{*} - 5\varepsilon) \sum_{i=1}^n \sum_{j=1}^n s_i b_j \max(p_i, p_j).
\end{align}

Putting~\eqref{eq:seleronesided-2},\eqref{eq:seleronesided-3} and \eqref{eq:sellerside-9} together, we obtain that 
\[\sum_{t=1}^n \omega_t^* \ALGW(\I, p_t) \geq  (r^* - 5\varepsilon) \OPTW(\I).\]

However, this contradicts the assumption that this mechanism achieves an approximation ratio lower than $r^{*} - 5\varepsilon$ on this instance. Thus, we show that $b_n \leq \varepsilon^{-1}$ always holds.

Finally, letting $\varepsilon \rightarrow 0$, we then complete our proof of Theorem~\ref{thm:one-sided}. 
\section{Proof of Theorem~\ref{thm:PartialMecha}}
\label{appendix:partial}

\paragraph{Seller's distribution mean $\E[S]$ is known.} We start by addressing the case in which only the mean of the seller's value, $\E[S]$, is known. The mechanism $\Mecha_S$ takes $\E[S]$ as input and randomly picks a number $x\sim U[0,3]$. Then $\Mecha_S$ sets the price as $x\cdot \E[S]$. As discussed in Section~\ref{sec:partial}, in order to demonstrate that $\Mecha_S$ achieves an approximation ratio of $\frac23$, it suffices to verify that $\inf_{\I = (F_S,F_B)\atop \E[S] = 1}\E_{q\sim \Mecha_{S}}\left[\ALG(q, \I) - \frac23\cdot \OPT(\I)\right]$ is non-negative.

We first aim to demonstrate that, in order to verify non-negativity, it is sufficient to consider two-point distributions for the seller and single-point distributions for the buyer. The intuition for focusing solely on such instances is rather straightforward. Consider the simple case in which the support is discrete. In this case, fixing the buyer's (or the seller's) distribution, the linear program to find the worst seller's (or the buyer's) distribution only has $2$ (or $1$) non-trivial constraints. This implies that there is an optimal solution, i.e., a worst-case distribution, that is supported on $2$ (or $1$) points. However, given that the support is, in fact, continuous, a more rigorous argument is necessary to validate this assertion. We now present it below.

Fix any distributions of the seller $F_S$, define $g(x)$ as the contribution to the objective when the buyer's value is $x$, i.e.,
\[g(x) = \E_{q\sim \Mecha_S\atop S\sim F_S}\left[S + \indic[S\leq q \leq x]\cdot (x - S) - \frac23 \max(S, x)\right].\]

This means that for any distribution $F_B$, it holds that 
\[\E_{q\sim \Mecha_{S}}\left[\ALG(q, \I) - \frac23\cdot \OPT(\I)\right] = \int g(x) {d} F_B(x),\] where $F_B(x)$ is the c.d.f. of distribution $F_B$ and $\I = (F_S, F_B)$ represents the instance.

Therefore, if there exists some distribution $F_B$ so that $\E_{q\sim \Mecha_{S}}\left[\ALG(q, \I) - \frac23\cdot \OPT(\I)\right]$ is negative, then there exists some $x_0\in \mathbb{R}_{\geq 0}$ such that $g(x_0) < 0$. Let $\I'$ represent the instance in which the seller's distribution remains to be $F_S$ and the buyer's distribution is a single-point distribution at $x_0$. It is clear that $\E_{q\sim \Mecha_{S}}\Big[\ALG(q, \I') - \frac23\cdot \OPT(\I')\Big] = g(x_0) < 0$. This means that for any distribution $F_S$, 
\[\inf_{F_B}\E_{q\sim \Mecha_{S}}\big[\ALG(q, \I) - \frac23\cdot \OPT(\I)\big]\] is non-negative if an only if for any single-point distribution $F_B$, $\E_{q\sim \Mecha_{S}}\left[\ALG(q, \I) - \frac23\cdot \OPT(\I)\right]$ is non-negative. This suggests that it suffices to examine all single-point distributions for the seller.

Now, we fix any distribution of the buyer $F_B$, and define $h(x)$ as the contribution when seller's value is $x$, i.e.,
\[h(x) = \E_{q\sim \Mecha_S\atop B\sim F_B}\left[x + \indic[x\leq q\leq B]\cdot\left(B - x\right) - \frac23\max(x, B)\right].\]

As the price $q$ is sampled from a continuous distribution $\Mecha_S$, and $F_B$ is a single-point distribution, we know that $h(x)$ is continuous. What's more, for any distribution $F_S$, it holds that 
\[\E_{q\sim \Mecha_{S}}\left[\ALG(q, \I) - \frac23\cdot \OPT(\I)\right] = \int h(x) {d} F_S(x) = \E_{S\sim F_S}\big[h(S)\big]\]

To show that two-point distributions can always achieve the minimum, define $r^{*}$ be the infimum of the following optimization problem:
\begin{align}\label{eq:rstaropt}
\underset{y,z,p\in \mathbb{R}_{\geq 0}}{\inf} \quad &p\cdot h(y) + (1 - p)\cdot h(z) \nonumber\\
\text{subject to:} \quad & py + (1 - p) z = 1  \nonumber \\
 & 0\leq p\leq 1.
\end{align}

Now define $k$ as $\sup_{y\in [0, 1)} \frac{r^{*} - h(y)}{1 - y}$. We first confirm the existence of a supremum, i.e., $k$ is a finite number here. This is because that if we define $p_y$ as  $\frac{1}{2 - y}$, it holds that $p_y \cdot y + (1 - p_y)\cdot 2 = 1$. Note that $r^*$ is the infimum of the optimization problem~\eqref{eq:rstaropt}, and this directly implies that $p_y\cdot h(y) + (1-p_y)\cdot h(2) \geq r^*$. After multiplying $\frac{1}{1-y}$ on both sides of the inequality and then reformulating it, we obtain that \[\frac{r^* - h(y)}{1 - y} \leq \frac{h(2) - h(y)}{2 - y}.\] This clearly illustrates that $\frac{r^* - h(y)}{1 - y}$ is at most $\frac{h(2) - h(y) }{2-y} \leq |h(2)| + |h(y)|$. Furthermore, given that $h(y)$ is a continuous function in $[0,1]$, it follows that $|h(y)|$ is bounded within the same interval. Thus, this means that  $\frac{r^* - h(y)}{1 - y}$  is upper bounded by a constant for any $y\in [0, 1)$, thereby confirming the existence of the supremum.

For any small enough constant $\varepsilon > 0$, let $L(x)$ be the line that go through the point $(1, r^{*} - \varepsilon)$ with a slope of $k - \varepsilon$, i.e. $L(x) = (k - \varepsilon)(x - 1) + r^{*} - \varepsilon$. We first argue that $L(x)$ is  a lower bound of $h(x)$.

\begin{lemma}
For any $x \in \mathbb{R}_{\geq 0}$, $L(x) \leq h(x)$.
\end{lemma}
\begin{proof}

First, it is clear that \[L(1) = r^{*} - \varepsilon < r^{*} \leq h(1),\]
and this implies that the inequality holds when $x = 1$.

As $k$ is the supremum of  $\frac{r^{*} - h(y)}{1 - y}$ over the interval $[0, 1)$, we know that there exists a $y'\in [0, 1)$, where $\frac{r^{*} - h(y')}{1 - y'}$ is at least $k - \varepsilon$. Suppose there exists some $x > 1$ such that $h(x) < \frac{r^{*} - h(y')}{1 - y'} (x - 1) + r^{*}$. Notice that $\left(y', x, \frac{x - 1}{x - y'}\right)$ is a valid solution for the optimization problem~\eqref{eq:rstaropt}, and it holds that 
\[h\left(y'\right)\frac{x - 1}{x - y'} + h\left(x\right) \frac{1 - y'}{x - y'} < h\left(y'\right)\frac{x - 1}{x - y'} + \left(\frac{r^{*} - h(y')}{1 - y'} (x - 1) + r^{*}\right)\frac{1 - y'}{x - y'} = r^{*}\]
which contradicts the fact that $r^{*}$ is the infimum of~\eqref{eq:rstaropt}. This implies that for any $x > 1$, \[h(x) \geq  \frac{r^{*} - h(y')}{1 - y'} (x - 1) + r^{*} \geq (k - \varepsilon) (x - 1) + r^{*} - \varepsilon = L(x).\]

Finally we consider the case where $x \in [0, 1)$. According to the definition of $k$, it follows that $\frac{r^{*} - h(x)}{1 - x} \leq k$ for all $x\in [0, 1)$. This implies that $h(x) \geq k(x - 1) + r^{*}$ for all $x\in [0, 1)$. As $x\in [0, 1)$, $x - 1$ lies within the interval $[-1, 0)$.  Thus
\[h(x)\geq k(x - 1) + r^{*} \geq (k - \varepsilon)(x - 1) + r^{*} - \varepsilon = L(x).\]
\end{proof}

Therefore, for any distribution $F_S$ satisfying $\E_{S\sim F_S}[S] = 1$, we could see that 
\[\E_{S\sim F_S}[h(S)]\geq \E_{S\sim F_S}[L(S)] = (k - \varepsilon)\cdot \E_{S\sim F_S}[(S - 1)] + r^{*} - \varepsilon = L(1) \geq r^{*} - \varepsilon.\]

This means that $\E_{S\sim F_S}\left[h(S)\right]$ is at least $r^{*} - \varepsilon$ for any distribution $F_S$ such that $\E_{S\sim F_S}[S] = 1$. What's more, by the definition of $r^{*}$, we can also see that there always exists a two-point distribution so that $\E_{S} [h(S)]$ is at most $r^{*} + \varepsilon$. Taking $\varepsilon \rightarrow 0$, we then show that it suffices to prove that $\E_{S} [h(S)]$ is non-negative for any two-point distribution $F_S$.

To sum up, we have argued that it suffices to show that 
\[\E_{q\sim \Mecha_{S}}\left[\ALG(q, \I) - \frac23\cdot \OPT(\I)\right]\geq 0 \] for all instances $\I$ that have a single-point buyer distribution and a two-point seller distribution. In our following proof, we assume that the value of the buyer is always $y$, and the seller's value is $x$ with probability $p$, and $\frac{1 - xp}{1 - p}$ with probability $1 - p$, where $x \in [0, 1]$ and $p\in [0, 1)$. Notice here the mean of the seller is scaled to $1$ and $\frac{1 - xp}{1 - p}$ is at least $1$. Recall that the price $q$ is chosen uniformly from the interval $[0, 3]$. We now break the problem into the following different cases, and argue that $\E_{q\sim \Mecha_{S}}\left[\ALG(q, \I) - \frac23\cdot \OPT(\I)\right]$ is non-negative in each of the cases:
\begin{itemize}
    \item $y\leq x$. Notice that when $y\leq x$, both $\ALG(q, \I)$ and $\OPT(\I)$ are $1$ since the trade never happens. 
    \item $x < y \leq \frac{1 - xp}{1 - p}$. In this case, we can see that $\E_{q\sim \Mecha_S}\left[\ALG(q, \I)\right] = 1 + (y-x)p\cdot \frac{\min(y, 3) - x}{3}$  and $\OPT(\I) = 1 + yp - xp$. Thus, we need to show that \[\E_{q\sim \Mecha_{S}}\left[\ALG(q, \I) - \frac23\cdot \OPT(\I)\right] = \frac13\left(1 + (y - x)\cdot p \cdot \left(\min(y,3) - x{-2}\right)\right)\] is non-negative. When $y\geq 3$, by some simple calculations, it holds that 
    \begin{align*}
    \E_{q\sim \Mecha_{S}}\left[\ALG(q, \I) - \frac23\cdot \OPT(\I)\right] = \frac13\left(1 + (y - x) p (1 - x)\right) \geq \frac13,
    \end{align*} where the last inequality follows from the non-negativity of $1-x$, $y - x$ and $p$. When $y < 3$, we can see that \begin{align*}
    \E_{q\sim \Mecha_{S}}\left[\ALG(q, \I) - \frac23\cdot \OPT(\I)\right] = \frac13\left(1 + (y - x) p (y - x - 2)\right) \geq 0,
    \end{align*}where the last inequality is because that $(y - x)(y - x - 2) \geq -1$ and $p \in [0, 1]$. 
    \item $x \leq \frac{1 - xp}{1 - p} \leq y$. In this case, \small {$\E_{q\sim \Mecha_S}\left[\ALG(q, \I)\right] = 1 + (y-x)p \frac{\min(y, 3) - x}{3} + \left(y - \frac{1 - xp}{1 - p}\right)\cdot (1 - p)\cdot \frac{\min(y, 3) - \min\left(\frac{1 - xp}{1 - p}, 3\right)}{3}$} and $\OPT(\I) = y$. We first assume $y \geq \frac{1-xp}{1-p}\geq 3$. Notice that $\frac{1-xp}{1-p}\geq 3$ implies that $(3 - x) p \geq 2$. Thus, 
    
    \begin{align*}
         \E_{q\sim \Mecha_{S}}\left[\ALG(q, \I) - \frac23\cdot \OPT(\I)\right] &= 1 + \frac13(y - x) p (3 - x) - \frac23 y\\
         & \geq \frac23 (y - x) + 1 - \frac23 y = 1 - \frac23 x \geq \frac13.
    \end{align*}

    We now consider the case where $y \geq3 
    \geq \frac{1-xp}{1-p}$. We know that  
    \begin{align*}        
    \E_{q\sim \Mecha_{S}}\left[\ALG(q, \I) - \frac23\cdot \OPT(\I)\right] &=1 + \frac13(y - x) p (3 - x)  + \frac13 \left(y - \frac{1 - xp}{1 - p}\right)(1 - p)\left(3 - \frac{1 - xp}{1 - p}\right)- \frac23 y\\
    & = \frac{p(x^2-2x+3)-2}{3(1-p)} + 1\geq \frac13,
    \end{align*} where the second equality follows from the equation that $\frac13p (3 - x) + \frac13(1 - p)\left(3 - \frac{1 - xp}{1 - p}\right) = \frac23$ and the last inequality is because $x^2 - 2x + 3\geq 2$ and also $p\in [0, 1]$. Finally, when $y\leq 3$, it holds that  
    \begin{align*}        
    \E_{q\sim \Mecha_{S}}\left[\ALG(q, \I) - \frac23\cdot \OPT(\I)\right] &= 1 + \frac13 p (y-x)^2 + \frac13 (1 - p)\left(y-\frac{1 - xp}{1 - p}\right)^2 - \frac23 y\\
    & \geq \frac{p(x - 1)^2}{3(1-p)} \geq 0,
    \end{align*} where the first inequality is from the fact that $1 + \frac13 p (y-x)^2 + \frac13 (1 - p)\left(y-\frac{1 - xp}{1 - p}\right)^2 - \frac23 y$ is minimized at $y = 2$.
\end{itemize}
Consequently, we demonstrate that $\Mecha_S$ achieves an approximation ratio of $\frac23$. 

\paragraph{Buyer's distribution mean $\E[B]$ is known.} We now prove that $\Mecha_B$ also has an approximation ratio of $\frac23$. Following the same argument, we could see that it suffices to prove that 
\[\E_{q\sim \Mecha_{B}}\left[\ALG(q, \I) - \frac23\cdot \OPT(\I)\right] \] is non-negative for all instances that have a single-point distribution for the seller and a two-point distribution for the buyer. Similarly, we assume that the seller's value is always $y$, and the buyer's value is $x$ with probability $p$ and $\frac{1 - xp}{1 - p}$ with probability $1 - p$ where $x\in [0,1]$ and $p\in [0, 1)$. Recall that the price $q$ is chosen from $[0, 2]$ according to the following cumulative distribution function $F(x)$:\begin{align*}
\begin{split}
F(x)=\left\{
\begin{aligned}
&\frac{x}{3-3x}  \quad & 0 \leq x\leq \frac12&\\
&(4x - 1) / 3  \quad & \frac12 <  x\leq \frac23\\
&(x + 1) / 3  \quad & \frac23 < x\ \leq 2\\
& 1 \quad & x\geq 2
\end{aligned}
\right.
\end{split}
\end{align*}

Now let us consider the following cases.
\begin{itemize}
    \item $x \leq \frac{1 - xp}{1 - p} \leq y$. This case is trivial, as both $\ALG(q, \I)$ and $\OPT(\I)$ are $y$.
    \item $x \leq y \leq \frac{1 - xp}{1 - p}$. In this case, we know that  $\E_{q\sim \Mecha_S} \left[\ALG(q, \I)\right] = y + (1 - p)\left(\frac{1 - xp}{1 - p} - y\right)\cdot\left(F\left(\frac{1 - xp}{1 - p}\right) - F(y)\right)$ and $\OPT(\I)$ is $1 - xp + yp$. Thus
    
    \begin{align*}
        \E_{q\sim \Mecha_{B}}\left[\ALG(q, \I) - \frac23\cdot \OPT(\I)\right] & = y \left(1 + F(y) - F\left(\frac{1 - xp}{1 - p}\right)\right) + (1 - xp + yp) \left(F\left(\frac{1 - xp}{1 - p}\right) - F(y) - \frac23\right)
\end{align*}
    
   \begin{enumerate}
   \item We first examine the scenario where $\frac{1 - xp}{1 - p} \geq 2$. Under these conditions, the function $F\left(\frac{1 - xp}{1 - p}\right)$ is equal to 1. Thus we get that in this scenario, \begin{align*}
        \E_{q\sim \Mecha_{B}}\left[\ALG(q, \I) - \frac23\cdot \OPT(\I)\right] & = y F(y) + (1 - xp + yp) \left(\frac13- F(y)\right)
\end{align*} 

	When $y \leq \frac{1}{2}$, the function $F(y)$ is upper bounded by $\frac{1}{3}$, which means that $\frac13 - F(y)$ is non-negative. Thus, we get that  \begin{align*}
        \E_{q\sim \Mecha_{B}}\left[\ALG(q, \I) - \frac23\cdot \OPT(\I)\right] & \geq yF(y) + (y - yp + yp)\left(\frac13 - F(y)\right) = \frac13 y \geq 0, 
    \end{align*}

    where the first inequality is because that $\frac{1 - xp}{1 - p} \geq y$ implies $1 - xp \geq y - yp$.
    
    When$y \geq \frac{1}{2}$, $\frac{1}{3} - F(y)$ is always non-positive. Since $y - x \geq 0$,  $(1 - xp + yp) \left(\frac{1}{3}- F(y)\right)$ is minimized when $p = 1$ for any fixed $x$ and $y$.\footnote{{Although in our definition we do not allow $p$ to be $1$, we can nevertheless lower bound $(1 - xp + yp) \left(\frac{1}{3}- F(y)\right)$  by setting $p$ to $1$.}} As a result, it can be concluded that 
    \begin{align*}
        \E_{q\sim \Mecha_{B}}\left[\ALG(q, \I) - \frac23\cdot \OPT(\I)\right] & \geq y F(y) + (1 - x + y) \left(\frac13- F(y)\right)\\
        & = \left(F(y) - \frac13\right)\cdot x + \frac13 (y + 1) - F(y)\\
        & \geq \frac13 (y + 1) - F(y) \geq 0.
    \end{align*} The second inequality arises from $F(y) \geq \frac13$ when $y\geq \frac12$, and the last inequality is follows from the observation that $\frac13(y + 1)\geq F(y)$ for all $y\in \mathbb{R}_{\geq 0}$.

\item We now proceed to the case where $\frac{1 - xp}{1 - p} < 2$. As $x \leq 1$, we know that $\frac{1 - xp}{1 - p}$ is never less than $1$, and consequently, $F\left(\frac{1 - xp}{1 - p}\right) = \left(\frac{1 - xp}{1 - p} + 1\right)/3$. Depending on the value of $y$, there are three distinct cases, which we present below. It suffices to demonstrate that all these terms have non-negative values. We postpone the proof of their non-negativity to Appendix~\ref{subsec:nonnega}.

\begin{table}[H]
\begin{center}
 \begin{tabularx}{0.77\textwidth}{@{}lML@{}}
 \toprule
 {Value of $y$} & \multicolumn{1}{c}{$\E_{q\sim \Mecha_{B}}\left[\ALG(q, \I) - \frac23\cdot \OPT(\I)\right] = $}
                              & \multicolumn{1}{l}{}\\ 
 \midrule
  $y\in \left[0, \frac12\right)$              & \frac13\left(pxy + \frac{p(1-x)(2-p-px)}{1 - p}- \frac{p(1-x)}{1 - y}\right)
                      & eq:Case11 \\
  $y\in \left[\frac12, \frac23\right)$              & \frac13\left(\frac{(1-px)^2}{1-p} + (5px-4)y+4(1-p)y^2\right)
                      & eq:Case12 \\
  $y\in \left[\frac23, 2\right)$ & \frac{y (1 + y)  + p^2 (y - x + 2) (y - x)   - p \left( y \left(3 + 2(y - x)\right)-2\right)-1}{3 (1 - p)} & eq:Case13 \\
  \bottomrule
  \end{tabularx}
  \end{center}
\end{table}
\end{enumerate}

    \item $y\leq x \leq \frac{1 - xp}{1 - p}$.  We get that  $\E_{q\sim \Mecha_B}\left[\ALG(q, \I)\right] = y + (1 - p)\left(\frac{1 - xp}{1 - p} - y\right)\cdot\left(F\left(\frac{1 - xp}{1 - p}\right) - F(y)\right) + p(x-y)\left(F(x)-F(y)\right)$ and $\OPT(\I)$ is exactly $1$.

Consequently, we must demonstrate that for any $y \leq x \leq \frac{1-xp}{1-x}$, the following term is always non-negative:
\begin{align*}
	y + (1 - p)\left(\frac{1 - xp}{1 - p} - y\right)\cdot\left(F\left(\frac{1 - xp}{1 - p}\right) - F(y)\right) + p(x-y)\left(F(x)-F(y)\right) - \frac23.
\end{align*}

To establish this, we consider the intervals in which the values of $x$,$y$ and $\frac{1-xp}{1-p}$ falls. Depending on the values of $x,y$ and $\frac{1-xp}{1-p}$, $F(x), F(y)$ and $F\left(\frac{1-xp}{1-p}\right)$ has different expressions, and we prove the non-negativity for each term. We present the expressions to demonstrate their non-negativity in Table~\ref{table:Discussion} and provide the proof in~\Cref{sec:nng for discussion table}.

\begin{table}[h]
\centering
\caption{Different Cases with Different Values}
\label{table:Discussion}
\begin{tabular}{|c|c|cc|}
\hline
                  &  & \multicolumn{2}{c|}{$y\in \left[0,\frac12\right)$}    \\ \hline
\multirow{4}{*}{$\frac{1-xp}{1-p}\geq 2$} & $x\in \left[0,\frac12\right)$ & \multicolumn{2}{c|}{\inlineeq{\frac{1}{3} - p(x - y)\left(1 - \frac{x}{3(1 - x)}\right) - \frac{y}{3}}\label{eq:Case21}}    \\ \cline{2-4} 
                  & $x\in \left[\frac12,\frac23\right)$ & \multicolumn{2}{c|}{\inlineeq{\frac{1}{3}\left(1- 4p\left(1 - x\right)\left(x - y\right) - y\right)}\label{eq:Case22}}    \\ \cline{2-4} 
                  & $x\in \left[\frac23, 1\right]$ & \multicolumn{2}{c|}{\inlineeq{\frac13 \left(1 - p (2 - x) (x - y) - y)\right)}\label{eq:Case24}}    \\ \hline
\multirow{4}{*}{$\frac{1-xp}{1-p}< 2$} & $x\in \left[0,\frac12\right)$ & \multicolumn{2}{c|}{\inlineeq{\frac{p\left(1 + p(x - y)(1 - x - x^2) + y - x(1 + x)y - 4(1 - x)x\right)}{3(1 - p)(1 - x)}}\label{eq:Case31}}    \\ \cline{2-4} 
                  & $x\in \left[\frac12,\frac23\right)$ & \multicolumn{2}{c|}{\inlineeq{\frac{p\left(1 + (4 - 3p)x^2 + 2(1 - p)y - x(4 + 3y - p(2 + 3y))\right)}{3(1 - p)}}\label{eq:Case32}}    \\ \cline{2-4} 
                  &$x\in \left[\frac23, 1\right]$  & \multicolumn{2}{c|}{\inlineeq{\frac{p (1 - x)^2}{3 (1 - p)}}\label{eq:Case34}}    \\ \hline
                  &  & \multicolumn{1}{c|}{$y\in \left[\frac12,\frac23\right)$} &  $y\in \left[\frac23, 1\right]$\\ \hline
\multirow{4}{*}{$\frac{1-xp}{1-p}\geq 2$} & $x\in \left[0,\frac12\right)$ & \multicolumn{1}{c|}{\xmark} & \xmark \\ \cline{2-4} 
                  & $x\in \left[\frac12,\frac23\right)$ & \multicolumn{1}{c|}{\inlineeq{\frac13 \left(2 - 4 p (1 - x) (x - y) + y (4y - 5)\right)}\label{eq:Case23}} & \xmark \\ \cline{2-4} 
                  & $x\in \left[\frac23, 1\right]$ & \multicolumn{1}{c|}{\inlineeq{\frac13 \left(2 - p (2 - x) (x - y) + y (4y - 5)\right)}\label{eq:Case25}} & \inlineeq{\frac13 \big(y^2 - p(2-x)(x-y)\big)  }\label{eq:Case26} \\ \hline
\multirow{4}{*}{$\frac{1-xp}{1-p}< 2$} & $x\in \left[0,\frac12\right)$ & \multicolumn{1}{c|}{\xmark} & \xmark \\ \cline{2-4} 
                  & $x\in \left[\frac12,\frac23\right)$ & \multicolumn{1}{c|}{\begin{footnotesize}\inlineeq{\frac{2p^2x\left(1 - \frac{3}{2}x\right) + (1-2y)^2 - 4px\left(1 - x\right) + (6-3x-4y)py - 2y\left(1 - \frac{3}{2}x\right)p^2}{3(1 - p)}}\label{eq:Case33}\end{footnotesize}} & \xmark  \\ \cline{2-4} 
                  & $x\in \left[\frac23, 1\right]$ & \multicolumn{1}{c|}{\inlineeq{\frac{1 + p (-2 + x) x}{{3}(1 - p)} - {\frac{1}{3}\cdot\left(4 y - 4 y^2\right)}}\label{eq:Case35}} & \inlineeq{\frac{y(1 + y) -1 + p(2 + (x - 2)x - y - y^2)}{3(1 - p)}}\label{eq:Case36} \\ \hline
\end{tabular}
\end{table}

\end{itemize}

 \paragraph{$\frac{2}{3}$ upper bounds the approximation ratio.} We now prove that $\frac23$ is the upper bound when we only know the mean of the buyer's value or the seller's value. We first consider the case when the mean of the seller's value is known. We introduce two instances with  $\E[S]=1$. As only $\E[S]$ is known, any mechanism must choose the price from the same distribution for the two different instances. Given any mechanism, let $P_S$ be the distribution of prices when $\E[S] = 1$, and $r_S$ be the approximation ratio of this mechanism. The two instances are as follows:

\begin{itemize}
    \item The values for both agents are deterministic. The seller's value is always $1$, and the buyer's value is always  $H$, which should be think of as a large number that approaches infinity. In order to obtain $r_S$ fraction of the optimal welfare in this instance when $H\rightarrow \infty$, the condition $\Pr_{p\sim P_S}[p \geq 1] \geq r_S$ must be satisfied. 
    \item The seller's value is $0$ with probability $1-\varepsilon$ and $1/\varepsilon$ with probability $\varepsilon$, and the buyer has a deterministic value of $1 - \varepsilon$. As $\varepsilon\rightarrow 0$, it is clear that $\OPT(\I)$ has a limit of $2$, while $\E\left[\ALG(q,\I)\right]$ has a limit of $1 + \Pr_{p\sim P_S}[0\le p < 1]$. Hence, $\Pr_{p\sim P_S}[0\leq p < 1] \geq 2r_S - 1$ must hold to achieve an $r_S$-approximate in this instance.
\end{itemize}

Notice that $\Pr_{p\sim P_S}[p \geq 1] + \Pr_{p\sim P_S}[0\leq p < 1]  = 1$, which implies that $3r_S - 1\leq 1$. Since such inequality holds for all mechanisms, this indicates that the upper bound for any mechanism given only $\E[S]$ is at most $\frac23$.

Finally, we prove the upper bound when only the mean of the buyer's value is known. Similarly, we again present two instances with identical $\E[B] = 1$ and demonstrate that no mechanisms can achieve an approximation ratio exceeding $\frac23$ in both instances. For any mechanism that only uses the information of $\E[B]$, we let $P_B$ be the distribution of prices when $\E[B] = 1$, and $r_B$ be the approximation ratio of this mechanism. The two instances can be described as follows:

\begin{itemize}
    \item The values for both agents are deterministic. The seller's value is always $0$, and the buyer's value is always $1$. In this instance, the trade must happen with probability at least $r_B$, which means that $\Pr_{p\sim P_B}[0\leq p\leq 1] \geq r_B$.
    \item The buyer's value is $0$ with probability $1-\varepsilon$ and $1/\varepsilon$ with probability $\varepsilon$, and the seller has a deterministic value of $1 + \varepsilon$. If we let $\varepsilon$ approach $0$, it is clear that $\OPT(\I)$ has a limit of $2$, while $\E\left[\ALG(q,\I)\right]$ has a limit of $1 + \Pr_{p\sim P_B}[p > 1]$. Thus, $\Pr_{p\sim P_B}[p > 1] \geq 2r_B - 1$ must hold to achieve an $r_B$-approximation in this instance.
\end{itemize}

Again, $\Pr_{p\sim P_B}[p > 1] + \Pr_{p\sim P_B}[0\leq p \leq 1]  = 1$ implies that $r_B$ is at most $\frac23$. This concludes our proof. 

\subsection{Proof of Non-negativity}
\label{subsec:nonnega}
We demonstrate the non-negativity of each term individually.

\begin{itemize}
\item[\ref{eq:Case11}:] It suffices to prove that $pxy + \frac{p(1-x)(2-p-px)}{1 - p}- \frac{p(1-x)}{1 - y}$ is always non-negative. We could see that 
	\begin{align*}
		\begin{split}
			pxy + \frac{p(1-x)(2-p-px)}{1 - p}- \frac{p(1-x)}{1 - y} &\geq \frac{p(1-x)(2-p-px)}{1 - p}- \frac{p(1-x)}{1 - y}\\
			& \geq \frac{p(1-x)(2-p-px)}{1 - p}- 2p(1-x)\\
			& = {\frac{p(1-x)(p-px)}{1-p}}\\
   &{={\frac{p^2(1-x)^2}{1-p}}}\\
			& \geq 0,
		\end{split}
	\end{align*}
 where the first inequality follows from the non-negativity of $pxy$, and the second inequality is because that the term $\frac{p(1-x)}{1 - y}$ is maximized at $y = \frac12$, since the value of $y$ is at most $\frac12$. The last inequality holds {as $p<1$.}

\item[\ref{eq:Case12}:] Our goal is to demonstrate the non-negativity of $4(1-p)y^2 + (5px-4)y + \frac{(1-px)^2}{1-p}$. {Note that \begin{align*}
	\begin{split}
		4(1-p)y^2 + (5px-4)y + \frac{(1-px)^2}{1-p} &= \frac{\left(2(1-p)y-(1-px)\right)^2}{(1-p)}+pxy\\
		& \geq 0.
	\end{split}
\end{align*} 
The inequality follows from the non-negativity of $pxy$ and the fact that $p<1$.
}

\item[\ref{eq:Case13}:] We need to show that $y (1 + y)  + p^2 (y - x + 2) (y - x)   - p \left( y \left(3 + 2(y - x)\right)-2\right)-1$ is non-negative. We first rearrange the terms to get the quadratic form with respect to $y$, i.e.
\begin{align*}
	&y (1 + y)  + p^2 (y - x + 2) (y - x)   - p \left( y \left(3 + 2(y - x)\right)-2\right)-1 \\&= (1-p)^2 y^2 + (1-p)\left(1 - 2 p (1 - x)\right)y+ p \left(2 - p (2 - x) x\right) - 1.
\end{align*}
	
	When we do not have any restrictions on $y$, we observe	 that the minimum is attained at $y = \frac{2p(1 - x) - 1}{2(1-p)}$.  We first consider the case when $\frac{2p(1 - x) - 1}{2(1-p)} \geq \frac23$. From this, we first get that $-6px \geq 7 - 10p$. As $-6px$ is always non-positive, $7 - 10p$ must also be at most $0$, meaning that $p\geq \frac{7}{10}$. Taking $y$ to be the minimum point, we get that 
	\begin{align*}
		 (1-p)^2 y^2 + (1-p)\left(1 - 2 p (1 - x)\right)y+ p \left(2 - p (2 - x) x\right) - 1 &\geq -p^2 + 3p-px- \frac54\\
		 & \geq -p^2 + 3p + \frac76 - \frac{10}{6}p - \frac54\\
		 & \geq 0.
		 	\end{align*} where the last inequality follows from the fact that the quadratic function $-p^2 + 3p + \frac76 - \frac{10}{6}p - \frac54$ remains positive over the interval $[\frac{7}{10},1]$.

	When $\frac{2p(1 - x) - 1}{2(1-p)} < \frac23$, as $y$ is at least $\frac23$,  the minimum is attained at $y = \frac23$. In this case, it holds that 
	\begin{align*}
		(1-p)^2 y^2 + (1-p)\left(1 - 2 p (1 - x)\right)y+ p \left(2 - p (2 - x) x\right) - 1 &\geq \frac{1}{9} \left( (16 - 30x + 9x^2) p^2 + (-2 + 3x)4p +1  \right)\\
  &{= \frac{1}{9}\left((3x-2)(3x-8)p^2 + 4(3x-2)p+1)\right)}\\
  & {= \frac{1}{9}\left((3x-2)(3x-8)\left(p+\frac{2}{3x-8}\right)^2 + 1-\frac{4(3x-2)}{3x-8}\right)}\\
  & {= \frac{1}{9}\left((3x-2)(3x-8)\left(p+\frac{2}{3x-8}\right)^2 +\frac{9x}{8-3x}\right)}
\geq 0.
	\end{align*}
{The final inequality is because $x\in\left[0,\frac{2}{3}\right]$.}

 \subsection{Non-negativity of the Terms in~\Cref{table:Discussion}}~\label{sec:nng for discussion table}
\item[\ref{eq:Case21}:] Notice that $\left(1 - \frac{x}{3(1 - x)}\right)$ remains non-negative for $x\in [0, \frac12]$. Consequently, the term is minimized at $p = 1$.
\begin{align*}
\frac{1}{3} - p(x - y)\left(1 - \frac{x}{3(1 - x)}\right) - \frac{y}{3} &\geq     \frac{1}{3} - (x - y)\left(1 - \frac{x}{3(1 - x)}\right) - \frac{y}{3}\\
& = \frac{1}{3(1-x)} \left(1 + 2y - x\left(4-4x + 3y\right)\right){=\frac{1}{3(1-x)} \left((1-2x)^2+2y-3xy\right)}.
\end{align*}{The final term in the inequality above is non-negative because $x \leq \frac12$ and $y\geq 0$ implies that $2y - 3xy \geq 0$.}

\item[\ref{eq:Case22}:]  We need to verify the non-negativity of $1- 4p\left(1 - x\right)\left(x - y\right) - y$. Notice that
\begin{align*}
    1- 4p\left(1 - x\right)\left(x - y\right) - y&\geq 1- 4\left(1 - x\right)\left(x - y\right) - y\\
    & \geq 1- 4\left(1 - \left(1 + y\right) / 2\right)\left( \left(1 + y\right)/ 2{-y}\right) - y\\
    & = y (1 - y) \geq 0,
\end{align*}where the first inequality is derived from that the term is minimized at $p = 1$, and the second inequality is because that the term is minimized when $x = \frac{(1 + y)}{2}$.

\item[\ref{eq:Case23}:] We here aim to ensure that $2 - 4 p (1 - x) (x - y) + y (4y - 5)$ is non-negative.
\begin{align*}
    2 - 4 p (1 - x) (x - y) + y (4y - 5) &\geq 2 - 4(1 - x) (x - y) + y (4y - 5)\\ 
            & \geq  2 - 4\left(1 - \left(1 + y\right) / 2\right) \left(\left(y+1\right)/2 - y\right) + y (4y - 5)\\
            & = {3\left(y-\frac{1}{2}\right)^2+\frac{1}{4}}\geq 0.
\end{align*}where the first inequality and the second inequality holds due to similar reasons as in \eqref{eq:Case22}. The term is minimized at $p = 1$ and $x = (1 + y) / 2$.

\item[\ref{eq:Case24}:] We aim to show that $1 - p (2 - x) (x - y) - y$ is at least $0$. It is clear to see that 
\begin{align*}
\begin{split}
    1 - p (2 - x) (x - y) - y &\geq 1 - (2 - x)(x - y) - y \\
                        & \geq 1 - (2 - 1)\cdot (1 - y) - y\\
                        & = 0,
\end{split} 
\end{align*} where the first inequality once again follows from the fact that this term is minimized at $p = 1$. For the second inequality, note that if we do not have any constraints on $x$, the term would be minimized at $x = (2 + y) / 2 \geq 1$. However, since the value of $x$ is at most $1$, we find  that this term actually attains its minimum at $x = 1$.

\item[\ref{eq:Case25}:] It is sufficient to prove that $2 - p (2 - x) (x - y) + y (4y - 5)$ is always non-negative. Observe that if we consider $y$ as a constant, this term is nearly identical to \eqref{eq:Case24}, except for constants. Consequently, it also reaches the minimum at $p = 1$ and $x = 1$.
\begin{align*}
2 - p (2 - x) (x - y) + y (4y - 5) 
& \geq 2 -  (2 - 1) (1 - y) + y (4y - 5)\\
& = \left(1-2y\right)^2 \geq 0.
\end{align*} 

\item[\ref{eq:Case26}:] The term for which we aim to ensure non-negativity is $y^2 - p(2 - x) (x - y)$. Similar to \eqref{eq:Case24}, it is evident that the term attains its minimum at $p = 1$ and $x = 1$.
\begin{align*}
    y^2 - p(2 - x) (x - y) &\geq y^2 - (1 - y)\\
    & = y^2 + y - 1 \geq 0.
\end{align*} where the last inequality follows from the non-negativity of $y^2 + y - 1$ on $\left[\frac23, 1\right]$.

\item[\ref{eq:Case31}:] As $1 - x$, $1 - p$, $p$ are always non-negative, we only need to prove that the value of $1 + p(x - y)(1 - x - x^2) + y - x(1 + x)y - 4(1 - x)x$ is at least $0$. Since $1 - x - x^2$ is non-negative for $x\in \left[0, \frac12\right)$, both $p(x-y)(1-x-x^2)$ and $(1 - x - x^2)y$ are always non-negative, and thus the following holds:
\begin{align*}
	1 + p(x - y)(1 - x - x^2) + y - x(1 + x)y - 4(1 - x)x & \geq 1 + \left(1 - x - x^2\right)y - 4(1 - x)x\\
	& \geq 1 - 4(1 - x)x\\
	& \geq 0.
\end{align*}where the last inequality dues to the non-negativity of $4x^2-4x+1$.

\item[\ref{eq:Case32}:] We need show that $1 + (4 - 3p)x^2 + 2(1 - p)y - x\left(4 + 3y - p(2 + 3y)\right)$ is always non-negative. We first reorganize the term and get $1 + (4 - 3 p) x^2 + (2 (1 - p) - 3 x (1 - p)) y - x (4 - 2 p)$. As $x\leq \frac23$, we could see that $2 (1 - p) - 3 x (1 - p)$ is always non-negative. Thus,
\begin{align*}
	1 + (4 - 3 p) x^2 + (2 (1 - p) - 3 x (1 - p)) y - x (4 - 2 p) &\geq (4 - 3 p) x^2 - (4-2p) x + 1\\
	& \geq  1 -\frac{(4-2p)^2}{4(4-3p)}\\
	& {=\frac{p-p^2}{(4-3p)} }\geq 0
\end{align*} where the second inequality is derived from that the quadratic term obtains its minimum at $x = -\frac{4-2p}{2(4-3p)}$. The final inequality is because $p-p^2\geq 0$ for $p\in [0, 1]$.

\item[\ref{eq:Case33}:] We aim to demonstrate that $2p^2x\left(1 - \frac{3}{2}x\right) + 1 - 4y + 4y^2 - 4px\left(1 - x\right) + 6py - 2y\left(1 - \frac{3}{2}x\right)p^2 - 3pxy - 4py^2$ is always non-negative when $x,y\in \left[\frac12,\frac23\right)$ and $p\in [0, 1)$. Notice that $x < \frac23$ implies that $1 - \frac23x$ is positive in the interval. Therefore, we only need to establish that $1 - 4y + 4y^2 - 4px\left(1 - x\right) + 6py - 2y\left(1 - \frac{3}{2}x\right)p^2 - 3pxy - 4py^2$ is non-negative. Fixing $x,y$, we could see that this is a quadratic term with respect to $p$, and the coefficient of $p^2$ is $-2y\left(1-\frac23x\right)$ which is non-positive. Therefore, the minimum of this quadratic term is achieved at the boundary, i.e., $p = 0$ or $p = 1$. When $p = 0$, the expression is equivalent to
\[4y^2 - 4y + 1 = (2y - 1)^2 \geq 0.\] When $p = 1$, the expression becomes {\[1+x^2\left(4-\frac{9}{2}y\right)+x(-4+3y)=(2x - 1)^2+3x\left(1-\frac{3}{2}x\right)y\geq 0.\] The inequality is because $3x(1-\frac{3}{2}x)$ is non-negative for $x\leq \frac{2}{3}$. }

\item[\ref{eq:Case34}:] This case is trivial, as $p\in [0,1)$.

\item[\ref{eq:Case35}:] This one is also easy. First, notice that $(-2 + x) x \geq -1$. Therefore, it holds that
\begin{align*}
	\frac{1+p(x-2)x}{1-p} + 4y^2 - 4y& \geq \frac{1-p}{1-p} + 4y^2 - 4y\\
	& = (2y - 1)^2 \geq 0.
\end{align*}

\item[\ref{eq:Case36}:] It suffices to show that $-1 + y(1 + y) + p(2 + (-2 + x)x - y - y^2)$ is always non-negative. Once again, using the fact that $(-2 + x) x \geq -1$, we observe that
\begin{align*}
	-1 + y(1 + y) + p(2 + (-2 + x)x - y - y^2) &\geq y(1 + y) - 1+(1-y-y^2)p\\
	& \geq y(1 + y) - 1+ (1-y-y^2) = 0
\end{align*} where the second inequality arises from the observation that $(1-y-y^2)$ remains negative when $y\in \left[\frac23, 1\right]$, and thus the expression attains the minimum value at $p = 1$.

\end{itemize}

\section{Missing Proofs in Section~\ref{sec:Sample}}
\label{appendix:SampleAppendix}

\subsection{Mechanisms over Buyer's information} \label{appendix:NonbuyerInfo}
In the sample setting, we only consider mechanisms over seller's information. We do not consider quantile or order statistics mechanisms over buyer's information since it is impossible to get any constant approximation with these family of mechanisms.

\begin{theorem}
\label{thm:NonBuyerInfo}
No quantile mechanism over buyer's distribution or order statistic mechanism over only buyer's samples can achieve a constant fraction of the optimal welfare.
\end{theorem}

\begin{proof}
    
We first show that there is no constant approximation quantile mechanism $Q$ over buyer's distribution. Remind that $F_B$ and $F_S$ respectively stand for the distribution of the buyer and the seller, and we will also use $Q$ to denote the corresponding distribution over the buyer's quantile.

To start with, we can assume that distribution $Q$ does not have point mass at $1$. That's because if we set the $1$-quantile of the buyer's distribution, i.e. $F_B^{-1}(1)$, as the price, we have $\Pr_{B\sim F_B}[B \geq F_B^{-1}(1)] = 0$. This means that the trade will never happen under such price and thus this price will not increase the welfare. Therefore, if we move this probability mass to other values, the welfare and also the approximation ratio will not decrease, and we prove that this assumption is with out loss of generality.

Now, for an arbitarily small $\varepsilon > 0$, we will show that there is no $\varepsilon$-approximation quantile mechanism over buyer's distribution. For any quantile mechanism $Q$ over buyer's distribution, we construct the following set \[
\mathcal{X} = \{{t\in [0,1]}\given \Pr_{x\sim Q}[x\ge t] \le \varepsilon / 2\}
\]

Since there is no point mass at $1$, this set will contain some $t\in \mathcal{X} $ and $ t \neq 1$. Consider the following instance $\I = (F_S, F_B)$:

\begin{align*}
\begin{split}
    		 B \sim F_{B}, B=\left\{
\begin{aligned}
&0   \quad & w.p. \quad t\\
&H & w.p.\quad 1 - t
\end{aligned}
\right.
\quad\quad\quad
S \sim F_{S}, S=
\begin{aligned}
&\varepsilon / 2  \quad & w.p. \quad 1\\
\end{aligned}
\end{split}
\end{align*}

where $H = \frac{1}{1 - t}$ is a large enough number. 

In this instance, the intuition is that all the welfare is hide at some very little probability of the buyer, and we must make sure that the trade is very likely to happen when the buyer has a very high value. However, since we don't know the value of the seller, it is hard for us to make sure that $p\ge S$ which means that this trade will not happen. 

Remind that we define $\OPTW(\I)$ as the optimal welfare, i.e., $\E_{S\sim F_S, B\sim F_B}[\max(S, B)]$ against instance $\I = (F_S, F_B)$ and $\ALGW(Q, \I)$ as the welfare of mechanism $Q$ against instance $\I$. Formally speaking, we have that \[\OPTW(\I) \ge H \cdot (1 - t) = 1\]

and also 

\begin{align*}
\begin{split}
\ALGW(Q, \I)  &= \E[S] + \E_{p\sim F_B^{-1}(Q)}[(B - S)\cdot \indic[B\ge p\ge S]]\\
			&\le \varepsilon/2 + \E_{p\sim F_B^{-1}(Q)}[B \cdot \indic[p\ge S]]\\
			&\le \varepsilon/2 + H \cdot (1 - t)\cdot \varepsilon / 2 = \varepsilon
\end{split}
\end{align*}

where the last inequality holds since $ p = F_B^{-1}(x)\ge S$ is equivlent to $x \ge t$ where $x$ is drawn from $Q$, and this happends w.p. at most $\varepsilon/2$ by the definition of $t$. So for every distribution $Q$ over buyer's quantile, we find an instance $\I$ so that $\ALGW(Q, \I) \le \varepsilon \cdot \OPTW(\I)$, which completes the first part of our proof.

Next we aim to show that for any $\varepsilon > 0, N > 0$, there is no $\varepsilon$-approximation mechanisms using only $N$ samples from the buyer.

First, for any mechanisms $\Mecha$ using $N$ samples from the buyer, it can be formallized as a mapping 
\[
f: \mathbb{R}_{\ge 0}^N \mapsto \Delta\left( \mathbb{R}_{\ge 0} \right)
\]

where $f(x_1, x_2, \cdots, x_N)$ stands for the distribution of price selected by this mechanism after receiving $N$ samples $(x_1, x_2, \cdots, x_N)$. Let $D$ be $f(0, 0, \dots, 0)$, which is the distribution of the price if this mechanism sees $N$ samples all with value $0$. Similarly, we consider the following set:
\[
\mathcal{H}' = \{{t\in \mathbb{R}_{\ge 0}}\given \Pr_{x\sim D}[x\ge t] \le \varepsilon / 2\}
\]

Again we know this set is non-empty, so let $t$ be any real positive number in the set $\mathcal{H}'$. Therefore, we could construct an instance $\I = (F_S, F_B)$ satisfying
\begin{align*}
\begin{split}
    		 B \sim F_{B}, B=\left\{
\begin{aligned}
&0   \quad & w.p. \quad (1-\varepsilon/4)^{1/N}\\
&H & w.p.\quad 1 - (1-\varepsilon/4)^{1/N	}
\end{aligned}
\right.
\quad\quad\quad
S \sim F_{S}, S=
\begin{aligned}
&t + 1 \quad & w.p. \quad 1\\
\end{aligned}
\end{split}
\end{align*}

where $H > \frac{t + 1}{\varepsilon / 4\cdot (1 - (1 - \varepsilon / 4)^{1/N}}$ is a large enough number.

In this instance, we can see that with just $N$ samples, no mechanism can distinguish this instance with another instance whose buyer always have a value of $0$. Therefore, it can not get the welfare hidden at the buyer. Formally speaking:
\[
\OPTW(\I) \ge H \cdot \left(1 - (1-\varepsilon/4)^{1/N}\right) > \frac{t+1}{\varepsilon / 4}
\]

To calculate $\ALGW(\Mecha, \I)$, we consider the case when all the samples are zero and the case when there is at least one non-zero number in the samples. In the latter case, the probability that at least one sample is non-zero is at most $1 -\left( (1 - \varepsilon / 4)^{1/N}\right)^N = \varepsilon / 4$ which is negligible. In the former case, since $\Pr_{p\sim f(0, 0, \dots, 0)}[p \ge t]\le \varepsilon / 2$, the trade happends w.p. at most $\varepsilon / 2$. Therefore, we could expand $\ALGW(\Mecha, \I)$ into:
\begin{align*}
	\begin{split}
\ALGW(\Mecha, \I)  &\le  \E_{p\sim f(0, 0, \dots, 0)}[S + (B - S)\cdot \indic[B \ge p \ge S]] \cdot \Pr[\text{All } N \text{ samples are } 0]\\
	& \quad + \OPTW(\I) \cdot \Pr[\text{at least } 1 \text{ sample is not } 0]\\
	& \le (t + 1+ \E[B \cdot\indic[p\ge S]]) \cdot 1  + \OPTW(\I) \cdot(\varepsilon / 4)\\
	& \le \OPTW(\I) \cdot (\varepsilon / 4 + \varepsilon / 2 + \varepsilon / 4)\\
	& = \varepsilon \cdot \OPTW(\I)
	\end{split}
\end{align*}

where the second inequality holds since 
$ t + 1\le (\varepsilon / 4) \cdot \OPTW(\I)$,
$\Pr_{p\sim f(0, 0,\dots,0)}[p\ge S] \le \varepsilon / 2$ and 
$\E_{B\sim F_B}[B]\le \E_{S\sim F_S, B\sim F_B}[\max(S, B)] = \OPTW(\I)$.

And this finishes our proof.
\end{proof}

\subsection{Proof of Lemma~\ref{lem:smallsamplekey}}
\label{appendix:proofofsmallsamplekey}

The proof here is quite straight forward. As we show in Section~\ref{subsec:Connection}, each order statistic mechanism corresponds to a quantile mechanism. Thus $\mathcal{C}(\mathcal{P}(Q))$ is exactly the approximation ratio of the order statistic mechanism $Q$. What's more, we could see that the $\Delta_N$ enumerates all possible order statistic mechanisms with $N$ samples. Therefore, this directly implies that $\argmax_{Q\in \Delta_N} \mathcal{C}(\mathcal{P}(Q))$ is the optimal order statistic mechanism with $N$ samples.

\subsection{Proof of Lemma~\ref{lem:symratio}}
\label{appendix:proofofsymratio}

Fix an instance $\I = (F, F)$, recall that $S$ and $B$ are the random variables respectively indicating the value of the seller and the buyer. Define $\ALG$ to be the random variable which indicates the welfare of our mechanism in the realization, which is $S + (B - S) \cdot \indic[B\ge p \ge S]$ where $p$ is the price chosen by our quantile mechanism $Q$. Similarly, let $\OPT$ be the random variable which indicates the optimal welfare in the realization, which is $\max(B, S)$.

To prove Lemma~\ref{lem:symratio}, we introduce the following lemma.

\begin{lemma}\label{lem:symcalc}
For any quantile mechanism $Q$, let $\ALG$ and $\OPT$ respectively be the random variables indicating the welfare of the mechanism $Q$ and the optimal welfare in the realization. Let $r$ be 
\[\min_{\I = (F, F)}\inf_{x\in [0,1)} \frac{\Pr[\ALG \ge F^{-1}(x)]}{\Pr[\OPT\ge F^{-1}(x)]}\] where $F(x)$ is the cumulative distribution function of distribution $F$, and $F^{-1}(x)$ is the quantile function. The quantile mechanism $Q$ is at least $r$-approximate. 
\end{lemma}

\begin{proof}
	We have \[\Pr[\ALG\ge F^{-1}(x)] \ge r\cdot \Pr[\OPT\ge F^{-1}(x)]\] for all $x\in [0, 1]$ and quantile function $F^{-1}(x)$.
	
	Without loss of generality, we could assume the distribution has a support over $[0, a]$. Notice that since we assume the distribution is continuous w.l.o.g. in the sample setting, $F^{-1}(x)$ is a continuous and increasing function over $[0, 1]$ and $F^{-1}(0) = 0, F^{-1}(1) = a$, so we have
	\begin{align*}
		\ALGW(\I, Q) &= \E[\ALG]\\
			& = \int_{0}^a \Pr[\ALG \ge x] \mathop{dx}\\
			& = \int_{0}^1 \Pr[\ALG \ge F^{-1}(z)] \mathop{d F^{-1}(z)}\\
			& \ge \int_{0}^1 r \cdot \Pr[\OPT \ge F^{-1}(z)] \mathop{d F^{-1}(z)}\\
			& = r\cdot \int_{0}^a \Pr[\OPT \ge x] \mathop{dx} = r\cdot \E[\OPT]\\
			& = r\cdot \OPTW(\I)
	\end{align*}holds for any instance $\I = (F, F)$, which implies that quantile mechanism $Q$ is at least $r$-approximate.

\end{proof}

With Lemma~\ref{lem:symcalc}, we are able to give a lower bound of approximation ratio for any quantile function $Q$.

Fixing the buyer and seller's distribution $F$, we only need to calculate the term $\Pr[\ALG \ge F^{-1}(x)]$ and $\Pr[\OPT \ge F^{-1}(x)]$. The event $\OPT\ge F^{-1}(x)$ happens if and only if either $B$ or $S$ is greater than $F^{-1}(x)$. Thus,
\begin{equation}
\label{eq:termOPT}
\Pr[\OPT\ge F^{-1}(x)] = 1-x^2	
\end{equation}

The event $\ALG\ge F^{-1}(x)$ happens if and only if one of the following conditions is satisfied:
\begin{itemize}
	\item $S\ge F^{-1}(x)$
	\item $p\le F^{-1}(x)$, $S\le p$  and $B\ge F^{-1}(x)$. Here $S\le p\le B$, thus the trade takes place, and $B\ge F^{-1}(x)$.
	\item $p> F^{-1}(x)$, $S\le F^{-1}(x)$ and $B\ge p$. Since $S\le p\le B$, the seller trades the item to the buyer,and we have $B\ge p\ge F^{-1}(x)$.
\end{itemize}

Note that these three events are disjoint, so we could calculate the probability for each event to happen and add them up.

For the first event, $\Pr[S\ge F^{-1}(x)] = 1 - x$.

For the second event, we just enumerate the quantile of $p$. Suppose the quantile of $p$ is $t$, which means that $F^{-1}(t) = p$. Then, we have $\Pr[S\le p] = t$ and $\Pr[B\ge F^{-1}(x)] = 1-x$. Thus, this event takes place w.p. $\int_{[0,x]} t(1-x)\, dQ(t)$ where $Q(t)$ is the c.d.f. of distribution $Q$.

For the third event, we use the same idea. Suppose the quantile of $p$ is $t$, we have $\Pr[S\le F^{-1}(x)] = x$ and $\Pr[B\ge p] = 1 - t$. Therefore, this event happens w.p. $\int_{(x,1]}(1-t)x\, dQ(t)$.

By adding the terms above up, we have:
\begin{equation}
\label{eq:termALG}
	\Pr[\ALG\ge F^{-1}(x)] = {\int_{[0,x]} t(1-x) \, d Q(t) + \int_{(x, 1]} (1-t)x\, dQ(t) + (1-x)}
\end{equation}

Therefore, combining Lemma~\ref{lem:symcalc} and Equation~\eqref{eq:termOPT} and \eqref{eq:termALG}, we have that for any quantile mechanism $Q$ with c.d.f. $Q(x)$, the minimum of the following optimization problem lower bounds the approximation ratio of the quantile mechanism $Q$.
\begin{align*}
\begin{split}    
& \min_{\I = (F, F)}  \inf_{x\in [0,1)} \frac{\int_0^x q(t)\cdot t(1-x)\mathop{dt} + \int_x^1 q(t)\cdot (1-t)x\mathop{dt} + (1-x)}{1-x^2} \\
 & = \inf_{x\in [0,1)}  \frac{\int_0^x q(t)\cdot t(1-x)\mathop{dt} + \int_x^1 q(t)\cdot (1-t)x\mathop{dt} + (1-x)}{1-x^2}
\end{split}
\end{align*} where the equality holds since we could see that the term is independent from $F(x)$. 

We now left to show that the approximation ratio of $Q$ is also upper bounded by $r$. It suffices to show that for any $\varepsilon > 0$ there exists some instance $\I = (F, F)$ such that  \[\ALGW(Q, I) \le (r + \varepsilon) \cdot \OPTW(\I) \]

First, Recall Equation~\eqref{eq:termOPT} and~\eqref{eq:termALG}. We could see that both the term $\Pr[\ALG\ge F^{-1}(x)]$ and the term $\Pr[\OPT\ge F^{-1}(x)]$ are independent of the distribution $F$. Thus, \[r = \min_F\inf_{x\in [0,1)} \frac{\Pr[\ALG \ge F^{-1}(x)]}{\Pr[\OPT\ge F^{-1}(x)]} = \inf_{x\in [0,1)}  \frac{\int_0^x q(t)\cdot t(1-x)\mathop{dt} + \int_x^1 q(t)\cdot (1-t)x\mathop{dt} + (1-x)}{1-x^2}\]

Suppose the optimum of $ \inf_{x\in [0,1)}  \frac{\int_0^x q(t)\cdot t(1-x)\mathop{dt} + \int_x^1 q(t)\cdot (1-t)x\mathop{dt} + (1-x)}{1-x^2}$ is attained at $x^*$. Consider the following instance $\I = (F, F)$ satisfying
\begin{align*}
\begin{split}
    		 v \sim F, v=\left\{
\begin{aligned}
&U[0, r \cdot (1 - x^*) \cdot \varepsilon/2]   \quad & w.p. \quad x^*\\
&U[1, 1 + \varepsilon/2]   & w.p.\quad 1 - x^*
\end{aligned}
\right.
\end{split}
\end{align*}

Notice that in this instance, The event $\indic[\ALG \ge F^{-1}(x^*)]$ is equivalent to $\indic\Big[\ALG \in [1, 1 + \varepsilon/2]\Big]$. Such argument also holds for $\OPT$. Thus, we can see that
\begin{align*}
\begin{split}
	\ALGW(Q, \I) &= \E[\ALG] \\
	&   \le \Pr[\ALG \in [1, 1 + \varepsilon/2 ]]\cdot (1 + \varepsilon/2) + \Pr[\ALG \in [0, r \cdot (1 - x^*) \cdot \varepsilon/2]] \cdot (r \cdot (1 - x^*) \cdot \varepsilon/2)\\
		& \leq \Pr[\ALG \ge F^{-1}(x^*)]\cdot ( 1 + \varepsilon/2) + r \cdot (1 - x^*) \cdot \varepsilon/2 \\
		 &= r\cdot \Pr[\OPT\ge F^{-1}(x^*)]\cdot (1 + \varepsilon/2) + r \cdot (1 - x^*) \cdot \varepsilon/2\\
		 & \leq r\cdot \E[\OPT]\cdot (1 + \varepsilon / 2) + r\cdot (1 - x^*) \cdot \varepsilon/2\\
		 & \leq (r + \varepsilon) \cdot \E[\OPT] =  (r + \varepsilon) \cdot \OPTW(\I)
\end{split}
\end{align*} where the second equation uses the fact that  $r\cdot \Pr[\OPT\ge F^{-1}(x^*)] = \Pr[\ALG \ge F^{-1}(x^*)]$ since $r$ is attained at $x^*$.

Since the above holds for every $\varepsilon > 0$, this completes our proof of Lemma~\ref{lem:symratio}.
\subsection{Proof of Theorem~\ref{thm:symoptorderstatistic}}
\label{appendix:proofsymoptorderstatistic}

The proof of Theorem~\ref{thm:symoptorderstatistic} is directly a combination of Lemma~\ref{lem:smallsamplekey} and Lemma~\ref{lem:symratio}. From Lemma~\ref{lem:symratio}, we know that the approximation ratio of a quantile mechanism $Q$ with c.d.f. $Q(x)$ is exactly
\[ \inf_{x\in [0,1)} \frac{\int_{[0,x]} q(t)t(1-x) \, d t + \int_{(x, 1]} q(t)(1-t)x\, dt + (1-x)}{1-x^2}. \]
where we use $q(x)\, dx$ instead of $dQ(x)$ since we have a continuous probability density function.

We could see that the set of all distributions over $[N]$ is actually $\Big\{\{P_i\}_{i\in [n]}~|~P_i\geq 0\text{ and } \sum_{i=1}^n P_i = 1\Big\}$ where distribution $P$ would choose $i$ w.p. $P_i$. What's more, the probabiltiy density function of the quantile of such order statistic mechanism is exactly $p(x) = \sum_{i=1}^n P_i \pdf(x)$ where $\pdf(x) = N \binom{n-1}{i-1} \cdot x ^{i-1}\cdot(1-x)^{N-i}$. Therefore, combining it with Lemma~\ref{lem:smallsamplekey}, it follows that Theorem~\ref{thm:symoptorderstatistic} holds.

\subsection{Analysis of Empirical Risk Minimization Mechanism}
\label{appendix:ERM}

In this section, we give the upper bounds of the Empirical Risk Minimization mechanism. 

When $N = 1$, the empirical distribution is a one-point distribution, so any price $x\in \mathbb{R}_{\ge 0}$ is optimal. Therefore, we consider the following instance $\I = (F, F)$ s.t.

\begin{align*}
\begin{split}
    		 x \sim F, x=\left\{
\begin{aligned}
&0   \quad & w.p. \quad 1 - 1/H\\
&H \quad & w.p. \quad 1/H
\end{aligned}
\right.
\end{split}
\end{align*}

Since any price is optimal for ERM, we can assume that it will always select $H + 1$ so that the trade will never take place. Therefore taking $H\rightarrow \infty$ in this instance, the approximation ratio tends to $0.5$.

When $N = 2$, the empirical distribution is a two-point distribution. Suppose the two samples are $X_1 \le X_2$, we can see that any price in $[X_1, X_2)$ is optimal for this empirical distribution. Thus, we can assume that it will always select $X_1$. So, ERM is equivalent to an order statistic mechanism that always selects the smallest sample. Therefore, we can solve the following optimization problem: 
\[\inf_{x\in [0, 1)}\frac{\frac13x^3-x^2-\frac13x+1}{1-x^2} = \frac{2}{3}\]

Then, by Lemma~\ref{lem:symratio}, we know that there exists an instance such that ERM achieves exactly $\frac{2}{3}$ approximation. 

Now we are in the case that $N = 3$. By some calculations, we know that the second smallest sample will always be an optimal choice. Therefore, ERM is equivalent to an order statistic mechanism that always selects the second smallest sample when $N = 3$. Similarly, we can calculate that 

\[\inf_{x\in [0, 1)}\frac{\frac12x^4-x^3-\frac12x+1}{1-x^2} = \frac{3}{4}\]

Applying Lemma~\ref{lem:symratio} again, there is an instance such that ERM has an approximation ratio of exactly $\frac{3}{4}$.

Our proof strategy changes when the number of samples is greater than $5$. We consider a particular instance $\I = (F, F)$ and calculate the performance of ERM on such instance. 

\begin{align*}
\begin{split}
    		 x \sim F, x=\left\{
\begin{aligned}
&0   \quad & w.p. \quad (\sqrt2-1)\cdot(1 - \frac{1}{n})\\
&\frac1n  \quad & w.p. \quad 2 - \sqrt2\\
&1 \quad & w.p. \quad \frac{1}{n}(\sqrt2-1)
\end{aligned}
\right.
\end{split}
\end{align*}

Actually, this is a counterexample appeared in \cite{kang2022fixed}. They show that 
\[\OPTW(\I) = \frac{4(\sqrt2-1)}{n} - \frac{4\sqrt2 -1	}{n^2}\]

Since we will let $n \rightarrow \infty$, we will ignore the $O(\frac1{n^2})$ terms in the following calculation. Notice that $\Pr[x = 1] = O(\frac{1}{n})$, the probability that there is at least $1$ sample with value $1$ is negligible, so we will also assume that all samples are $0$ or $\frac{1}{n}$. Recall that there will be a tie-breaker coordinate drawn uniformly from $[0, 1]$ for each variable, and we will compare the tie-breaker coordinate if they have the same value. Now suppose there are $k_1$ samples with value $0$ and $k_2$ samples with value $\frac{1}{n}$. We know that the largest $0$ or the smallest $\frac{1}{n}$ is an optimal price for the empirical distribution when $k_1, k_2\neq 0$. Now, if we choose the largest $0$, as the price $p$, the expected welfare is:
\begin{align}
\begin{split}	
	W_1(k_1) = &\Pr[S = 0] \cdot \Pr[S < p] \cdot E[B] + \Pr\left[S = \frac{1}{n}\right] \cdot \frac{1}{n} + \Pr[S = 1] \cdot 1\\
	= &(\sqrt2-1)\cdot \frac{k_1}{k_1+1} \left((2-\sqrt2+\sqrt2-1)\cdot \frac1n\right) + (2-\sqrt2)\cdot \frac{1}{n} + \frac{1}{n}\cdot (\sqrt2-1)
\end{split}
\end{align}

if we choose the smallest $\frac1n$, as the price $p$, the expected welfare is:
\begin{small}
\begin{align}
\begin{split}	
	W_2(k_2) = &\Pr[S = 0] \cdot E[B \text{ \& } B > p] + \Pr\left[S = \frac{1}{n}\right] \cdot \left(\frac{1}{n} +\Pr[S < p]\cdot \E[B \text{ \& } B = 1]\right)  + \Pr[S = 1] \cdot 1\\
	=  &(\sqrt2-1)\cdot  \left((2-\sqrt2)\cdot \frac{1}{n}\cdot \frac{k_2}{k_2+1} + \frac{1}{n}(\sqrt2-1))\right) + (2-\sqrt2)\cdot \left(\frac{1}{n} + \frac{1}{k_2+1} \cdot(\sqrt2-1)\frac1n\right) + \frac{1}{n}\cdot (\sqrt2-1)
\end{split}
\end{align}
\end{small}
When all the samples have the same value, any price is optimal. Similar to the case when $N = 1$, the trade may never happen, so the expected welfare is $\frac{1}{n}$ in this case. 

Now, when there are $5$ samples, suppose the ERM will choose the largest $0$ when there are $1 \sim 3$ samples with value $0$, and choose the smallest $\frac1n$ when there are $1$ samples with value $\frac1n$. The expected welfare of ERM when $N = 5$ is :
\begin{small}
\begin{align*}
	\begin{split}
		\sum_{i=1}^3 (\sqrt2-1)^i (2-\sqrt2)^{5-i} \binom{5}{i} \cdot W_1(i) + \sum_{i=1}^1 (\sqrt2-1)^{5-i} (2-\sqrt2)^{i} \binom{5}{i} \cdot W_2(i) + (p^5 + (1-p)^5)\cdot \frac1n
	\end{split}
\end{align*}
\end{small}

Now compare it to the optimal:
\[\OPTW(\I) = \frac{4(\sqrt2-1)}{n} - \frac{4\sqrt2 -1	}{n^2}\]

By numerical calculations, the ratio is $\approx 0.76$ as $n\rightarrow \infty$.

Similarly, when there are $10$ samples, suppose the ERM will choose the largest $0$ when there are $1 \sim 6$ samples with value $0$, and choose the smallest $\frac1n$ when there are $1\sim 3$ samples with value $\frac1n$. The expected welfare of ERM when $N = 10$ is :
\begin{small}
\begin{align*}
	\begin{split}
		\sum_{i=1}^6 (\sqrt2-1)^i (2-\sqrt2)^{10-i} \binom{10}{i} \cdot W_1(i) + \sum_{i=1}^3 (\sqrt2-1)^{10-i} (2-\sqrt2)^{i} \binom{10}{i} \cdot W_2(i) + (p^{10} + (1-p)^{10})\cdot \frac1n
	\end{split}
\end{align*}
\end{small}

By calculations, the approximation ratio is $\approx 0.80$ as $n\rightarrow \infty$.

\subsection{Proof of Lemma~\ref{lem:asymratio}}
\label{appendix:proofasymratio}

Now we fix the quantile mechanism $Q$ and suppose its c.d.f. is $Q(x)$. Let $r$ be \[\min_{x\in [0,1]}{\int_{[0,x]}  t\, dQ(t)  +  1-x}.\]

\cite{blumrosen2021almost} already prove that $r$ is the lower bound the approximation ratio. We are only left to show that it is also the upper bound.	We aim to show that for any $\varepsilon > 0$, there exists an instance $\I = (F_S, F_B)$ such that
	\[\ALGW(\I, Q)\le (r + \varepsilon) \cdot \OPTW(\I)\]
	
	Now suppose the optimization problem above achieves its minimum at $x^{*}$. Consider the following instance $\I = (F_S, F_B)$.
	
\begin{align*}
\begin{split}
    		 S \sim F_{S}, S=\left\{
\begin{aligned}
&U[0,\varepsilon]  \quad & w.p. \quad x^{*} \\
&H + \varepsilon & w.p.\quad 1 - x^{*} 
\end{aligned}
\right.
\quad\quad\quad
B \sim F_{B}, B=
\begin{aligned}
&H\quad & w.p. \quad 1\\
\end{aligned}
\end{split}
\end{align*}
	
	where $H > 1$ is a sufficiently large number. We can see that in this instance, 
	\[\OPTW(\I) \geq H\]
	
	Now we compute the expected welfare for our mechanism. When its price has a quantile $t$ smaller than or equal to $x^*$, the trade will happen with probability exactly $t$. When the quantile of its price is greater than $x^*$, the trade will never happen. 
	\begin{align*}
		\ALGW(\I) &= \E[S] + \E[(B - S)\indic[B\le p\le S]]\\
		& \le  x^* \varepsilon + (H + \varepsilon)(1-x^*) + \int_{[0, x^*]} tH \,dQ(t)\\
		& \le H\cdot\left( 1-x^{*} +\int_{[0, x^*]} t\,dQ(t)\right) + \varepsilon \\
		& \le H\cdot r + \varepsilon \le H\cdot(r + \epsilon)\\
		& \le (r + \varepsilon)\cdot \OPTW(\I)
	\end{align*} where the third inequality follows from $r = \int_{[0, x^*]} t\, dQ(t) + (1 - x^*)$.
	
Therefore, we could find an instance $\I = (F_S, F_B)$ such that $\ALGW(Q, \I) \le (r + \varepsilon) \cdot \OPTW(\I)$ for any small enough $\varepsilon > 0$, and this concludes our proof.

\subsection{Proof of Theorem~\ref{thm:asymoptorderstatistic}}
\label{appendix:proofasymoptorderstatistic}

The proof of Theorem~\ref{thm:asymoptorderstatistic} is nearly identical to Theorem~\ref{thm:symoptorderstatistic}. From Lemma~\ref{lem:asymratio}, we know that the approximation ratio of a quantile mechanism $Q$ with a continuous p.d.f. $q(x)$ is exactly
\[  \mathcal{C}(p)=\min_{x\in [0,1]} {\int_0^x q(t) t\, dt  +  1-x}. \]
where we use $q(x)\, dx$ instead of $dQ(x)$ since we have a continuous probability density function.

Again, we know that the set of all distributions over $[N]$ is actually $\Big\{\{P_i\}_{i\in [n]}~|~P_i\geq 0\text{ and } \sum_{i=1}^n P_i = 1\Big\}$ where distribution $P$ would choose $i$ w.p. $P_i$. What's more, the probabiltiy density function of the quantile of such order statistic mechanism is exactly $p(x) = \sum_{i=1}^n P_i \pdf(x)$ where $\pdf(x) = N \binom{n-1}{i-1} \cdot x ^{i-1}\cdot(1-x)^{N-i}$. Therefore, combining it with Lemma~\ref{lem:smallsamplekey}, it follows that Theorem~\ref{thm:asymoptorderstatistic} holds.
\subsection{Proof of Lemma~\ref{lem:convert}}
\label{appendix:proofconvert}

Before we give the proof, we first introduce some notations and lemmas about Bernstein that may useful to our proof.

\begin{definition}[Stochastic Bernstein Polynomials]
The stochastic Bernstein polynomial of degree $n$ for a continous function $f$ on $[0, 1]$ is defined as

\[\left(B_n^X f\right)(t) = \sum_{k=0}^n f(X_k) p_{n,k}(t)\]

in which $X_0, X_1, \cdots X_n$ are the order statistics of $(n + 1)$ independent copies
of the random variable uniformly distributed in $[0, 1]$, and,

\[p_{n,k}(t) = \binom{n}{k} x^{k} (1-x)^{n-k}, 0\le k\le n, 0\le t\le 1\]

\end{definition}

Now fix the continous function $q(x)$ we aim to approximate, define $\omega(h)$ as the following function:

\[\omega(h) = \sup_{0\le x, y\le 1\atop|x-y|\le h} |{q}(x) - {q}(y)|\]

Now we can introduce the lemma in \cite{DBLP:journals/moc/SunWZ21} that help us approximate the function ${q}(x)$ by order statistics.

\begin{lemma}[Theorem~2.11 In \cite{DBLP:journals/moc/SunWZ21}]
\label{lem:Bernstein}
Let $\varepsilon > 0$ and $f\in C[0, 1]$ be given. Suppose that $\omega\left(\frac{1}{\sqrt{n}}\right) < \varepsilon / 6.2$. Then the following inequality holds true:

\begin{align*}
	\Pr\left[\left\| B_n^X  f - f \right\|_{\infty} > \varepsilon \right] \le 2(n+1)\exp\left(-\frac{2\varepsilon^2}{\omega^2\left(\frac{1}{\sqrt{n}}\right)}\right)
\end{align*}

\end{lemma}

We are now ready to prove Lemma~\ref{lem:convert}.

We first present our mechanism $P$. For some instance $\I = (F_S, F_B)$, suppose there are $n$ samples $X_1\le X_2\le \cdots \le X_n$ drawn from the distribution. We draw another $n$ samples $Y_1 \le Y_2\le \cdots \le Y_n$ uniformly and independently from $[0, 1]$. Let $s = \sum_{i=1}^n q(Y_i)$ be the sum. Then our mechanism will choose $X_i$ with probability ${q}(Y_i) / s$

Now, let 
\[g(x) = \sum_{i=1}^n q(Y_i)f_n^i(x)/s \quad \text{where}~ f_n^i(x) = n\binom{n - 1}{i-1}\cdot x^{i-1}(1-x)^{n-i}\]

be the corresponding probability density function of the order statistic mechanism $P$. It suffices to prove that with high probability \[|g(x) - q(x)| \leq \varepsilon \quad \forall x\in [0, 1]\]

To prove this, we introduce an intermediate function $h(x)$:

\[h(x) = \sum_{i=1}^n q(Y_i) \cdot \binom{n-1}{i-1}x^{i-1}(1-x)^{N-i} = \sum_{i=1}^n {q}(Y_i) p_{n-1,i-1}(x)\]

As we can see, $h$ is the stochastic Bernstein polynomial of ${q}$ with degree $n - 1$ and $\omega\left(\frac{1}{\sqrt{n - 1}}\right) \leq \varepsilon /100$. Applying lemma~\ref{lem:Bernstein}, we know that 
\begin{equation}\label{eq:hqdis}
\Pr\left[\left\|h - q\right\|_{\infty} > \varepsilon / 4\right] \le 2n \exp\left(-\frac{\varepsilon^2}{8 \cdot \omega^2\left(\frac{1}{\sqrt{n - 1}}\right)} \right)\le \varepsilon 
\end{equation} where the last inequality comes from the assumption in the statement of lemma.

Thus, we only need to show that the difference between $h$ and $g$ is small. First we have
\begin{align*}
    \begin{split}
        g(x) &=  \sum_{i=1}^N \frac{{q}(Y_i)}{S} f_n^i(x) \\
            & = \sum_{i=1}^N {q}(Y_i) \binom{n-1}{i-1}x^{i-1}(1-x)^{n-i} \frac{n}{s}\\
            & = \frac{n}{s} \cdot h(x)
    \end{split}
\end{align*}

So it is equivalent to prove that $n$ and $s$ are close w.h.p. We have the following lemma.

\begin{lemma}\label{lem:Bernsteinhelp}
\begin{align*}
\begin{split}
\Pr\left[\left(1 - \frac{\varepsilon}{4M}\right)\cdot n\le s \le  \left(1 + \frac{\varepsilon}{4M}\right) \cdot n\right] \ge 1 - \varepsilon
\end{split}
\end{align*}
\end{lemma}

We first use the lemma to continue our proof, before proving the lemma itself. We know that $|g(x) - h(x)|\leq \varepsilon/2 ~ \forall x\in [0, 1]$ if $n \left(1-\frac{\varepsilon}{4M}\right)\leq s \le n \left(1 + \frac{\varepsilon} {4M}\right)$ and $h(x) < 2M ~ \forall x\in [0, 1]$. Therefore, 
\begin{small}
\begin{align}\label{eq:ghdis}
\begin{split}
   \Pr[|g(x) - h(x)| \geq \varepsilon/2 ~\exists x\in [0, 1]] &\le \Pr\left[s > n \left(1 + \frac{\varepsilon}{4M}\right)\right] + \Pr\left[s < n \left(1 - \frac{\varepsilon}{4M}\right)\right] +\Pr[\exists_{x\in [0, 1]} h(x) > 2M]\\
   & \le 2\varepsilon
\end{split}
\end{align}
\end{small}

where the second inequality is from Lemma~\ref{lem:Bernsteinhelp} and the fact that with probability at least $1 - \varepsilon$, $\|h - {q}\| \le \varepsilon / 4$ and ${q}$ has a maximum of $M$ on $[0, 1]$.

Now combining inequality~\eqref{eq:hqdis} and \eqref{eq:ghdis}, we know that with probability at least $1 - 3\varepsilon$, we have:
\[|g(x) - {q}(x)|\le |g(x) - h(x)| + |h(x) - {q}(x)| \leq \varepsilon\]

Since such probability is strictly greater than $0$, we know that there exists some order statistic mechanism $P$ over $N$ samples such that $|g(x)  - q(x) | \leq \varepsilon  ~ \forall x\in [0, 1]$. Also we show that our construction will find such order statistics with high probability. Finally, as we assumed in the statement, it holds that
\[\mathcal{C}(P) \geq \mathcal{C}(Q) - c\cdot |g - q|_{\infty} \geq \mathcal{C}(Q) - c\varepsilon\]

This completes the proof of Lemma~\ref{lem:convert}.

\begin{proof}[Proof of Lemma~\ref{lem:Bernsteinhelp}]
    We know that $\E[s] = n$. Notice that $s$ is the sum of $n$ i.i.d. random variables ranging in $[0, M]$. Therefore, by Chernoff bound, it holds that
    \begin{align*}
        \begin{split}
            \Pr\left[\left(1 - \frac{\varepsilon}{4M}\right)\cdot n\le s \le  \left(1 + \frac{\varepsilon}{4M}\right) \cdot n\right] &\geq 1-  \exp\left(-\frac{\varepsilon^2 n}{48M^2}\right) -  \exp\left(-\frac{\varepsilon^2 n}{32M^2}\right)\\
            & \geq 1 - \varepsilon
        \end{split}
    \end{align*}where the last inequality is from the property in the statement.
\end{proof}

\subsection{Proof of Lemma~\ref{lem:symturning}}
\label{appendix:proofofsymturning}

let $Q'$ be the quantile mechanism with p.d.f. $\tilde{q}(x)$ satisfying

\begin{align*}
\begin{split}
\tilde{q}(x) = \left\{
\begin{aligned}
&0 \quad & x\in \left[0, \frac{\sqrt{2}}{2} - \frac{1}{32}\varepsilon\right]\bigcup \left[\frac{\sqrt{2}}{2}+\frac{1}{32}\varepsilon, 1\right]\\
&(32/\varepsilon)^2 \cdot \left(x - \frac{\sqrt2}{2} + \frac{1}{32}\varepsilon\right)  & x\in \left[\frac{\sqrt2}{2} - \frac{1}{32}\varepsilon, \frac{\sqrt2}{2}\right]\\
&-(32/\varepsilon)^2\cdot \left(x - \frac{\sqrt2}{2} - \frac{1}{32}\varepsilon\right)   & x\in \left[\frac{\sqrt2}{2}, \frac{\sqrt2}{2} + \frac{1}{32}\varepsilon\right]\\
\end{aligned}
\right.
\end{split}
\end{align*}

One could see that the quantile mechanism $\tilde{Q}$ would choose a price with its quantile in $\left[\frac{\sqrt2}{2} - \frac{1}{32}\varepsilon, \frac{\sqrt2}{2} + \frac{1}{32}\varepsilon\right]$. 

Fix a price $p$ and an instance $\I = (F, F)$, $\alg$ is the random variable indicating the welfare for the fixed price $p$ in the realzation, i.e. $S + (B - S)\indic[B\ge p\ge S]$ where $B,S$ are drawn independently from $F$. Thus, it suffices to show that for any instance $\I = (F, F)$, the following holds when $\frac{\sqrt2}2 - \frac{1}{32}\varepsilon \le t\leq \frac{\sqrt2}2 + \frac{1}{32}\varepsilon$.
 \[	\E\left[\alg\given p = F^{-1}(t)\right] \ge \left(\frac{2+\sqrt2}{4} - \varepsilon / 2\right) \cdot \OPTW(\I)\]

To prove the approximation ratio, we need to use Lemma~\ref{lem:symratio}. Notice that the term $\E[\ALG\given p = F^{-1}(t)]$ could be understood as the expected welfare of a quantile mechanism that always selects the $t$-quantile as the price, so Lemma~\ref{lem:symratio} could also be applied to analyze the ratio between ${\E\left[\alg\given p = F^{-1}(t)\right]}$ and ${\OPTW(\I)}$. 
For $x\ge t$, we know that
\begin{align*}
\begin{split}
\inf_{x\in [t, 1)}\frac{(1-x)+t(1-x)}{1-x^2} =	\frac{1 + t}{2}
\end{split}
\end{align*}
	
When $x<t$, it holds that
\begin{align*}
\begin{split}
\min_{x\in [0, t]}\frac{(1-x)+x(1-t)}{1-x^2} = \frac{1 + \sqrt{1 -t^2}}{2}
\end{split}
\end{align*}

As we can see, $\frac{1 + t}{2}\ge \frac{1 + \sqrt{1 -t^2}}{2}$ when $1 \ge t\ge \frac{\sqrt2}2$ and $\frac{1 + t}{2}\le \frac{1 + \sqrt{1 -t^2}}{2}$ when $\frac{\sqrt{2}}{2}\ge t\ge 0$. By Lemma~\ref{lem:symratio}, we know that for any instance $\I = (F, F)$,

\begin{align*}
\begin{split}
 \E\left[\alg\given p = F^{-1}(t)\right] \ge \left\{
\begin{aligned}
&\frac{1 + \sqrt{1 -t^2}}{2} \cdot \OPTW(\I)  \quad & t\in \left[\frac{\sqrt2}2, 1\right] \\
&\frac{1 + t}{2} \cdot \OPTW(\I)  & t\in \left[0, \frac{\sqrt2}2\right]
\end{aligned}
\right.
\end{split}
\end{align*}

Now if $\frac{\sqrt2}2 - \frac{1}{32}\varepsilon \le t \le \frac{\sqrt2}2$, it holds that

\begin{align*}
\begin{split}
\E\left[\alg\given p = F^{-1}(t)\right] &\ge\frac{1 + t}{2} \cdot \OPTW(\I) \\
	& \ge \left(\frac{2 + \sqrt2}{4} - \varepsilon / 2\right) \cdot \OPTW(\I)
\end{split}
\end{align*}.

For $\frac{\sqrt2}2 \le t\le \frac{\sqrt2}2 + \frac{1}{32}\varepsilon$, we have
\begin{align*}
\begin{split}
\E\left[\alg\given p = F^{-1}(t)\right] &\ge\frac{1 + \sqrt{1 -t^2}}{2} \cdot \OPTW(\I) \\
& \ge \left(\frac{2 + \sqrt2}{4} - \varepsilon / 2\right) \cdot \OPTW(\I).
\end{split}
\end{align*}

Finally, by some simple calculation, it is easy to see that  $\omega\left(c \cdot \varepsilon^{3.5}\right) \leq 32^2 c\cdot \varepsilon^{1.5}$ and  $\max_{x\in[0,1]} \tilde{q}(x) \leq 32/\varepsilon$ hold. This concludes the proof.

\subsection{Proof of Lemma~\ref{lem:symclose}}
\label{appendix:proofsymclose}

Suppose $|q_1(x) - q_2(x)| \leq c$ for all $x\in [0, 1]$. We could see that
\begin{align*}
 &\frac{\int_0^x (q_1(t) - q_2(t))t(1-x) \, d t + \int_x^1 (q_1(t) - q_2(t)) (1-t)x\, dt}{1-x^2}\\
 & \geq  - \frac{\int_0^x c t(1-x) \, d t + \int_x^1 c (1-t)x\, dt }{1-x^2}\\
 & \geq - \frac{\int_0^x c t(1-x) \, d t + \int_x^1 c (1-x)x\, dt }{1-x^2}\\
 & = - \frac{\int_0^x c t \, d t + \int_x^1 c x\, dt }{1+x} \geq -c.
\end{align*} holds for all $x\in [0, 1]$.

Taking the infimum over $[0, 1)$, this directly implies that $\mathcal{C}(Q_1) - \mathcal{C}(Q_2) \geq - c$.

\subsection{Proof of Lemma~\ref{lem:asymturning}}
\label{appendix:proofasymturning}

let $\tilde{Q}$ be the quantile mechanism with p.d.f. $\tilde{q}(x)$ satisfying
\begin{align*}
\begin{split}
\tilde{q}(x) = \left\{
\begin{aligned}
&0 \quad & x\in \left[0, \frac{1}{e} - \frac{\left(1 - e^{-1}\right)}{\left(e-\frac12\varepsilon\right)}\varepsilon\right]\\
    &\frac{(e-\frac12\varepsilon)^2}{\varepsilon (1-e^{-1})} \left(x - \left(\frac{1}{e} - \frac{\left(1 - e^{-1}\right)}{\left(e-\frac12\varepsilon\right)}\varepsilon\right)\right)   & x\in \left[\frac{1}{e} - \frac{\left(1 - e^{-1}\right)}{\left(e-\frac12\varepsilon\right)}\varepsilon, \frac{1}{e}\right]\\
&\frac{1}{x} - \frac{1}{2}\varepsilon   & x\in \left[\frac{1}{e}, 1\right]\\
\end{aligned}
\right.
\end{split}
\end{align*}

By Lemma~\ref{lem:asymratio}, the approximation ratio of $\tilde{Q}$ can be computed as
\begin{align*}
\begin{split}
	\min_{x\in [0,1]}\int_0^x \tilde{q}(t)\cdot t\mathop{dt}  +  1-x &\ge \min_{x\in [0,1]}\int_{\frac{1}{e}}^x \left(\frac{1}{t} - \frac{1}{2}\varepsilon\right)\cdot t\mathop{dt}  +  1-x\\
	&\geq \min_{x\in [0,1]}\int_{\frac{1}{e}}^x \frac{1}{t} \cdot t\mathop{dt}  +  1-x - \frac{1}{2}\varepsilon\\
	& = 1 - \frac{1}{e} - \frac{1}{2} \varepsilon.
\end{split}
\end{align*}

Finally, it is also straightforward to check that $\omega\left(c\cdot \varepsilon^{-2.5}\right) \leq e^2(1-e^{-1})^{-1}c^{0.5}\cdot\varepsilon^{1.5} \leq 100c^{0.5} \cdot \varepsilon^{1.5}$ and $\max_{x\in[0,1]} \tilde{q}(x)$ is at most $2e$.

\subsection{Proof of Lemma~\ref{lem:asymclose}}
\label{appendix:proofasymclose}

Suppose $|q_1(x) - q_2(x)| \leq c$ for all $x\in [0, 1]$. We could see that
\begin{align*}
 &\left(\int_0^x q_1(t)t\,dt + (1 - x) \right)- \left(\int_{0}^x q_2(t)t\, dt + (1 - x)\right) \\
 & \geq  - \int_{0}^x ct\, dt  \geq -c.
\end{align*} holds for all $x\in [0, 1]$.

Taking the minimum over $[0, 1)$, this directly implies that $\mathcal{C}(Q_1) - \mathcal{C}(Q_2) \geq - c$.

\section{Details of Numerical Experiments}
\label{appendix:experiments}

\subsection{Symmetric Instance}
\label{appendix:symmetricexperiment}

We now present the details of numerical experiments for the symmetric instances $\I = (F, F)$. We first formally write down the optimization problem that indicates the optimal order statistic mechanisms. As we proved in Section~\ref{subsec:symsmallsample}, suppose the following optimization problem $\mathsf{PO}$ achieves its maximum $\opt_{\mathsf{PO}}$ at $(P_1^*, P_2^*, \dots, P_N^*)$. Let $P^*$ be the distribution over $[N]$ such that $\Pr_{x\sim P^*}[x = i] = P_i^*$. Then, $P^*$ is the optimal order statistic mechanism in the symmetric setting with $N$ samples and its approximation ratio is $\opt_{\mathsf{PO}}$. Notice that since the c.d.f. of the order statistic mechanism $Q(x)$ is differentiable, we use $q(t)\, dt$ instead of $dQ(t)$ for ease of computation. Here $q(x)$ is the probality density function of mechanism $Q$.

\begin{center}
\begin{tabular}{|c|}
\hline 
The Optimization Problem $\mathsf{PO}$ \\
\hline
\parbox{13cm}{
\begin{align}
\max_{P_1,\dots, P_N} \quad
 \min_{x\in [0,1]} \quad &\frac{\int_0^x q(t)\cdot t(1-x)\mathop{dt} + \int_x^1 q(t)\cdot (1-t)x\mathop{dt} + (1-x)}{1-x^2}\nonumber \\
\textsf{s.t.} \quad  &
q(x) = \sum_{i=1}^N P_i \cdot \pdf(x)\nonumber\\
& \pdf(x) = N\cdot \binom{N-1}{i-1}\cdot x^{i-1} (1-x)^{n-i} \quad \forall i\in [N] \nonumber\\
&\sum_{i=1}^n P_i = 1 \label{eq:PiSum1}\\
&P_i \geq 0 \quad \forall i\in [N] \label{eq:PiNonZero}
\end{align}
} \\
\hline 
\end{tabular}
\end{center}

Now we aim to solve the optimization problem $\mathsf{PO}$ numerally with different numbers of samples $N$. For the inner minimization problem, it must be solved accurately so that it precisely reflect the approximation ratio of the order statistics mechanism. We use binary search to find its optimum. When we need to check whether
\[\inf_{x\in [0,1)} \frac{\int_0^x q(t)\cdot t(1-x)\mathop{dt} + \int_x^1 q(t)\cdot (1-t)x\mathop{dt} + (1-x)}{1-x^2} \ge r\]

It is equivalent to check if 
\[{\int_0^x q(t)\cdot t(1-x)\mathop{dt} + \int_x^1 q(t)\cdot (1-t)x\mathop{dt} + (1-x)} - r(1-x^2) \ge 0 \quad \forall x \in [0, 1]\]

Notice that $q(t)$ is a polynomial of degree $N$, thus we only need to find the minimum of a single-variable polynomial over $[0, 1]$, and this could be efficiently done by finding the roots of its derivatives. We do the binary search for $100$ times, so the error caused by binary search is at most $2^{-100}$, which is much smaller than the floating point errors and can be ignored. Then we use some empirical algorithms to search for parameters in the outer maximization problem. Unlike the inner minimization problem, we do not need to get an exact optimum of the outer maximization problem since it only reflects the ratio of the order statistic mechanism we have found. The code can be found at our \href{https://github.com/BilateralTradeAnonymous/On-the-Optimal-Fixed-Price-Mechanism-in-Bilateral-Trade}{GitHub repository}\begin{footnotesize}(\texttt{https:\\//github.com/BilateralTradeAnonymous/On-the-Optimal-Fixed-Price-Mechanism-in-Bilateral-Trade})\end{footnotesize}.

\subsection{General Instance}
\label{appendix:generalexperiment}

Now again we formally write down the optimization problem that characterizes the optimal order statistic mechanism with $N$ samples. Similarly, as we proved in Section~\ref{subsec:generalsmallsample}, suppose the following optimization problem $\mathsf{QO}$ achieves its maximum $\opt_{\mathsf{QO}}$ at $(P_1^{*}, P_2^{*}, \dots, P_n^{*})$. Let $P^*$ be the distribution over $[N]$ such that $\Pr_{x\sim P^*}[x = i] = P_i^*$. Then, $P^*$ is the optimal order statistic mechanism in the general setting with $N$ samples and its approximation ratio  is exactly $\opt_{\mathsf{QO}}$. Again notice that since the c.d.f. of the order statistic mechanism $Q(x)$ is differentiable, we use $q(t)\, dt$ instead of $dQ(t)$ where $q(x)$ is the p.d.f. of the order statistic mechanism.

\begin{center}
\begin{tabular}{|c|}
\hline 
The Optimization Problem $\mathsf{QO}$ \\
\hline
\parbox{13cm}{
\begin{align}
\max_{P_1,\dots, P_N} \quad &{\min_{x\in [0,1]}{\int_0^x q(t)\cdot t\mathop{dt} +  1-x}}\nonumber \\
\textsf{s.t.} \quad  &
q(x) = \sum_{i=1}^N P_i \cdot \pdf(x)\nonumber\\
& \pdf(x) = N\cdot \binom{N-1}{i-1}\cdot x^{i-1} (1-x)^{n-i} \quad \forall i\in [N] \nonumber\\
&\sum_{i=1}^n P_i = 1 \label{eq:APiSum1}\\
&P_i \geq 0 \quad \forall i\in [N] \label{eq:APiNonZero}
\end{align}
} \\
\hline 
\end{tabular}
\end{center}

Now we solve the optimization problem numerically for different fixed number $N$. Again the inner minimization problem could be solved efficiently by calculating the zero point of its derivatives and we search for the parameters in the outer maximization problem to get a good enough solution. The code could be found at our \href{https://github.com/BilateralTradeAnonymous/On-the-Optimal-Fixed-Price-Mechanism-in-Bilateral-Trade}{Github repository}.

\end{document}